\documentclass[11pt]{article}

\usepackage{latexsym}
\usepackage{amsmath}
\usepackage{amssymb}
\usepackage{amsfonts,amsthm}
\usepackage{enumitem}
\usepackage[top=1in, bottom=1in, left=1in, right=1in]{geometry}
\usepackage{stackrel}

\newcommand{\ignore}[1]{}

\newtheorem{theorem}{Theorem}[section]
\newtheorem{lemma}[theorem]{Lemma}

\newtheorem{corollary}[theorem]{Corollary}

\newtheorem{fact}[theorem]{Fact}

\newtheorem{proposition}[theorem]{Proposition}

\usepackage{thmtools}
\usepackage{thm-restate}

 \usepackage{url}
 \usepackage{subcaption}
\usepackage{crop}
\usepackage{enumerate}
\usepackage{etoolbox}
\usepackage{float}
\usepackage{hyperref}
\usepackage{latexsym}
\usepackage{mathtools}
\usepackage{relsize}
\usepackage{amsmath}
\usepackage{amssymb}
\usepackage{amsfonts}
\usepackage{balance}
\usepackage{tabu}
\usepackage{longtable}

\usepackage{multirow}
\usepackage{lscape}
\usepackage[lined, algoruled, linesnumbered]{algorithm2e}
\usepackage{lipsum}

\let\oldReturn\Return
\renewcommand{\Return}{\State\oldReturn}

\usepackage{abstract}
\usepackage{hyperref,xcolor}
\definecolor{winered}{rgb}{0.5,0,0}

\usepackage{enumitem}

\hypersetup
{
    pdfborder={0 0 0},
    colorlinks=true,
    linkcolor={winered},
    urlcolor={winered},
    filecolor={winered},
    citecolor={winered},
    linktoc=all,
}

\setlist[description]{leftmargin=\parindent,labelindent=\parindent}

\begin{document}
\title{Computing the $4$-Edge-Connected Components of a Graph in Linear Time\thanks{Research at the University of Ioannina supported by the Hellenic Foundation for Research and Innovation (H.F.R.I.) under the ``First Call for H.F.R.I. Research Projects to support Faculty members and Researchers and the procurement of high-cost research equipment grant'', Project FANTA (eFficient Algorithms for NeTwork Analysis), number HFRI-FM17-431. G. F. Italiano is partially supported by MIUR, the Italian Ministry for Education, University and Research, under PRIN Project AHeAD (Efficient Algorithms for HArnessing Networked Data)}}
\author{Loukas Georgiadis$^{1}$ \and Giuseppe F. Italiano$^{2}$ \and Evangelos Kosinas$^{3}$}

\maketitle

\begin{abstract}
We present the first linear-time algorithm that computes the $4$-edge-connected components of an undirected graph.
Hence, we also obtain the first linear-time algorithm for testing $4$-edge connectivity.
Our results are based on a linear-time algorithm that computes the $3$-edge cuts of a $3$-edge-connected graph $G$, and a linear-time procedure that, given the collection of all $3$-edge cuts, partitions the vertices of $G$ into the $4$-edge-connected components.
\end{abstract}

\footnotetext[1]{Department of Computer Science \& Engineering, University of Ioannina, Greece. E-mail: loukas@cs.uoi.gr}
\footnotetext[2]{LUISS University, Rome, Italy. E-mail: gitaliano@luiss.it}
\footnotetext[3]{Department of Computer Science \& Engineering, University of Ioannina, Greece. E-mail: ekosinas@cs.uoi.gr}

\section{Introduction}
\label{sec:introduction}

Let $G=(V,E)$ be a connected undirected graph with $m$ edges and $n$ vertices.
An \emph{(edge) cut} of $G$ is a set of edges $S \subseteq E$ such that $G \setminus S$ is not connected.
We say that $S$ is a \emph{$k$-cut} if its cardinality is $|S|=k$. Also, we refer to the $1$-cuts as the \emph{bridges} of $G$.
A cut $S$ is \emph{minimal} if no proper subset of $S$ is a cut of $G$.
The \emph{edge connectivity} of $G$, denoted by $\lambda(G)$, is the minimum cardinality of an edge cut of $G$.
A graph is \emph{$k$-edge-connected} if $\lambda(G) \ge k$.

A cut $S$ separates two vertices $u$ and $v$, if $u$ and $v$ lie in different connected components of $G \setminus S$.
Vertices $u$ and $v$ are $k$-edge-connected, denoted by $u \stackrel[]{G}{\equiv}_k v$, if there is no $(k-1)$-cut that separates them. By Menger's theorem~\cite{menger}, $u$ and $v$ are $k$-edge-connected if and only if there are $k$-edge-disjoint paths between $u$ and $v$.
A \emph{$k$-edge-connected component} of $G$ is a maximal set $C \subseteq V$ such that there is no $(k-1)$-edge cut in $G$ that disconnects any two vertices $u,v \in C$ (i.e., $u$ and $v$ are in the same connected component of $G \setminus S$ for any $(k-1)$-edge cut $S$).
We can define, analogously, the \emph{vertex cuts} and the \emph{$k$-vertex-connected components} of $G$.

Computing and testing the edge connectivity of a graph, as well as its $k$-edge-connected components, is a classical subject in graph theory, as it is an important notion in several application areas (see, e.g., \cite{connectivity:nagamochi-ibaraki}), that has been extensively studied since the 1970's.
It is known how to compute the $(k-1)$-edge cuts, $(k-1)$-vertex cuts,  $k$-edge-connected components and $k$-vertex-connected components of a graph in linear time for $k \in \{2,3\}$~\cite{GI:ECtoVC,3-connectivity:ht,NagamochiIbaraki:3CC,dfs:t,Tsin:3CC}.
The case $k=4$ has also received significant attention \cite{3cuts:Dinitz,4C:Online,KR:4C,Kanevsky:4CC}. Unfortunately, none of the previous algorithms achieved linear running time.
In particular,
Kanevsky and Ramachandran~\cite{KR:4C} showed
how to test whether a graph is $4$-vertex-connected in
$O(n^2)$ time. Furthermore,
Kanevsky et al.~\cite{Kanevsky:4CC} gave an $O(m+n\alpha(m,n))$-time algorithm to compute the $4$-vertex-connected components of a $3$-vertex-connected graph, where $\alpha$ is a functional inverse of Ackermann's function~\cite{dsu:tarjan}.
Using the reduction of Galil and Italiano~\cite{GI:ECtoVC} from edge connectivity to vertex connectivity, the same bounds can be obtained for $4$-edge connectivity. Specifically, one can test whether a graph is $4$-edge-connected in $O(n^2)$ time, and one can
compute the $4$-edge-connected components of a $3$-edge-connected graph in $O(m+n\alpha(m,n))$ time.
Dinitz and Westbrook~\cite{4C:Online} presented an $O(m+n\log{n})$-time algorithm to compute the $4$-edge-connected components of a general graph $G$ (i.e., when $G$ is not necessarily $3$-edge-connected).
Nagamochi and Watanabe~\cite{ni93} gave an $O(m+k^2n^2)$-time algorithm to compute the $k$-edge-connected components of a graph $G$, for any integer $k$.
We also note that the edge connectivity of a simple undirected graph can be computed in ${O}(m \mathrm{polylog}{n})$ time, randomized~\cite{FasterEC,Karger:MinCut} or deterministic~\cite{LocalFlow:EC,KT:EC}.
The best current bound is $O(m \log^2{n} \log{\log}^2{n})$, achieved by Henzinger et al.~\cite{LocalFlow:EC} which provided an improved version of the algorithm of Kawarabayashi and Thorup~\cite{KT:EC}.

\paragraph*{Our results and techniques}
In this paper we present the first linear-time algorithm that computes the $4$-edge-connected components of a general graph $G$, thus resolving a problem that remained open for more than 20 years.
Hence, this also implies the first linear-time algorithm for testing $4$-edge connectivity.
We base our results on the following ideas. First, we extend the framework of Georgiadis and Kosinas~\cite{TwinlessSAP} for computing $2$-edge cuts (as well as mixed cuts consisting of a single vertex and a single edge) of $G$.
Similar to known linear-time algorithms for computing $3$-vertex-connected and $3$-edge-connected components~\cite{3-connectivity:ht,Tsin:3CC}, Georgiadis and Kosinas~\cite{TwinlessSAP} define various concepts with respect to a depth-first search (DFS) spanning tree of $G$. 
We extend this framework by introducing new key parameters that can be computed efficiently and provide characterizations of the various types of $3$-edge cuts that may appear in a $3$-edge-connected graph.
We deal with the general case by dividing $G$ into auxiliary graphs $H_1, \ldots, H_{\ell}$, such that each $H_i$ is $3$-edge-connected and corresponds to a different $3$-edge-connected component of $G$. Also,  for any two vertices $x$ and $y$, we have $x \stackrel[]{G}{\equiv}_4 y$ if and only if $x$ and $y$ are both in the same auxiliary graph $H_i$ and $x \stackrel[]{H_i}{\equiv}_4 y$.
Furthermore, this reduction allows us to compute in linear time the number of \emph{minimal $3$-edge cuts} in a general graph $G$.
Next, in order to compute the $4$-edge-connected components in each auxiliary graph $H_i$, we utilize the fact that a minimum cut of a graph $G$ separates $G$ into two connected components.
Hence, we can define the set $V_C$ of the vertices in the connected component of $G\setminus{C}$ that does not contain a specified root vertex $r$. We refer to the number of vertices in $V_C$ as the \emph{$r$-size} of the cut $C$.
Then, we apply a recursive algorithm that successively splits $H_i$ into smaller graphs according to its $3$-cuts.
When no more splits are possible, the connected components of the final split graph correspond to the $4$-edge-connected components of $G$.
We show that we can implement this procedure in linear time by processing the cuts in non-decreasing order with respect to their $r$-size.

\section{Concepts defined on a DFS-tree structure}
\label{sec:preliminaries}

Let $G=(V,E)$ be a connected undirected graph, which may have multiple edges.
For a set of vertices $S \subseteq V$, the induced subgraph of $S$, denoted by $G[S]$, is the subgraph of $G$ with vertex set $S$ and edge set $\{e\in E \mid $ both ends of $e$ lie in $S\}$.
Let $T$ be the spanning tree of $G$ provided by a depth-first search (DFS) of $G$~\cite{dfs:t}, with start vertex $r$.
The edges in $T$ are called tree-edges; the edges in $E \setminus T$ are called back-edges, as their endpoints have ancestor-descendant relation in $T$.
A vertex $u$ is an ancestor of a vertex $v$ ($v$ is a descendant of $u$) if the tree path from $r$ to $v$ contains $u$.
Thus, we consider a vertex to be an ancestor (and, consequently, a descendant) of itself.
We let $p(v)$ denote the parent of a vertex $v$ in $T$.
If $u$ is a descendant of $v$ in $T$, we denote the set of vertices of the simple tree path from $u$ to $v$ as $\mathit{T}[u,v]$. The expressions $\mathit{T}[u,v)$ and $\mathit{T}(u,v]$ have the obvious meaning (i.e., the vertex on the side of the parenthesis is excluded).
%
%
From now on, we identify vertices with their preorder number (assigned during the DFS).
Thus, $v$ being an ancestor of $u$ in $T$ implies that $v\leq u$. Let $T(v)$ denote the set of descendants of $v$, and let $\mathit{ND}(v)$ denote the number of descendants of $v$ (i.e. $\mathit{ND}(v)=|T(v)|$). With all $\mathit{ND}(v)$ computed, we can check in constant time whether a vertex $u$ is a descendant of $v$, since $u\in T(v)$ if and only if $v\leq u$ and $u<v+\mathit{ND}(v)$~\cite{domin:tarjan}.

Whenever $(x,y)$ denotes a back-edge, we shall assume that $x$ is a descendant of $y$. We let $B(v)$ denote the set of back-edges $(x,y)$, where $x$ is a descendant of $v$ and $y$ is a proper ancestor of $v$. Thus, if we remove the tree-edge $(v,p(v))$, $T(v)$ remains connected to the rest of the graph through the back-edges in $B(v)$.
This implies that $G$ is $2$-edge-connected if and only if $|B(v)|>0$, for every $v\neq r$. Furthermore, $G$ is $3$-edge-connected only if $|B(v)|>1$, for every $v\neq r$. We let $\mathit{b\_count}(v)$ denote the number of elements of $B(v)$ (i.e. $\mathit{b\_count}(v)=|B(v)|$). $\mathit{low}(v)$ denotes the lowest $y$ such that there exists a back-edge $(x,y)\in B(v)$. Similarly, $\mathit{high}(v)$ is the highest $y$ such that there exists a back-edge $(x,y)\in B(v)$.

We let $M(v)$ denote the nearest common ancestor of all $x$ for which there exists a back-edge $(x,y)\in B(v)$. Note that $M(v)$ is a descendant of $v$. Let $m$ be a vertex and $v_1,\dotsc,v_k$ be all the vertices with $M(v_1)=\dotsc=M(v_k)=m$, sorted in decreasing order. (Observe that $v_{i+1}$ is an ancestor of $v_i$, for every $i\in\{1,\dotsc,k-1\}$, since $m$ is a common descendant of all $v_1,\dotsc,v_k$.) Then we have $M^{-1}(m)=\{v_1,\dotsc,v_k\}$, and we define $\mathit{nextM}(v_i)$ $:=$ $v_{i+1}$, for every $i\in\{1,\dotsc,k-1\}$, and $\mathit{lastM}(v_i)$ $:=$ $v_k$, for every $i\in\{1,\dotsc,k\}$. Thus, for every vertex $v$, $\mathit{nextM}(v)$ is the successor of $v$ in the decreasingly sorted list $M^{-1}(M(v))$, and $\mathit{lastM}(v)$ is the lowest element in $M^{-1}(M(v))$.


The following two simple facts have been proved in \cite{TwinlessSAP}. 

\begin{fact}
All $\mathit{ND}(v)$, $\mathit{b\_count}(v)$, $M(v)$, $\mathit{low}(v)$ and $\mathit{high}(v)$ can be computed in total linear-time, for all vertices $v$.
\end{fact}

\begin{fact}
$B(u)=B(v)$ $\Leftrightarrow$ $M(u)=M(v)$, and $\mathit{high}(u)=\mathit{high}(v)$  $\Leftrightarrow$ $M(u)=M(v)$ and $\mathit{b\_count}(u)=\mathit{b\_count}(v)$.
\end{fact}

Furthermore, \cite{TwinlessSAP} implies the following characterization of a $3$-edge-connected graph.

\begin{fact}
\label{fact:3edgeconn}
$G$ is $3$-edge-connected if and only if $|B(v)|>1$, for every $v\neq r$, and $B(v)\neq B(u)$, for every pair of vertices $u$ and $v$, $u \not= v$.
\end{fact}

\begin{lemma}
\label{lemma:Muv}
Let $v$ be an ancestor of $u$ and $M(v)$ a descendant of $u$. Then, $M(v)$ is a descendant of $M(u)$.
\end{lemma}
\begin{proof}
Let $(x,y)\in B(v)$. Then $x$ is a descendant of $M(v)$, and therefore a descendant of $u$. Furthermore, $y$ is a proper ancestor of $v$, and therefore a proper ancestor of $u$. This shows that $(x,y)\in B(u)$, and thus we have $B(v)\subseteq B(u)$. This shows that $M(v)$ is a descendant of $M(u)$.
\end{proof}

The following lemma will be implicitly evoked several times in the following sections.

\begin{lemma}
Let $u$ be a proper descendant of $v$ such that $M(u)=M(v)$. Then, $B(v)\subseteq B(u)$. Furthermore, if the graph is $3$-edge-connected, $B(v)\subset B(u)$.
\end{lemma}
\begin{proof}
Let $(x,y)\in B(v)$. Then $x$ is a descendant of $M(v)$, and therefore a descendant of $M(u)$, and therefore a descendant of $u$. Furthermore, $y$ is a proper ancestor of $v$, and therefore a proper ancestor of $u$. This shows that $(x,y)\in B(u)$, and thus $B(v)\subseteq B(u)$ is established. If the graph is $3$-edge-connected, $B(v)\subset B(u)$ is an immediate consequence of fact \ref{fact:3edgeconn}.
\end{proof}

Now let us provide some extensions of those concepts that will be needed for our purposes. Assume that $G$ is $3$-edge-connected, and let $v\neq r$ be a vertex of $G$. By fact \ref{fact:3edgeconn}, $\mathit{b\_count}(v)>1$, and therefore there are at least two back-edges in $B(v)$. Of course, there is at least one back-edge $(x,y)\in B(v)$ such that $y=\mathit{low}(v)$. We let $\mathit{low1}(v)$ denote $y$, and $\mathit{low1D}(v)$ denote $x$. That is, $\mathit{low1}(v)$ is the $\mathit{low}$ point of $v$, and $\mathit{low1D}(v)$ is a descendant of $v$ which is connected with a back-edge to its $\mathit{low}$ point. (Of course, $\mathit{low1D}(v)$ is not uniquely determined, but we need to have at least one such descendant stored in a variable.) Similarly, we let $\mathit{highD}(v)$ denote a descendant of $v$ which is connected with a back-edge to the $\mathit{high}$ point of $v$. (Again, $\mathit{highD}(v)$ is not uniquely determined.) Then, there may exist another back-edge $(x',y')\in B(v)$ with $x'\neq x$ and $y'=y$. In this case, we let $\mathit{low2}(v)$ denote $y'$ (that is, $\mathit{low2}(v)$ is, again, the $\mathit{low}$ point of $v$) and $\mathit{low2D}(v)$ denote $x'$. If there is no back-edge $(x',y')\in B(v)$ with $x'\neq x$ and $y'=y$, let $(x',y')\in B(v)$ denote a back-edge with $y'=\mathit{min}(\{w\mid\exists (z,w)\in B(v)\}\setminus\{y\})$. Then we let $\mathit{low2}(v)$ denote $y'$ and $\mathit{low2D}(v)$ denote $x'$. Thus, if $v\neq r$, we know that $(\mathit{low1D}(v),\mathit{low}(v))$ and $(\mathit{low2D}(v),\mathit{low2}(v))$ are two distinct back-edges in $B(v)$.
We have defined $\mathit{low1}$, $\mathit{low1D}$, $\mathit{low2}$ and $\mathit{low2D}$ because we need to have stored, for every vertex $v\neq r$, two back-edges from $B(v)$ (see section \ref{section:one_tree_edge}). Any other pair of back-edges from $B(v)$ could do as well. 
It is easy to compute all $\mathit{low1}(v)$, $\mathit{low1D}(v)$, $\mathit{low2}(v)$ and $\mathit{low2D}(v)$ during the DFS. 

We let $l(v)$ denote the lowest $y$ for which there exists a back-edge $(v,y)$, or $v$ if no such back-edge exists. Thus, $\mathit{low}(v)\leq l(v)$. Now let $c_1,\dotsc,c_k$ be the children of $v$ sorted in non-decreasing order w.r.t. their $\mathit{low}$ point.
Then we call $c_1$ the $\mathit{low1}$ child of $v$, and $c_2$ the $\mathit{low2}$ child of $v$. (Of course, the $\mathit{low1}$ and $\mathit{low2}$ children of $v$ are not uniquely determined after a DFS on $G$, since we may have $\mathit{low}(c_1)=\mathit{low}(c_2)$.) We let $\tilde{M}(v)$ denote the nearest common ancestor of all $x$ for which there exists a back-edge $(x,y)\in B(v)$ with $x$ a proper descendant of $M(v)$. Formally, $\tilde{M}(v)$ $:=$ nca$\{x\mid\exists (x,y)\in B(v) \mbox{ and } x\neq M(v)\}$. If the set $\{x\mid\exists (x,y)\in B(v) \mbox{ and } x\neq M(v)\}$ is empty, we leave $\tilde{M}(v)$ undefined. We also define $M_{low1}(v)$ as the nearest common ancestor of all $x$ for which there exists a back-edge $(x,y)\in B(v)$ with $x$ being a descendant of the $\mathit{low1}$ child of $M(v)$, and $M_{low2}(v)$ as the nearest common ancestor of all $x$ for which there exists a back-edge $(x,y)\in B(v)$ with $x$ a descendant of the $\mathit{low2}$ child of $M(v)$. Formally, $M_{low1}(v)$ $:=$ nca$\{x\mid\exists (x,y)\in B(v) \mbox{ and } x \mbox{ is a descendant of the } \mathit{low1} \mbox{ child of } M(v)\}$ and $M_{low2}(v)$ $:=$ nca$\{x\mid\exists (x,y)\in B(v) \mbox{ and } x \mbox{ is a descendant of the } \mathit{low2} \mbox{ child of } M(v)\}$. If the set in the formal definition of $M_{low1}(v)$ (resp. $M_{low2}(v)$) is empty, we leave $M_{low1}(v)$ (resp. $M_{low2}(v)$) undefined.
\ignore{
Algorithm \ref{algorithm:computeM} shows how we can compute all $M(v)$ and $\mathit{nextM}(v)$, algorithm \ref{algorithm:computeM2} shows how we can compute all $\tilde{M}(v)$, and algorithm \ref{algorithm:computeM3} shows how we can compute all $M_{low1}(v)$ and $M_{low2}(v)$, for all vertices $v\neq r$, in total linear time. These algorithms process the vertices in a bottom-up fashion, and they work recursively on the descendants of a vertex. To perform these computations in linear time, we have to avoid descending to the same vertices an excessive amount of times during the recursion. To achieve this, we use a variable $\mathit{currentM}[w]$, that has the property that, during the course of the algorithm, when we process a vertex $v$, all back-edges that start from a descendant of $w$ and end in a proper ancestor of $v$ have their higher end in $T(\mathit{currentM}[w])$ (this means, of course, that $\mathit{currentM}[w]$ is a descendant of $w$). And so, if we want e.g. to compute $M_{low1}(v)$, we may descend immediately to $\mathit{currentM}[c_1]$, where $c_1$ is the $\mathit{low1}$ child of $M(v)$. In Lemma \ref{lemma:Malg_is_correct}, we give a formal proof of the correctness and linear complexity of algorithms \ref{algorithm:computeM2} and \ref{algorithm:computeM3}.
}

\subsection{Computing the DFS parameters in linear time}
\label{sec:dfs}

Algorithm \ref{algorithm:high} shows how we can easily compute $\mathit{highD}(v)$ during the computation of all $\mathit{high}$ points. The algorithm uses the static tree disjoint-set-union data structure of Gabow and Tarjan~\cite{dsu:gt} to achieve linear running time.

\begin{algorithm}
\caption{\textsf{Compute all $\mathit{high}(v)$ and $\mathit{highD}(v)$, for all vertices $v\neq r$}}
\label{algorithm:high}
\LinesNumbered
\DontPrintSemicolon
initialize a DSU structure on the vertices of $G$, where the $\mathit{link}$ operations are predetermined by the edges of $T$\;
\For{$v=n$ to $v=1$}{
  \ForEach{$u$ adjacent to $v$}{
    \If{$u$ is a descendant of $v$}{
      $x \leftarrow \mathit{find}(u)$\;
      \While{$x>v$}{
        $\mathit{high}[x] \leftarrow v$\;
        $\mathit{highD}[x] \leftarrow u$\;
        $\mathit{next} \leftarrow \mathit{find}(p(x))$\;
        $\mathit{link}(x,p(x))$\;
        $x \leftarrow \mathit{next}$
      }
    }
  }
}
\end{algorithm}

Algorithm \ref{algorithm:computeM} shows how we can compute all $M(v)$ and $\mathit{nextM}(v)$, algorithm \ref{algorithm:computeM2} shows how we can compute all $\tilde{M}(v)$, and algorithm \ref{algorithm:computeM3} shows how we can compute all $M_{low1}(v)$ and $M_{low2}(v)$, for all vertices $v\neq r$, in total linear time. These algorithms process the vertices in a bottom-up fashion, and they work recursively on the descendants of a vertex. To perform these computations in linear time, we have to avoid descending to the same vertices an excessive amount of times during the recursion. To achieve this, we use a variable $\mathit{currentM}[w]$, that has the property that, during the course of the algorithm, when we process a vertex $v$, all back-edges that start from a descendant of $w$ and end in a proper ancestor of $v$ have their higher end in $T(\mathit{currentM}[w])$ (this means, of course, that $\mathit{currentM}[w]$ is a descendant of $w$). And so, if we want e.g. to compute $M_{low1}(v)$, we may descend immediately to $\mathit{currentM}[c_1]$, where $c_1$ is the $\mathit{low1}$ child of $M(v)$. In Lemma \ref{lemma:Malg_is_correct}, we give a formal proof of the correctness and linear complexity of Algorithms \ref{algorithm:computeM2} and \ref{algorithm:computeM3}.

\begin{algorithm}[!h]
\caption{\textsf{Compute all $M(v)$ and $\mathit{nextM}(v)$, for all vertices $v\neq r$}}
\label{algorithm:computeM}
\LinesNumbered
\DontPrintSemicolon

\tcp{\textit{Compute all $M(v)$ and $\mathit{nextM}(v)$}}
\For{$v=n$ to $v=2$}{
  $\mathit{nextM}[v] \leftarrow \emptyset$\;
  $c \leftarrow v$, $m \leftarrow v$\;
  \While{$M(v)=\emptyset$}{
    \lIf{$l(m)<v$}{$M(v) \leftarrow m$, \textbf{break}}
    $c_1 \leftarrow \mathit{low1}$ child of $m$\;
    $c_2 \leftarrow \mathit{low2}$ child of $m$\;
    \lIf{$\mathit{low}(c_2)<v$}{$M(v) \leftarrow m$, \textbf{break}}
    $c \leftarrow c_1$, $m \leftarrow M(c)$\;
  }
  \lIf{$c\neq v$}{$\mathit{nextM}(c) \leftarrow v$}
}
\end{algorithm}

\begin{algorithm}[!h]
\caption{\textsf{Compute all $\tilde{M}(v)$, for all vertices $v\neq r$}}
\label{algorithm:computeM2}
\LinesNumbered
\DontPrintSemicolon
initialize an array $\mathit{currentM}$ with $n$ entries\;
\tcp{\textit{Compute all $\tilde{M}(v)$}}
\lForEach{vertex $v$}{$\mathit{currentM}[v] \leftarrow v$}
\For{$v=n$ to $v=2$}{
  $m \leftarrow M(v)$\;
  $c \leftarrow \mathit{low1}$ child of $m$\;
  \lIf{$\mathit{low}(c)\geq v$}{$\tilde{M}(v)\leftarrow\emptyset$, \textbf{continue}}
  $c' \leftarrow \mathit{low2}$ child of $m$\;
  \lIf{$\mathit{low}(c')<v$}{$\tilde{M}(v)\leftarrow m$, \textbf{continue}}
  $m \leftarrow \mathit{currentM}[c]$\;
  \While{$\tilde{M}(v)=\emptyset$}{
    \lIf{$l(m)<v$}{$\tilde{M}(v) \leftarrow m$, \textbf{break}}
    $c_1 \leftarrow \mathit{low1}$ child of $m$\;
    $c_2 \leftarrow \mathit{low2}$ child of $m$\;
    \lIf{$\mathit{low}(c_2)<v$}{$\tilde{M}(v) \leftarrow m$, \textbf{break}}
    $m \leftarrow \mathit{currentM}[c_1]$\;
  }
  $\mathit{currentM}[c] \leftarrow m$\;
}
\end{algorithm}

\begin{algorithm}[!h]
\caption{\textsf{Compute all $M_{low1}(v)$ and $M_{low2}(v)$, for all vertices $v\neq r$}}
\label{algorithm:computeM3}
\LinesNumbered
\DontPrintSemicolon
initialize an array $\mathit{currentM}$ with $n$ entries\;
\tcp{\textit{Compute all $M_{low1}(v)$}}
\lForEach{vertex $v$}{$\mathit{currentM}[v] \leftarrow v$}
\For{$v=n$ to $v=2$}{
  $m \leftarrow M(v)$\;
  $c \leftarrow \mathit{low1}$ child of $m$\;
  \lIf{$\mathit{low}(c)\geq v$}{$M_{low1}(v)\leftarrow\emptyset$, \textbf{continue}}
  \label{algM_null}
  $m \leftarrow \mathit{currentM}[c]$\;
  \label{algM_first_assignment}
  \While{$M_{low1}(v)=\emptyset$}{
    \label{algM_while1}
    \lIf{$l(m)<v$}{$M_{low1}(v) \leftarrow m$, \textbf{break}}
    \label{algM_low1}
    $c_1 \leftarrow \mathit{low1}$ child of $m$\;
    $c_2 \leftarrow \mathit{low2}$ child of $m$\;
    \lIf{$\mathit{low}(c_2)<v$}{$M_{low1}(v) \leftarrow m$, \textbf{break}}
    \label{algM_low2}
    $m \leftarrow \mathit{currentM}[c_1]$\;
    \label{algM_setM}
  }
  \label{algM_while2}
  $\mathit{currentM}[c] \leftarrow m$\;
  \label{algM_end}
}
\tcp{\textit{Compute all $M_{low2}(v)$}}
\lForEach{vertex $v$}{$\mathit{currentM}[v] \leftarrow v$}
\For{$v=n$ to $v=2$}{
  $m \leftarrow M(v)$\;
  $c \leftarrow \mathit{low2}$ child of $m$\;
  \lIf{$\mathit{low}(c)\geq v$}{$M_{low2}(v)\leftarrow\emptyset$, \textbf{continue}}
  $m \leftarrow \mathit{currentM}[c]$\;
  \While{$M_{low2}(v)=\emptyset$}{
    \lIf{$l(m)<v$}{$M_{low2}(v) \leftarrow m$, \textbf{break}}
    $c_1 \leftarrow \mathit{low1}$ child of $m$\;
    $c_2 \leftarrow \mathit{low2}$ child of $m$\;
    \lIf{$\mathit{low}(c_2)<v$}{$M_{low2}(v) \leftarrow m$, \textbf{break}}
    $m \leftarrow \mathit{currentM}[c_1]$\;
  }
  $\mathit{currentM}[c] \leftarrow m$\;
}
\end{algorithm}

\begin{lemma}
\label{lemma:MlowEtc}
Let $v$ and $v'$ be two vertices such that $v'$ is an ancestor of $v$ with $M(v')=M(v)$. Then, $\tilde{M}(v')$ (resp. $M_{low1}(v')$, resp. $M_{low2}(v')$), if it is defined, is a descendant of $\tilde{M}(v)$ (resp. $M_{low1}(v)$, resp. $M_{low2}(v)$).
\end{lemma}
\begin{proof}
Let $v'$ be an ancestor of $v$ such that $M(v')=M(v)$.

Assume, first, that $\tilde{M}(v')$ is defined. Then, there exists a back-edge $(x,y)\in B(v')$ where $x$ is a proper descendant of $M(v')$. Since $M(v')=M(v)$, $x$ is a proper descendant of $M(v)$. Furthermore, since $y$ is a proper ancestor of $v'$, it is also a proper ancestor of $v$. This shows that $(x,y)\in B(v)$, and $\tilde{M}(v)$ is an ancestor of $x$. Due to the generality of $(x,y)$, we conclude that $\tilde{M}(v)$ is an ancestor of $\tilde{M}(v')$.

Now assume that $M_{low1}(v')$ is defined. Then, there exists a back-edge $(x,y)\in B(v')$ where $x$ is a descendant of the $\mathit{low1}$ child of $M(v')$. Since $M(v')=M(v)$, $x$ is a descendant of the $\mathit{low1}$ child of $M(v)$. Furthermore, since $y$ is a proper ancestor of $v'$, it is also a proper ancestor of $v$. This shows that $(x,y)\in B(v)$, and $M_{low1}(v)$ is an ancestor of $x$. Due to the generality of $(x,y)$, we conclude that $M_{low1}(v)$ is an ancestor of $M_{low1}(v')$.

Finally, assume that $M_{low2}(v')$ is defined. Then, there exists a back-edge $(x,y)\in B(v')$ where $x$ is a descendant of the $\mathit{low2}$ child of $M(v')$. Since $M(v')=M(v)$, $x$ is a descendant of the $\mathit{low2}$ child of $M(v)$. Furthermore, since $y$ is a proper ancestor of $v'$, it is also a proper ancestor of $v$. This shows that $(x,y)\in B(v)$, and $M_{low2}(v)$ is an ancestor of $x$. Due to the generality of $(x,y)$, we conclude that $M_{low2}(v)$ is an ancestor of $M_{low2}(v')$.
\end{proof}

\begin{lemma}
\label{lemma:Malg_is_correct}
Algorithms \ref{algorithm:computeM2} and \ref{algorithm:computeM3} compute all $\tilde{M}(v)$, $M_{low1}(v)$ and $M_{low2}(v)$, for all vertices $v\neq r$, in total linear time.
\end{lemma}
\begin{proof}
Let us show e.g. that Algorithm \ref{algorithm:computeM3} correctly computes all $M_{low1}(v)$, for all $v\neq r$, in total linear time. The proofs for the other cases are similar. So let $v$ be a vertex $\neq r$. Since we are interested in the back-edges $(x,y)\in B(v)$ with $x$ a descendant of the $\mathit{low1}$ child $c$ of $M(v)$, we first have to check whether $\mathit{low}(c)<v$. If $\mathit{low}(c)\geq v$, then there is no such back-edge, and therefore we set $M_{low1}(v)\leftarrow\emptyset$ (in line \ref{algM_null}). If $\mathit{low}(c)<v$, then $M_{low1}(v)$ is defined, and in line \ref{algM_first_assignment} we assign $m$ the value $\mathit{currentM}[c]$. We claim that, at that moment, $\mathit{currentM}[c]$ is an ancestor of $M_{low1}(v)$, and every $\mathit{currentM}[c_1]$ that we will access in the \textbf{while} loop in line \ref{algM_setM} is also an ancestor of $M_{low1}(v)$; furthermore, when we reach line \ref{algM_end}, $\mathit{currentM}[c]$ is assigned $M_{low1}(v)$. It is not difficult to see this inductively. Suppose, then, that this was the case for every vertex $v'>v$, and let us see what happens when we process $v$. Let $c$ be the $\mathit{low1}$ child of $M(v)$. Initially, $\mathit{currentM}[c]$ was set to be $c$. Now, if $\mathit{currentM}[c]$ is still $c$, $M_{low1}(v)$ is a descendant of $c$ (by definition). Otherwise, due to the inductive hypothesis, $\mathit{currentM}[c]$ had been assigned $M_{low1}(v')$ during the processing of a vertex $v'>v$ with $M(v')=M(v)$. This implies that $v'$ is a descendant of $v$, and by Lemma \ref{lemma:MlowEtc} we have that $M_{low1}(v')$ is an ancestor of $M_{low1}(v)$. In any case, then, we have that $m=\mathit{currentM}[c]$ in an ancestor of $M_{low1}(v)$. Now we enter the \textbf{while} loop in line \ref{algM_while1}. If either $l(m)<v$ or $\mathit{low}(c_2)<v$, where $c_2$ is the $\mathit{low2}$ child of $m$, we have that $M_{low1}(v)$ is an ancestor of $m$. Since $m$ is also an ancestor of $M_{low1}(v)$, we correctly set $M_{low1}(v)\leftarrow m$ (in lines \ref{algM_low1} or \ref{algM_low2}). Otherwise, we have that $M_{low1}(v)$ is a descendant of the $\mathit{low1}$ child $c_1$ of $m$. Now, due to the inductive hypothesis, $\mathit{currentM}[c_1]$ is either $c_1$ or $M_{low1}(v')$ for a vertex $v'>v$ with $M(v')=m$. In the first case we obviously have that $\mathit{currentM}[c_1]$ is an ancestor of $M_{low1}(v)$. Now assume that the second case is true, and let $(x,y)$ be a back-edge with $x$ a descendant of $c_1$ and $y$ a proper ancestor of $v$. Then, since $v'>v$ and $v,v'$ have $m$ as a common descendant, we have that $v$ is ancestor of $v'$, and therefore $y$ is a proper ancestor of $v'$. This shows that $x$ is a descendant of $M_{low1}(v')$. Thus, due to the generality of $(x,y)$, we have that $M_{low1}(v)$ is a descendant of $M_{low1}(v')$. In any case, then, we have that $\mathit{currentM}[c_1]$ is an ancestor of $M_{low1}(v)$. Thus we set $m\leftarrow\mathit{currentM}[c_1]$ and we continue the \textbf{while} loop, until we have that $m=M_{low1}(v)$, in which case we will set $\mathit{currentM}[c]\leftarrow m$ in line \ref{algM_end}.
Thus we have proved that Algorithm \ref{algorithm:computeM3} correctly computes $M_{low1}(v)$, for every vertex $v\neq r$, and that, during the processing of a vertex $v$, every $\mathit{currentM}[c]$ that we access is an ancestor of $M_{low1}(v)$ (until, in line \ref{algM_end}, we assign $\mathit{currentM}[c]$ to $M_{low1}(v)$).

Now, to prove linearity, let $S(v)=\{m_1,\dotsc,m_k\}$, ordered increasingly, denote the (possible empty) set of all vertices that we had to descend to before leaving the \textbf{while} loop in lines \ref{algM_while1}-\ref{algM_while2}. (Thus, if $k\geq 1$, $m_k=M_{low1}(v)$.) In other words, $S(v)$ contains all vertices that were assigned to $m$ in line \ref{algM_setM}. We will show that Algorithm \ref{algorithm:computeM3} runs in linear time, by showing that, for every two vertices $v$ and $v'$, $v\neq v'$ implies that $S(v)\cap S(v')\subseteq\{M_{low1}(v)\}$, where we have $S(v)\cap S(v')=\{M_{low1}(v)\}$ only if $M_{low1}(v)=M_{low1}(v')$. Of course, it is definitely the case that $S(v)\cap S(v')=\emptyset$ if $v$ and $v'$ are not related as ancestor and descendant, since the \textbf{while} loop descends to descendants of the vertex under processing. So let $v'$ be a proper ancestor of $v$. If $M_{low1}(v')$ is not a descendant of the $\mathit{low1}$ child $c$ of $M(v)$, then we obviously have $S(v)\cap S(v')=\emptyset$ (since $S(v)$ consists of descendants of $c$, but the \textbf{while} loop during the computation of $M_{low1}(v')$ will not descend to the subtree of $c$). Thus we may assume that $M_{low1}(v')$ is a descendant of $c$. Now, let $S(v')=\{m_1,\dotsc,m_k\}$ and $m_0=\mathit{currentM}[c']$, where $c'$ is the $\mathit{low1}$ child of $M(v')$. We will show that every $m_i$, for every $i\in\{1,\dotsc,k\}$, is either an ancestor of $M(v)$ or a descendant of $M_{low1}(v)$. (This obviously implies that $S(v')\cap S(v)\subseteq\{M_{low1}(v)\}$.) First observe that $M(v')$ is either an ancestor of $M(v)$ or a descendant of $M_{low1}(v)$. To see this, suppose that $M(v')$ is not an ancestor of $M(v)$. Since $M_{low1}(v')$ is a descendant of $c$, there is at least one back-edge $(x,y)$ in $B(v')$ with $x$ a descendant of $c$. Then, since $y$ is a proper ancestor of $v'$ and $v'$ is a proper ancestor of $v$, we have that $(x,y)$ is in $B(v)$, and therefore $x$ is a descendant of $M_{low1}(v)$. Now let $(x',y')$ be a back-edge in $B(v')$. If $x'$ is a descendant of a vertex in $T[c,v']$, but not a descendant of $c$, then the nearest common ancestor of $x$ and $x'$ is in $T[M(v),v']$, and therefore $M(v')$ is an ancestor of $M(v)$, contradicting our supposition. Thus, $x'$ is a descendant of $c$. Furthermore, $y'$ is a proper ancestor of $v$, and therefore $(x',y')\in B(v)$. Thus, $x'$ is a descendant of $M_{low1}(v)$. Due to the generality of $(x',y')\in B(v')$, we conclude that $M(v')$ is a descendant of $M_{low1}(v)$. Thus we have shown that $M(v')$ is either an ancestor of $M(v)$ or a descendant of $M_{low1}(v)$.

Now, if $M(v')$ is a descendant of $M_{low1}(v)$, we obviously have $S(v)\cap S(v')=\emptyset$. Let's assume, then, that $M(v')$ is an ancestor of $M(v)$. If $M(v')$ coincides with $M(v)$, then $c'=c$, and so $m_0$ coincides with $\mathit{currentM}[c]$, which is a descendant of $M_{low1}(v)$ (since $M_{low1}(v)$ has already been calculated), and therefore every $m_i$, for every $i\in\{1,\dotsc,k\}$, is a proper descendant of $M_{low1}(v)$ (since $m_1$, if it exists, is a proper descendant of $m_0$), and so we have $S(v')\cap S(v)=\emptyset$.  So let's assume that $M(v')$ is a proper ancestor of $M(v)$. Then, $c'$ is an ancestor of $M(v)$. Suppose that $m_0$ is not an ancestor of $M(v)$. This means that $\mathit{currentM}[c']\neq c'$, and therefore there is a vertex $\tilde{v}>v'$ with $M(\tilde{v})=M(v')$ and $M_{low1}(\tilde{v})=\mathit{currentM}[c']$. Furthermore, since $m_0$ is not an ancestor of $M(v)$, it must be a descendant of $c$. Now, since $v'$ is an ancestor of $v$ and $M(v')$ is a proper ancestor of $M(v)$, Lemma \ref{lemma:Muv} implies that $M(v')$ is a proper ancestor of $v$. Since $M(v')=M(\tilde{v})$, this implies that $M(\tilde{v})$ is a proper ancestor of $v$, and therefore $\tilde{v}$ is a proper ancestor of $v$. Now let $(x,y)$ be a back-edge in $B(\tilde{v})$ such that $x$ is a descendant of $M_{low1}(\tilde{v})=\mathit{currentM}[c']=m_0$. Then, since $m_0$ is a descendant of $c$, $x$ is also  descendant of $c$. Furthermore, since $\tilde{v}$ is an ancestor of $v$, $y$ is a proper ancestor of $v$. This shows that $x$ is a descendant of $M_{low1}(v)$. Due to the generality of $(x,y)$, we conclude that $M_{low1}(\tilde{v})$ is a descendant of $M_{low1}(v)$. Thus we have shown that $m_0$ is either an ancestor of $M(v)$ or a descendant of $M_{low1}(v)$.

Now let's assume that $m_i$ is either an ancestor of $M(v)$ or a descendant of $M_{low1}(v)$, for some $i\in\{0,\dotsc,k-1\}$. We will prove that the same is true for $m_{i+1}$. If $m_i$ is a descendant of $M_{low1}(v)$, then the same is true for $m_{i+1}$. Let's assume, then, that $m_i$ is an ancestor of $M(v)$. Now we have that $m_{i+1}=\mathit{currentM}[c_1]$, where $c_1$ is the $\mathit{low1}$ child of $m_i$. If $m_i=M(v)$, then we have $c_1=c$, and therefore $\mathit{currentM}[c_1]=\mathit{currentM}[c]$ is a descendant of $M_{low1}(v)$ (since $M_{low1}(v)$ has already been computed). Suppose, then, that $m_i$ is a proper ancestor of $M(v)$. Then, $c_1$ is an ancestor of $M(v)$. If $\mathit{currentM}[c_1]=c_1$, we obviously have that $\mathit{currentM}[c_1]$ is an ancestor of $M(v)$. Otherwise, if $\mathit{currentM}[c_1]\neq c_1$, there is a vertex $\tilde{v}$ such that $M(\tilde{v})=m_i$ and $\mathit{currentM}[c_1]=M_{low1}(\tilde{v})$. Assume, first, that $\tilde{v}$ is an ancestor of $v$. Suppose that $M_{low1}(\tilde{v})$ is not an ancestor of $M(v)$. Then it must be a descendant of $c$. Let $(x,y)$ be a back-edge in $B(\tilde{v})$ with $x$ a descendant of $M_{low1}(\tilde{v})$. Then $x$ is a descendant of $c$. Furthermore, $y$ is a proper ancestor of $\tilde{v}$, and therefore a proper ancestor of $v$. This shows that $x$ is a descendant of $M_{low1}(v)$. Due to the generality of $(x,y)$, we conclude that $M_{low1}(\tilde{v})$ is a descendant of $M_{low1}(v)$. Thus, if $\tilde{v}$ is an ancestor of $v$, $M_{low1}(\tilde{v})$ is either an ancestor of $M(v)$ or a descendant of $M_{low1}(v)$. Suppose, now, that $\tilde{v}$ is a descendant of $v$. Let $(x,y)$ be a back-edge in $B(v)$. Then, $x$ is a descendant of $M(v)$, and therefore a descendant of $c_1$. Furthermore, $y$ is a proper ancestor of $v$, and therefore a proper ancestor of $\tilde{v}$. This shows that $x$ is a descendant of $M_{low1}(\tilde{v})$. Due to the generality of $(x,y)$, we conclude that $M(v)$ is a descendant of $M_{low1}(\tilde{v})$. In any case, then, $m_{i+1}$ is either an ancestor of $M(v)$ or a descendant of $M_{low1}(v)$. Thus, $S(v)\cap S(v')\subseteq\{M_{low1}(v)\}$ is established.
\end{proof}

\section{Computing the $3$-cuts of a $3$-edge-connected graph}
\label{sec:3-cuts}

In this section we present a linear-time algorithm that computes all the $3$-edge-cuts of a $3$-edge-connected graph $G=(V,E)$.
It is well-known that the number of the $3$-edge-cuts of $G$ is $O(n)$~\cite{connectivity:nagamochi-ibaraki} (e.g., it follows from the definition of the cactus graph~\cite{Cactus,KP:Cactus}), but we provide an independent proof of this fact.
Then, in Section~\ref{sec:reduction}, we show how to extend this algorithm so that it can also count the number of minimal $3$-edge-cuts of a general graph. Note that there can be $O(n^3)$ such cuts~\cite{3cuts:Dinitz}.

%
Our method is to classify the $3$-cuts on the DFS-tree $T$ in a way that allows us to compute them efficiently. If $\{e_1,e_2,e_3\}$ is a $3$-cut, we can initially distinguish three cases: either $e_1$ is a tree-edge and both $e_2$ and $e_3$ are back-edges (section \ref{section:one_tree_edge}), or $e_1$ and $e_2$ are two tree-edges and $e_3$ is a back-edge (section \ref{section:two_tree_edges}), or $e_1$, $e_2$ and $e_3$ is a triplet of tree-edges (section \ref{section:three_tree_edges}). Then, we divide those cases in subcases based on the concepts we have introduced in the previous section. Figure \ref{figure:all_cases} gives a general overview of the cases we will handle in detail in the following sections.

\begin{figure}
\begin{center}
\centerline{\includegraphics[trim={0cm 14cm 0cm 0cm},clip, width=0.9\textwidth]{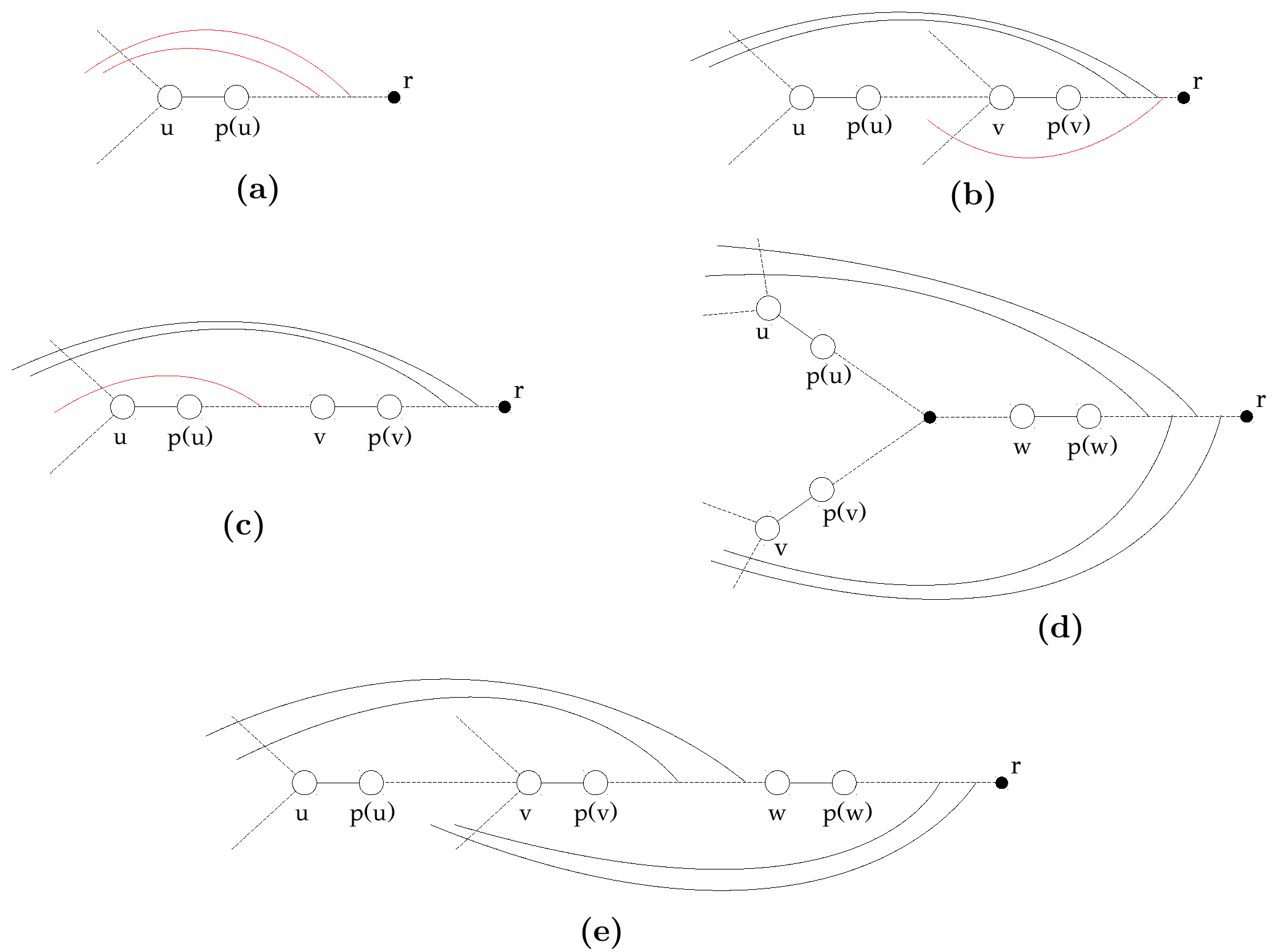}}
\caption{The types of $3$-cuts with respect to a DFS-tree. \textbf{(a)} One tree-edge $(u,p(u))$ and two back-edges (section \ref{section:one_tree_edge}). \textbf{(b)} Two tree-edges $(u,p(u))$ and $(v,p(v))$, where $u$ is a descendant of $v$, and one-back edge in $B(v)\setminus B(u)$ (section \ref{section:uv}). \textbf{(c)} Two tree-edges $(u,p(u))$ and $(v,p(v))$, where $u$ is a descendant of $v$, and one-back edge in $B(u)\setminus B(v)$ (section \ref{section:vu}). \textbf{(d)} Three tree-edges $(u,p(u))$, $(v,p(v))$ and $(w,p(w))$, where $w$ is an ancestor of $u$ and $v$, but $u$ and $v$ are not related as ancestor and descendant (section \ref{section:wuv}). \textbf{(d)} Three tree-edges $(u,p(u))$, $(v,p(v))$ and $(w,p(w))$, where $u$ is a descendant of $v$ and $v$ is a descendant of $w$ (section \ref{section:uvw}).}
\label{figure:all_cases}
\end{center}
\end{figure}

\subsection{One tree-edge and two back-edges}
\label{section:one_tree_edge}

\begin{lemma}
Let $\{(u,p(u)),e,e'\}$ be a $3$-cut such that $e$ and $e'$ are back-edges. Then $B(u)=\{e,e'\}$. Conversely, if for a vertex $u\neq r$ we have $B(u)=\{e,e'\}$ where $e$ and $e'$ are back-edges, then $\{(u,p(u)),e,e'\}$ is a $3$-cut.
\end{lemma}
\begin{proof}
After removing the tree-edge $(u,p(u))$, the edges that connect $T(u)$ with the rest of the graph are precisely those contained in $B(u)$. Let $e$ and $e'$ be two back-edges in $B(u)$. Then it is obvious that $\{(u,p(u)),e,e'\}$ is a $3$-cut if and only if $B(u)$ consists precisely of these two back-edges.
\end{proof}

Thus, to find all $3$-cuts of the form $\{(u,p(u)),e,e'\}$, where $e$ and $e'$ are back-edges, we only have to store, for every vertex $u$, two back-edges $e,e'\in B(u)$. Since $(\mathit{low1D}(u),\mathit{low1}(u))$ and  $(\mathit{low2D}(u),\mathit{low2}(u))$ are two such back-edges, we mark the triplet $\{(u,p(u)),(\mathit{low1D}(u),\mathit{low1}(u)),\linebreak (\mathit{low2D}(u),\mathit{low2}(u))\}$, for every $u$ that has $\mathit{b\_count}(u)=2$.

\subsection{Two tree-edges and one back-edge}
\label{section:two_tree_edges}

\begin{lemma}
\label{lemma:relation_of_uv}
Let $\{(u,p(u)),(v,p(v)),e\}$ be a $3$-cut such that $e$ is a back-edge. Then $u$ and $v$ are related as ancestor and descendant.
\end{lemma}
\begin{proof}
Suppose that $u$ and $v$ are not related as ancestor or descendant. Since the graph is $3$-edge-connected, $\mathit{b\_count}(u)>1$, and therefore there is least one back-edge $(x,y)\in B(u)\setminus\{e\}$. Since $v$ is not a descendant of $u$, $v\notin T[x,u]$; and since $v$ is not an ancestor of $u$, $v\notin T[p(u),y]$. Thus, by removing the edges $(u,p(u))$, $(v,p(v))$, and $e$, from the graph, $u$ remains connected with $p(u)$, through the path $T[u,x],(x,y),T[p(u),y]$. This contradicts that fact that $\{(u,p(u)),(v,p(v)),e\}$ is a $3$-cut.
\end{proof}

\begin{proposition}
\label{proposition:B(u)B(v)}
Let $\{(u,p(u)),(v,p(v)),e\}$ be a $3$-cut, where $e$ is a back-edge. Then, either $(1)$ $B(v)=B(u)\sqcup \{e\}$ or $(2)$ $B(u)=B(v)\sqcup \{e\}$. Conversely, if there exists a back-edge $e$ such that $(1)$ or $(2)$ is true, then $\{(u,p(u)),(v,p(v)),e\}$ is a $3$-cut.
\end{proposition}
\begin{proof}
($\Rightarrow$) By Lemma \ref{lemma:relation_of_uv}, we may assume, without loss of generality, that $v$ is an ancestor of $u$. Now, suppose that $(1)$ does not hold; we will prove that $(2)$ does. Since $(1)$ is not true, there must exist a back-edge $e'$ such that $e'\in B(v)$ and $e'\notin B(u)\cup\{e\}$, or $e'\notin B(v)$ and $e'\in B(u)\cup\{e\}$. Suppose the first is true: that is, there exists a back-edge $(x,y)$ such that $(x,y)\in B(v)$ and $(x,y)\notin B(u)\cup\{e\}$. Then $y$ is an ancestor of $v$, and therefore an ancestor of $u$. But, since $(x,y)\notin B(u)$, $x$ cannot be a descendant of $u$, and thus it belongs to $T(v)\setminus T(u)$. Now, by removing the edges $(u,p(u))$, $(v,p(v))$ and $e$ from the graph, we can see that $v$ remains connected with $p(v)$ through the path $T[v,x],(x,y),T[y,p(v)]$. This contradicts the fact that $\{(u,p(u)),(v,p(v)),e\}$ is a $3$-cut. Thus we have shown that there exists a back-edge $e'$ such that $e'\notin B(v)$ and $e'\in B(u)\cup\{e\}$, and also that $B(v)\subseteq B(u)\cup\{e\}$. Now, suppose that there exists a back-edge $(x,y)\neq e$ such that $(x,y)\notin B(v)$ and $(x,y)\in B(u)$. Then $x$ is a descendant of $u$, and therefore a descendant of $v$. But, since $(x,y)\notin B(v)$, $y$ is not a proper ancestor of $v$, and thus belongs to $T[p(u),v]$. Now, by removing the edges $(u,p(u))$, $(v,p(v))$ and $e$ from the graph, we can see that $u$ remains connected with $p(u)$ through the path $T[u,x],(x,y),T[y,p(u)]$. This contradicts the fact that $\{(u,p(u)),(v,p(v)),e\}$ is a $3$-cut. Thus we have shown that $e$ is the unique back-edge such that $e\notin B(v)$ and $e\in B(u)$, and also that $B(u)\subseteq B(v)\cup\{e\}$. In conjunction with $B(v)\subseteq B(u)\cup\{e\}$, this implies that $B(u)=B(v)\sqcup\{e\}$.\\
($\Leftarrow$) First, observe that both $(1)$ and $(2)$ imply that $u$ and $v$ are related as ancestor and descendant: Since the graph is $2$-edge-connected, we have $\mathit{b\_count}(x)>0$, for every vertex $x\neq r$; and whenever we have $B(u)\cap B(v)\neq\emptyset$, for two vertices $u$ and $v$, (and such is the case if either $(1)$ or $(2)$ is true), we can infer that $u$ and $v$ are related as ancestor and descendant. Now, due to the symmetry of the relations $(1)$ and $(2)$, we may assume, without loss of generality, that $v$ is an ancestor of $u$. Let's assume first that $(1)$ is true, and let $e=(x,y)$. Since $(x,y)\in B(v)$, $y$ is a proper ancestor of $v$, and therefore a proper ancestor of $u$. But, since $(x,y)\notin B(u)$, $x$ cannot be a descendant of $u$, and thus it belongs to $T(v)\setminus T(u)$. Furthermore, this is the only back-edge that starts from $T(v)\setminus T(u)$ and ends in a proper ancestor of $v$, since $B(v)\setminus\{e\}=B(u)$. Thus we can see that, by removing the edges $(u,p(u))$, $(v,p(v))$ and $e$ from the graph, the graph becomes disconnected. (For the subgraph $T(v)\setminus T(u)$ becomes disconnected from $T(u)\cup (T(r)\setminus T(v))$.) Now assume that $(2)$ is true, and let $e=(x,y)$. Since $(x,y)\in B(u)$, $x$ is a descendant of $u$, and therefore a descendant of $v$. But, since $(x,y)\notin B(v)$, $y$ is not a proper ancestor of $v$, and thus it belongs to $T[p(u),v]$. Furthermore, it is the only back-edge that starts from $T(u)$ and ends in $T[p(u),v]$, since $B(u)\setminus\{e\}=B(v)$. Thus we can see that, by removing the edges $(u,p(u))$, $(v,p(v))$ and $e$ from the graph, the graph becomes disconnected. (For the subgraph $T(v)\setminus T(u)$ becomes disconnected from $T(u)\cup (T(r)\setminus T(v))$.)
\end{proof}

Here we distinguish two cases, depending on whether $B(v)=B(u)\sqcup\{e\}$ or $B(u)=B(v)\sqcup\{e\}$.

\subsubsection{$v$ is an ancestor of $u$ and $B(v)=B(u)\sqcup\{e\}$.}
\label{section:uv}

Throughout this section let $V(u)$ denote the set of vertices $v$ that are ancestors of $u$ and such that $B(v)=B(u)\sqcup\{e\}$, for a back-edge $e$. By proposition \ref{proposition:B(u)B(v)}, this means that $\{(u,p(u)),(v,p(v)),e\}$ is a $3$-cut. The following lemma shows that, for every vertex $v$, there is at most one vertex $u$ such that $v\in V(u)$.

\begin{lemma}
Let $u,u'$ be two distinct vertices. Then $V(u)\cap V(u')=\emptyset$.
\end{lemma}
\begin{proof}
Suppose that there exists a $v\in V(u)\cap V(u')$. Then there are back-edges $e,e'$ such that $B(v)=B(u)\sqcup\{e\}$ and $B(v)=B(u')\sqcup\{e'\}$, and so we have $B(u)\sqcup\{e\}=B(u')\sqcup\{e'\}$. Since $\mathit{b\_count}(u)>1$ and $\mathit{b\_count}(u')>1$ (for the graph is $3$-edge-connected), we infer that $B(u)\cap B(u')\neq\emptyset$, and thus $u$ and $u'$ are related as ancestor and descendant. Thus we can assume, without loss of generality, that $u'$ is an ancestor of $u$. Now let $(x,y)\in B(u)$. Then $x$ is a descendant of $u$, and therefore a descendant of $u'$. Furthermore, since $B(v)=B(u)\sqcup\{e\}$, we have $(x,y)\in B(v)$, and so $y$ is a proper ancestor of $v$, and therefore a proper ancestor of $u'$. This shows that $(x,y)\in B(u')$, and thus we have $B(u)\subseteq B(u')$. In conjunction with $B(u)\sqcup\{e\}=B(u')\sqcup\{e'\}$ (which implies that $|B(u)|=|B(u')|$), we infer that $B(u)=B(u')$ (and $e=e'$). This contradicts the fact that the graph is $3$-edge-connected.
\end{proof}

Thus, the total number of $3$-cuts of the form $\{(u,p(u)),(v,p(v)),e\}$, where $u$ is a descendant of $v$ and $e$ is a back-edge such that $B(v)=B(u)\sqcup\{e\}$, is $O(n)$. Now we will show how to compute, for every vertex $v$, the vertex $u$ such that $v\in V(u)$ (if such a vertex $u$ exists), together with the back-edge $e$ such that $\{(u,p(u)),(v,p(v)),e\}$ is a $3$-cut, in total linear time.

Let $u,v,e$ be such that $v\in V(u)$ and $B(v)=B(u)\sqcup\{e\}$, and let $e=(x,y)$. Then $y$ is a proper ancestor of $v$, and therefore a proper ancestor of $u$, so $x$ cannot be a descendant of $v$ (since $e\notin B(u)$). Thus, $x$ is either on the tree-path $T(u,v]$, or it is a proper descendant of a vertex in $T(u,v]$, but not a descendant of $u$. In the first case we have $\tilde{M}(v)=M(u)$ (and $x=M(v)$); in the second case either $M_{low1}(v)=M(u)$ (and $x=M_{low2}(v)$) or $M_{low2}(v)=M(u)$ (and $x=M_{low1}(v)$). (For an illustration, see figure \ref{figure:V(u)}.) The following lemma shows how we can determine $u$ from $v$.

\begin{figure}
\begin{center}
\centerline{\includegraphics[trim={0cm 22cm 0cm 0cm}, scale=1,clip, width=\textwidth]{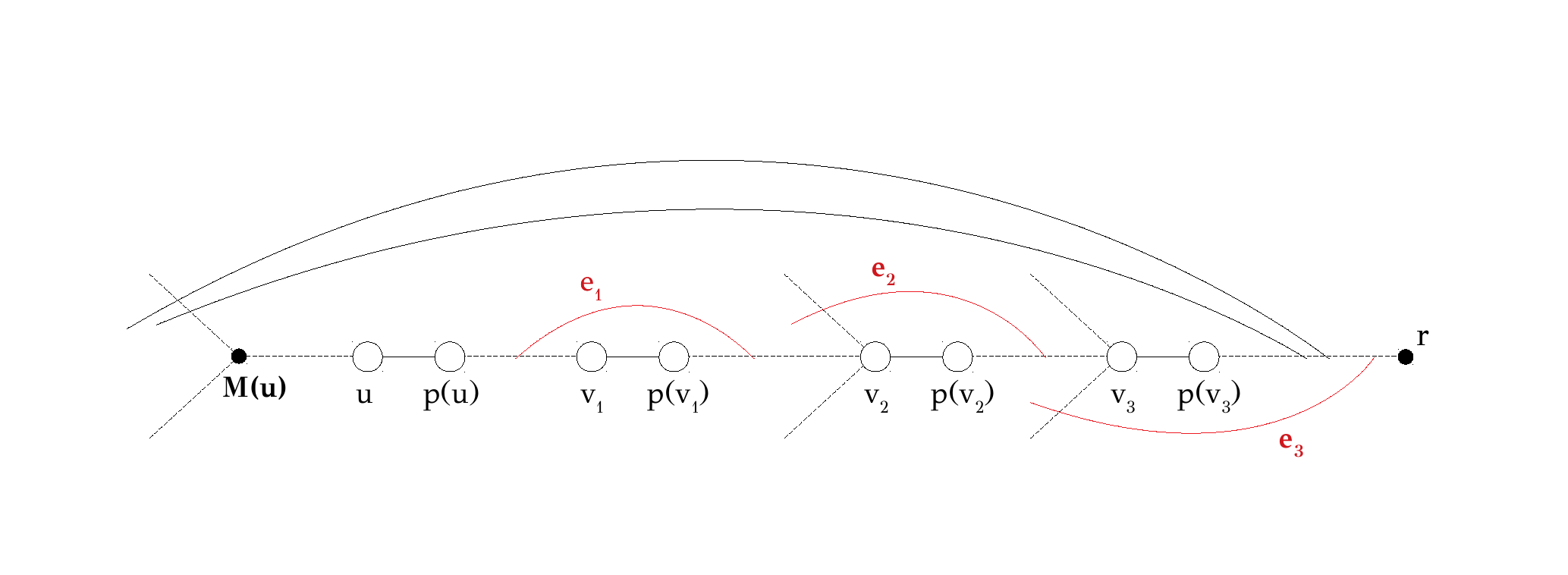}}
\caption{In this example we have $V(u)=\{v_1,v_2,v_3\}$, and every back-edge $e_i$ satisfies $B(v_i)=B(u)\sqcup\{e_i\}$. It should be clear that every $M(v_i)$ is an ancestor of $M(u)$, and $\tilde{M}(v_1)=M(u)$, $M_{low1}(v_2)=M(u)$ and $M_{low2}(v_3)=M(u)$. It is perhaps worth noting that, for every vertex $u$, we may have many vertices $v\in V(u)$ with $\tilde{M}(v)=M(u)$ or $M_{low1}(v)=M(u)$, but only the lowest $v$ in $V(u)$ may have $M_{low2}(v)=M(u)$.}
\label{figure:V(u)}
\end{center}
\end{figure}

\begin{lemma}
\label{lemma:find_u_from_v}
Let $v$ be an ancestor of $u$ such that $\tilde{M}(v)=M(u)$ or $M_{low1}(v)=M(u)$ or $M_{low2}(v)=M(u)$, and let $m$ $=$ $\tilde{M}(v)$ or $M_{low1}(v)$ or $M_{low2}(v)$, depending on whether $\tilde{M}(v)=M(u)$ or $M_{low1}(v)=M(u)$ or $M_{low2}(v)=M(u)$. Then, $v\in V(u)$ if and only if $u$ is the lowest element in $M^{-1}(m)$ which is greater than $v$ and such that $\mathit{high}(u)<v$ and $\mathit{b\_count}(v)=\mathit{b\_count}(u)+1$.
\end{lemma}
\begin{proof}
($\Rightarrow$) $v\in V(u)$ means that there exists a back-edge $e$ such that $B(v)=B(u)\sqcup\{e\}$. Thus we get immediately $\mathit{b\_count}(v)=\mathit{b\_count}(u)+1$ as a consequence. Furthermore, since $B(u)\subset B(v)$, we also get $\mathit{high}(u)<v$ (since for every $(x,y)\in B(u)$ it must be the case that $y$ is a proper ancestor of $v$, and therefore $\mathit{high}(u)$ is a proper ancestor of $v$). Now, suppose that there exists a $u'\in M^{-1}(m)$ which is lower than $u$ and greater than $v$. Then, since $B(u)=B(u')$ (and, in particular, $B(u')\subset B(u)$), there is a back-edge $(x,y)\in B(u)$ with $x\in T(u)$ and $y\in T[p(u),u']$. But this contradicts the fact that $\mathit{high}(u)<v$. \\
($\Leftarrow$) Let $(x,y)\in B(u)$. Then $x$ is a descendant of $u$, and therefore a descendant of $v$. Furthermore, $\mathit{high}(u)<v$ implies that $y$ is a proper ancestor of $v$. This shows that $(x,y)\in B(v)$, and thus we have $B(u)\subseteq B(v)$. Then, $\mathit{b\_count}(v)=\mathit{b\_count}(u)+1$ implies the existence of a back-edge $e\in B(v)\setminus B(u)$ such that $B(v)=B(u)\sqcup\{e\}$.
\end{proof}

Thus, for every vertex $v$, we have to check whether the lowest element $u$ of $M^{-1}(m)$ which is greater than $v$ satisfies $\mathit{b\_count}(v)=\mathit{b\_count}(u)+1$, for all $m\in\{\tilde{M}(v),M_{low1}(v),M_{low2}(v)\}$. To do this efficiently, we process the vertices in a bottom-up fashion, and we keep in a variable $\mathit{currentVertex}[m]$ the lowest element of $M^{-1}(m)$ currently under consideration, so that we do not have to traverse the list $M^{-1}(m)$ from the beginning each time we process a vertex. Algorithm \ref{algorithm:B(v)=B(u)e} is an implementation of this procedure.

\begin{algorithm}
\caption{\textsf{Find all $3$-cuts $\{(u,p(u)),(v,p(v)),e)\}$, where $u$ is a descendant of $v$ and $B(v)=B(u)\sqcup\{e\}$, for a back-edge $e$.}}
\label{algorithm:B(v)=B(u)e}
\LinesNumbered
\DontPrintSemicolon
initialize an array $\mathit{currentVertex}$ with $n$ entries\;
\tcp{\textit{$m=\tilde{M}(v)$}}
\lForEach{vertex $x$}{$\mathit{currentVertex}[x] \leftarrow x$}
\For{$v\leftarrow n$ to $v=1$}{
  $m \leftarrow \tilde{M}(v)$\;
  \lIf{$m=\emptyset$}{\textbf{continue}}
  \tcp{\textit{find the lowest $u\in M^{-1}(m)$ which is greater than $v$}}
  $u \leftarrow \mathit{currentVertex}[m]$\;
  \lWhile{$\mathit{nextM}(u)\neq\emptyset$ \textbf{and} $\mathit{nextM}(u)> v$}{$u \leftarrow \mathit{nextM}(u)$}
  $\mathit{currentVertex}[m] \leftarrow u$\;
  \tcp{\textit{check the condition in lemma \ref{lemma:find_u_from_v}}}
  \If{$\mathit{high}(u)<v$ \textbf{and} $\mathit{b\_count}(v)=\mathit{b\_count}(u)+1$}{
    mark the triplet $\{(u,p(u)),(v,p(v)),(M(v),l(M(v)))\}$
  }
}
\tcp{\textit{$m=M_{low1}(v)$}}
\lForEach{vertex $x$}{$\mathit{currentVertex}[x] \leftarrow x$}
\For{$v\leftarrow n$ to $v=1$}{
  $m \leftarrow M_{low1}(v)$\;
  \lIf{$m=\emptyset$}{\textbf{continue}}
  \tcp{\textit{find the lowest $u\in M^{-1}(m)$ which is greater than $v$}}
  $u \leftarrow \mathit{currentVertex}[m]$\;
  \lWhile{$\mathit{nextM}(u)\neq\emptyset$ \textbf{and} $\mathit{nextM}(u)> v$}{$u \leftarrow \mathit{nextM}(u)$}
  $\mathit{currentVertex}[m] \leftarrow u$\;
  \tcp{\textit{check the condition in lemma \ref{lemma:find_u_from_v}}}
  \If{$\mathit{high}(u)<v$ \textbf{and} $\mathit{b\_count}(v)=\mathit{b\_count}(u)+1$}{
    mark the triplet $\{(u,p(u)),(v,p(v)),(M_{low2}(v),l(M_{low2}(v)))\}$
  }
}
\tcp{\textit{$m=M_{low2}(v)$}}
\lForEach{vertex $x$}{$\mathit{currentVertex}[x] \leftarrow x$}
\For{$v\leftarrow n$ to $v=1$}{
  $m \leftarrow M_{low2}(v)$\;
  \lIf{$m=\emptyset$}{\textbf{continue}}
  \tcp{\textit{find the lowest $u\in M^{-1}(m)$ which is greater than $v$}}
  $u \leftarrow \mathit{currentVertex}[m]$\;
  \lWhile{$\mathit{nextM}(u)\neq\emptyset$ \textbf{and} $\mathit{nextM}(u)> v$}{$u \leftarrow \mathit{nextM}(u)$}
  $\mathit{currentVertex}[m] \leftarrow u$\;
  $\mathit{currentVertex}[m] \leftarrow \mathit{prev}$\;
  \tcp{\textit{check the condition in lemma \ref{lemma:find_u_from_v}}}
  \If{$\mathit{high}(u)<v$ \textbf{and} $\mathit{b\_count}(v)=\mathit{b\_count}(u)+1$}{
    mark the triplet $\{(u,p(u)),(v,p(v)),(M_{low1}(v),l(M_{low1}(v)))\}$
  }
}
\end{algorithm}

\subsubsection{$v$ is an ancestor of $u$ and $B(u)=B(v)\sqcup\{e\}$.}
\label{section:vu}

Throughout this section let $U(v)$ denote the set of vertices $u$ that are descendants of $v$ and such that $B(u)=B(v)\sqcup\{e\}$, for a back-edge $e$. By proposition \ref{proposition:B(u)B(v)}, this means that $\{(u,p(u)),(v,p(v)),e\}$ is a $3$-cut. The following lemma shows that, for every vertex $u$, there is at most one vertex $v$ such that $u\in U(v)$.

\begin{lemma}
Let $v,v'$ be two distinct vertices. Then $U(v)\cap U(v')=\emptyset$.
\end{lemma}
\begin{proof}
Suppose that there exists a $u\in U(v)\cap U(v')$. Then $v$ and $v'$ are related as ancestor and descendant, since they have a common descendant. Thus we may assume, without loss of generality, that $v'$ is an ancestor of $v$. Let $(x,y)$ be a back-edge in $B(v')$. Then, $y$ is a proper ancestor of $v'$, and therefore a proper ancestor of $v$. Furthermore, $u\in U(v')$ implies that $B(v')\subseteq B(u)$, and therefore $(x,y)\in B(u)$. Thus, $x$ is a descendant of $u$, and therefore a descendant of $v$. This shows that $(x,y)\in B(v)$, and thus we have $B(v')\subseteq B(v)$. Now, $u\in U(v)\cap U(v')$ means that there exist two back-edges $e,e'$ such that $B(u)=B(v)\sqcup\{e\}$ and $B(u)=B(v)\sqcup\{e\}$, and thus we have $B(v)\sqcup\{e\}=B(v')\sqcup\{e'\}$. Therefore, $|B(v)|=|B(v')|$. In conjunction with $B(v')\subseteq B(v)$, this implies that $B(v)=B(v')$ (and $e=e'$), contradicting the fact that the graph is $3$-edge-connected.
\end{proof}

Thus, the total number of $3$-cuts of the form $\{(u,p(u)),(v,p(v)),e\}$, where $u$ is a descendant of $v$ and $e$ is a back-edge such that $B(u)=B(v)\sqcup\{e\}$, is $O(n)$. We will now show how to compute, for every vertex $u$, the vertex $v$ such that $u\in U(v)$ (if such a vertex $v$ exists), together with the back-edge $e$ such that $\{(u,p(u)),(v,p(v)),e\}$ is a $3$-cut, in total linear time.

Let $u,v,e$ be such that $u\in U(v)$ and $B(u)=B(v)\sqcup\{e\}$, and let $e=(x,y)$. Then, $x$ is a descendant of $u$, and therefore a descendant of $v$. But since $e\notin B(v)$, $y$ is not an ancestor of $v$, and therefore $y\in T[p(u),v]$. Thus, $y=\mathit{high}(u)$ (and $x=\mathit{highD}(u)$), since every other back-edge $(x',y')\in B(u)$ is also in $B(v)$ and thus has $y'<v\leq y$. This shows how we can determine the back-edge $e$ from a pair of vertices $u,v$ that satisfy $u\in U(v)$. Furthermore, $B(u)=B(v)\sqcup\{e\}$ implies that $M(u)$ is an ancestor of $M(v)$. Thus, either $M(u)=M(v)$, or $M(u)$ is a proper ancestor of $M(v)$. In the second case, we have that either $\tilde{M}(u)=M(v)$ or $M_{low1}(u)=M(v)$ (since the $\mathit{low}$ point of $u$ is given by a back-edge in $B(v)$). (For an illustration, see figure \ref{figure:U(v)}.) Now the following lemma shows how we can determine $v$ from $u$.

\begin{figure}
\begin{center}
\centerline{\includegraphics[trim={0cm 22cm 0cm 0cm}, scale=1,clip, width=\textwidth]{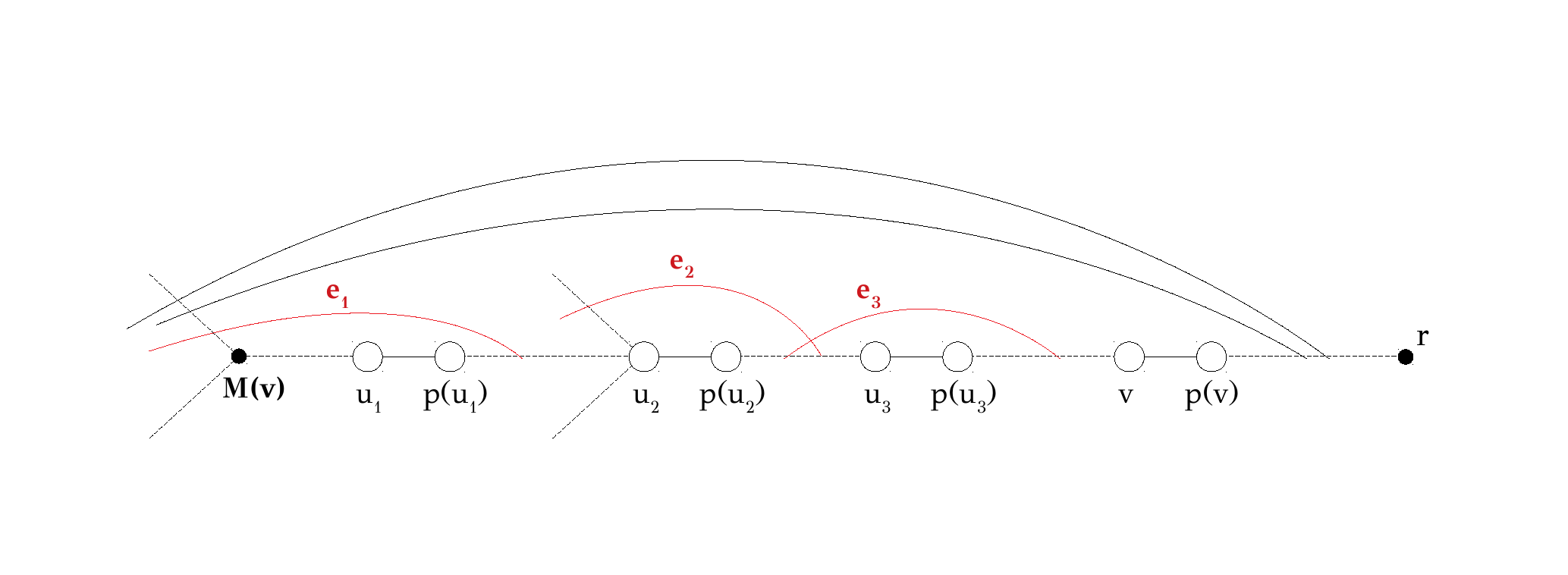}}
\caption{In this example we have $U(v)=\{u_1,u_2,u_3\}$, and every back-edge $e_i$ satisfies $B(u_i)=B(v)\sqcup\{e_i\}$. It should be clear that every $M(u_i)$ is an ancestor of $M(v)$, and $M(u_1)=M(v)$, $M_{low1}(u_2)=M(v)$ and $\tilde{M}(u_3)=M(v)$. It is perhaps worth noting that, for every vertex $v$, only one $u\in U(v)$ may have $M(u)=M(v)$ (that is, the one satisfying $\mathit{nextM}(u)=v$), but we may have many vertices $u\in V(v)$ with $\tilde{M}(u)=M(v)$ or $M_{low1}(u)=M(v)$.}
\label{figure:U(v)}
\end{center}
\end{figure}

\begin{lemma}
\label{lemma:find_v_from_u}
Let $u$ be a descendant of $v$ such that $M(u)=M(v)$ or $\tilde{M}(u)=M(v)$ or $M_{low1}(u)=M(v)$, and let $m$ $=$ $M(u)$ or $\tilde{M}(u)$ or $M_{low1}(u)$, depending on whether $M(u)=M(v)$ or $\tilde{M}(u)=M(v)$ or $M_{low1}(u)=M(v)$. Then $u\in U(v)$ if and only if $v$ is the greatest element in $M^{-1}(m)$ which is lower than $u$ and such that $\mathit{b\_count}(u)=\mathit{b\_count}(v)+1$.
\end{lemma}
\begin{proof}
($\Rightarrow$) $u\in U(v)$ means that there exists a back-edge $e$ such that $B(u)=B(v)\sqcup\{e\}$. Thus we get immediately that $\mathit{b\_count}(u)=\mathit{b\_count}(v)+1$. Now suppose, for the sake of contradiction, that there exists a $v'\in M^{-1}(m)$ which is greater than $v$ and lower than $u$. Let $(x,y)\in B(v')$. Then $y$ is a proper ancestor of $v'$, and therefore a proper ancestor of $u$. Furthermore, $x$ is a descendant of $M(v')$ ($=M(v)$), and so every one of the relations $M(u)=M(v)$, $\tilde{M}(u)=M(v)$ or $M_{low1}(u)=M(v)$ implies that $x$ is a descendant of $u$. This shows that $(x,y)\in B(u)$, and thus we have $B(v')\subseteq B(u)$. Now, since $M(v)=M(v')$ and $v'$ is a proper ancestor of $v$, we have $B(v)\subset B(v')$. Since $\mathit{b\_count}(u)=\mathit{b\_count}(v)+1$,  $B(v)\subset B(v')\subseteq B(u)$ implies that $B(u)=B(v')$, contradicting the fact that the graph is $3$-edge-connected.\\
($\Leftarrow$) Let $(x,y)\in B(v)$. Then $y$ is a proper ancestor of $v$, and therefore a proper ancestor of $u$. Furthermore, $x$ is a descendant of $M(v)$, and every one of the relations $M(u)=M(v)$, $\tilde{M}(u)=M(v)$ or $M_{low1}(u)=M(v)$ implies that $x$ is a descendant of $M(u)$. This shows that $(x,y)\in B(u)$. Thus we have $B(v)\subseteq B(u)$, and so $\mathit{b\_count}(u)=\mathit{b\_count}(v)+1$ implies that there exists a back-edge $e$ such $B(u)=B(v)\sqcup\{e\}$.
\end{proof}

Thus, for every vertex $u$, we have to check whether the greatest element $v$ in $M^{-1}(m)$ which is lower than $u$ satisfies $\mathit{b\_count}(u)=\mathit{b\_count}(v)+1$, for all $m\in\{M(u),\tilde{M}(u),M_{low1}(u)\}$. To do this efficiently, we process the vertices in a bottom-up fashion, and we keep in a variable $\mathit{currentVertex}[m]$ the lowest element of $M^{-1}(m)$ currently under consideration, so that we do not have to traverse the list $M^{-1}(m)$ from the beginning each time we process a vertex. Algorithm \ref{algorithm:B(u)=B(v)e} is an implementation of this procedure.

\begin{algorithm}
\caption{\textsf{Find all $3$-cuts $\{(u,p(u)),(v,p(v)),e)\}$, where $u$ is a descendant of $v$ and $B(u)=B(v)\sqcup\{e\}$, for a back-edge $e$.}}
\label{algorithm:B(u)=B(v)e}
\LinesNumbered
\DontPrintSemicolon
initialize an array $\mathit{currentVertex}$ with $n$ entries\;
\tcp{\textit{$m=M(v)$; just check whether the condition of Lemma \ref{lemma:find_v_from_u} is satisfied for $\mathit{nextM}(u)$}}
  \If{$\mathit{b\_count}(u)=\mathit{b\_count}(\mathit{nextM}(u))+1$}{
    mark the triplet $\{(u,p(u)),(\mathit{nextM}(u),p(\mathit{nextM}(u))),(\mathit{highD}(u),\mathit{high}(u))\}$
  }
\tcp{\textit{$m=\tilde{M}(u)$}}
\lForEach{vertex $x$}{$\mathit{currentVertex}[x] \leftarrow x$}
\For{$u\leftarrow n$ to $u=1$}{
  $m \leftarrow \tilde{M}(u)$\;
  \lIf{$m=\emptyset$}{\textbf{continue}}
  \tcp{\textit{find the greatest $v\in M^{-1}(m)$ which is lower than $u$}}
  $v \leftarrow \mathit{currentVertex}[m]$\;
  \lWhile{$v\neq\emptyset$ \textbf{and} $v\geq u$}{$v \leftarrow \mathit{nextM}(v)$}
  $\mathit{currentVertex}[m] \leftarrow v$\;
  \tcp{\textit{check the condition in Lemma \ref{lemma:find_v_from_u}}}
  \If{$\mathit{b\_count}(u)=\mathit{b\_count}(v)+1$}{
    mark the triplet $\{(u,p(u)),(v,p(v)),(\mathit{highD}(u),\mathit{high}(u))\}$
  }
}
\tcp{\textit{$m=M_{low1}(u)$}}
\lForEach{vertex $x$}{$\mathit{currentVertex}[x] \leftarrow x$}
\For{$u\leftarrow n$ to $u=1$}{
  $m \leftarrow M_{low1}(u)$\;
  \lIf{$m=\emptyset$}{\textbf{continue}}
  \tcp{\textit{find the greatest $v\in M^{-1}(m)$ which is lower than $u$}}
  $v \leftarrow \mathit{currentVertex}[m]$\;
  \lWhile{$v\neq\emptyset$ \textbf{and} $v\geq u$}{$v \leftarrow \mathit{nextM}(v)$}
  $\mathit{currentVertex}[m] \leftarrow v$\;
  \tcp{\textit{check the condition in Lemma \ref{lemma:find_v_from_u}}}
  \If{$\mathit{b\_count}(u)=\mathit{b\_count}(v)+1$}{
    mark the triplet $\{(u,p(u)),(v,p(v)),(\mathit{highD}(u),\mathit{high}(u))\}$
  }
}
\end{algorithm}

\subsection{Three tree-edges}
\label{section:three_tree_edges}

\begin{lemma}
Let $\{(u,p(u)),(v,p(v)),(w,p(w))\}$ be a $3$-cut, and assume, without loss of generality, that $w<\mathit{min}\{v,u\}$. Then $w$ is an ancestor of both $u$ and $v$.
\end{lemma}
\begin{proof}
Suppose that $w$ is neither an ancestor of $u$ nor an ancestor of $v$. Let $(x,y)\in B(w)$. Then $x$ is a descendant of $w$, and therefore it is not a descendant of either $u$ or $v$. In other words, $u,v\notin T[x,w]$. Furthermore, $y$ is a proper ancestor of $w$. Since neither $u$ nor $v$ is an ancestor of $w$ (since $w<\mathit{min}\{v,u\}$), we have that $u,v\notin T[w,r]$, and therefore $u,v\notin T[w,y]$. Thus, by removing the tree-edges $(u,p(u))$, $(v,p(v))$ and $(w,p(w))$, $w$ remains connected with $p(w)$ through the path $T[w,x],(x,y),T[y,p(w)]$, contradicting the fact that $\{(u,p(u)),(v,p(v)),(w,p(w))\}$ is a $3$-cut. This shows that $w$ is an ancestor of either $u$ or $v$ (or both). Suppose, for the sake of contradiction, that $w$ is not an ancestor of $u$. Then $w$ is an ancestor of $v$. This implies that $u$ is not a descendant of $v$ (for otherwise it would be a descendant of $w$). If $u$ is an ancestor of $v$, it must necessarily be an ancestor of $w$ (because $v\in T(w)$ and $u\notin T(w)$), but $w<u$ forbids this case. Thus, $u$ is not a ancestor of $v$. So far, then, we have that $u$ is not related as ancestor and descendant with either $w$ or $v$. Thus we may follow the same reasoning as above, to conclude that, by removing the tree-edges $(u,p(u))$, $(v,p(v))$ and $(w,p(w))$, $u$ remains connected with $p(u)$, again contradicting the fact that $\{(u,p(u)),(v,p(v)),(w,p(w))\}$ is a $3$-cut. This shows that $w$ is an ancestor of $u$. Using the same argument we can also prove that $w$ is an ancestor of $v$.
\end{proof}

At this point we distinguish two cases, depending on whether $u$ and $v$ are related as ancestor and descendant.

\subsubsection{$u$ and $v$ are not related as ancestor and descendant}
\label{section:wuv}

In what follows we will provide some characterizations of the $3$-cuts of the form $\{(u,p(u)),(v,p(v)),(w,p(w))\}$, where $w$ is an ancestor of $u$ and $v$, and $u,v$ are not related as ancestor and descendant. It will be useful to keep in mind the situation depicted in Figure \ref{figure:Muv}.

\begin{proposition}
\label{lemma:wuv_1}
Let $u$ and $v$ be two vertices which are not related as ancestor and descendant, and let $w$ be an ancestor of both $u$ and $v$. Then, $\{(u,p(u)),(v,p(v)),(w,p(w))\}$ is a $3$-cut if and only if $B(w)=B(u)\sqcup B(v)$.
\end{proposition}
\begin{proof}
($\Rightarrow$) Let $(x,y)\in B(w)$, and let's assume that $(x,y)\notin B(u)$. Since $y$ is a proper ancestor of $w$, and therefore a proper ancestor of $u$, from $(x,y)\notin B(u)$ we infer that $x$ is not a descendant of $u$. Suppose for the sake of contradiction that $x$ is not a descendant of $v$, either. This means that neither $u$ nor $v$ is in $T[x,w]$, and so, by removing the edges $(u,p(u))$, $(v,p(v))$ and $(w,p(w))$, $w$ remains connected with $p(w)$ through the path $T[w,x],(x,y),T[y,p(w)]$. This contradicts that fact that $\{(u,p(u)),(v,p(v)),(w,p(w))\}$ is a $3$-cut. Thus we have established that $x$ is a descendant of $v$. Since $y$ is also a proper ancestor of $v$, we have $(x,y)\in B(v)$. Thus we have shown that $B(w)\subseteq B(u)\cup B(v)$. Conversely, let $(x,y)\in B(u)\cup B(v)$, and assume, without loss of generality, that $(x,y)\in B(u)$. Then, $x$ is a descendant of $u$, and therefore a descendant of $w$. Now suppose, for the sake of contradiction, that $y$ is not a proper ancestor of $w$. Then we have $w\notin T[p(u),y)$, and since $w$ is not a descendant of $u$, we also have $w\notin T[x,u]$. Furthermore, since $u$ and $v$ are not related as ancestor and descendant, $v$ is not contained neither in $T[p(u),y)$ nor in $T[x,u]$. Thus, by removing the edges $(u,p(u))$, $(v,p(v))$ and $(w,p(w))$, $u$ remains connected with $p(u)$ through the path $T[u,x],(x,y),T[y,p(u)]$. This contradicts that fact that $\{(u,p(u)),(v,p(v)),(w,p(w))\}$ is a $3$-cut. Thus we have shown that $y$ is a proper ancestor of $w$, and so we have that $(x,y)\in B(w)$. Thus we have established that $B(u)\cup B(v)\subseteq B(w)$, and so we have $B(w)=B(u)\cup B(v)$. Since $u$ and $v$ are not related as ancestor and descendant, we have $B(u)\cap B(v)=\emptyset$. We conclude that $B(w)=B(u)\sqcup B(v)$.\\
($\Leftarrow$) Consider the sets of vertices $T(u)$, $T(v)$, $A=T(w)\setminus(T(u)\cup T(v))$ and $B=T(r)\setminus T(w)$. Since $u$ and $v$ are not related as ancestor and descendant, and $w$ is an ancestor of both $u$ and $v$, these sets are mutually disjoint. Now, since $B(u)\subset B(w)$, all back-edges that start from $T(u)$ end either in $T(u)$ or in $B$. Similarly, since $B(v)\subset B(w)$, all back-edges that start from $T(v)$ end either in $T(v)$ or in $B$. Furthermore, a back-edge that starts from $A$ cannot reach $B$ and must necessarily end in $A$, since it starts from a descendant of $w$, but not from a descendant of either $u$ or $v$ (while we have $B(w)=B(u)\sqcup B(v)$). Thus, by removing from the graph the tree-edges $(u,p(u))$, $(v,p(v))$ and $(w,p(w))$, the graph becomes separated into two parts: $T(u)\cup T(v)\cup B$ and $A$.
\end{proof}

\begin{figure}
\begin{center}
\centerline{\includegraphics[trim={0cm 18cm 0cm 0cm}, scale=1,clip, width=\textwidth]{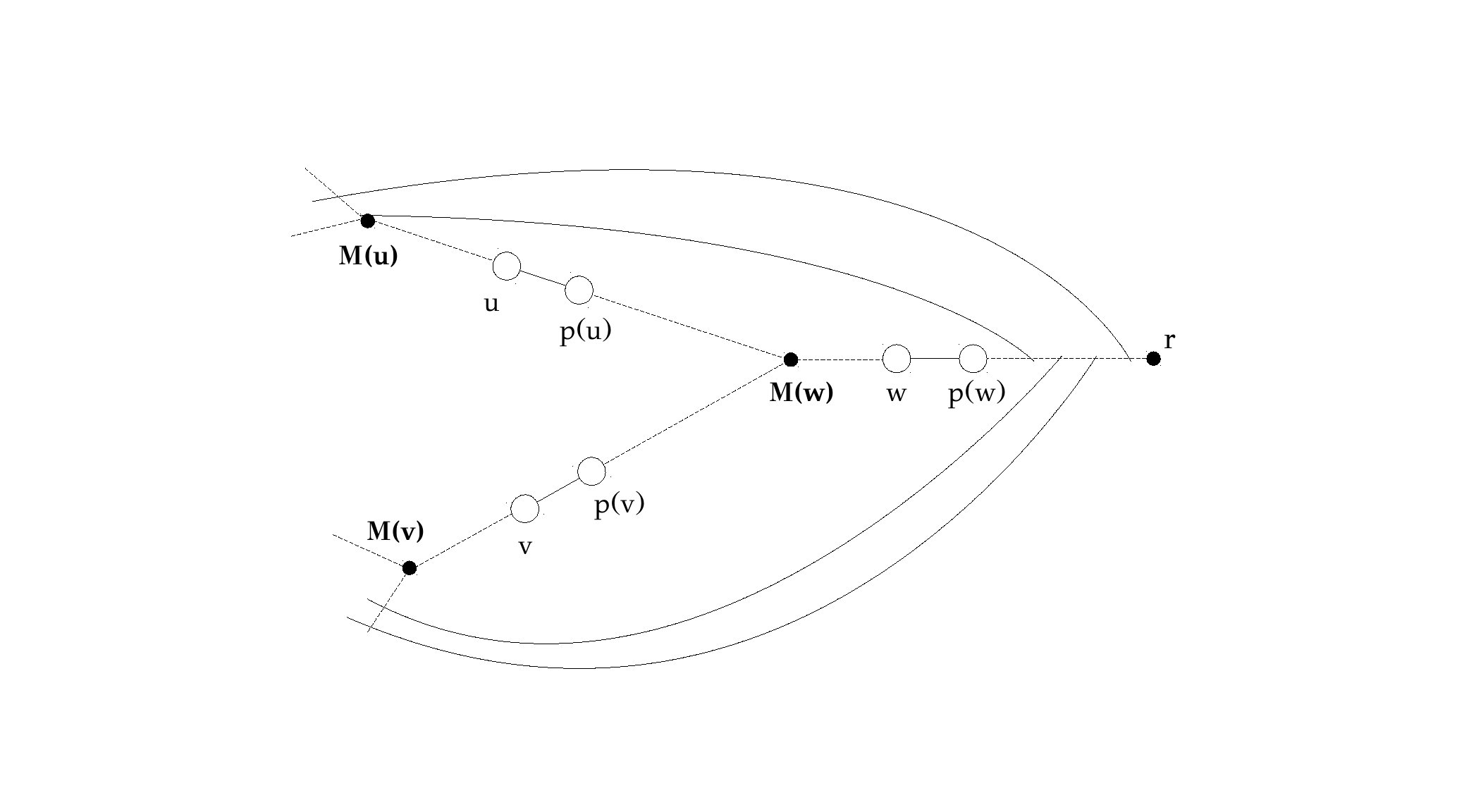}}
\caption{In this example we have $B(w)=B(u)\sqcup B(v)$. Observe that $M_{low1}(w)=M(u)$ and $M_{low2}(w)=M(v)$. Furthermore, $\mathit{high}(u)<w$ and $\mathit{high}(v)<w$. Also, if there is another vertex $u'$ with $M(u')=M(u)$, it must either be a descendant of $u$ or an ancestor of $w$. Thus, $u$ is the lowest vertex in $M^{-1}(M_{low1}(w))$ which is greater than $w$. Similarly, $v$ is the lowest vertex in $M^{-1}(M_{low2}(w))$ which is greater than $w$. By Lemmata \ref{lemma:wuv} and \ref{lemma:wuv_3}, these properties (together with $\mathit{b\_count}(w)=\mathit{b\_count}(u)+\mathit{b\_count}(v)$) are sufficient to establish $B(w)=B(u)\sqcup B(v)$. Notice also that, if we remove the tree-edges $(u,p(u))$, $(v,p(v))$ and $(w,p(w))$, the graph becomes disconnected into two components: $T(u)\cup T(v)\cup (T(r)\setminus T(w))$ and $T(w)\setminus (T(u)\cup T(v))$. (See also the ``$\Leftarrow$" part of the proof of proposition \ref{lemma:wuv_1}.)}
\label{figure:Muv}
\end{center}
\end{figure}

\begin{lemma}
\label{lemma:wuv}
Let $u$ and $v$ be two vertices which are not related as ancestor and descendant, and let $w$ be an ancestor of both $u$ and $v$. Then $B(w)=B(u)\sqcup B(v)$ if and only if: $M_{low1}(w)=M(u)$ and $M_{low2}(w)=M(v)$ (or $M_{low1}(w)=M(v)$ and $M_{low2}(w)=M(u)$), and $\mathit{high}(u)<w$, $\mathit{high}(v)<w$, and $\mathit{b\_count}(w)=\mathit{b\_count}(u)+\mathit{b\_count}(v)$.
\end{lemma}
\begin{proof}
($\Rightarrow$) $\mathit{b\_count}(w)=\mathit{b\_count}(u)+\mathit{b\_count}(v)$ is an immediate consequence of $B(w)=B(u)\sqcup B(v)$. Furthermore, since every $(x,y)\in B(u)$ is also in $B(w)$, it has $y<w$, and so $\mathit{high}(u)<w$. With the same reasoning, we also get $\mathit{high}(v)<w$. Now, since $B(w)=B(u)\sqcup B(v)$, we have that $M(w)$ is an ancestor of both $M(u)$ and $M(v)$. Since $u$ and $v$ are not related as ancestor and descendant, $M(u)$ and $M(v)$ are not related as ancestor or descendant, either. This implies that they are both proper descendants of $M(w)$. Now, suppose, for the sake of contradiction, that $M(u)$ and $M(v)$ are descendants of the same child $c$ of $M(w)$. Then there must exist a back-edge $(x,y)\in B(w)$ such that $x=M(w)$ or $x$ is a descendant of a child of $M(w)$ different from $c$. (Otherwise, $M(w)$ would be a descendant of $c$, which is absurd.) But this contradicts the fact that $B(w)=B(u)\sqcup B(v)$, since $(x,y)$ does not belong neither in $B(u)$ nor in $B(v)$. Thus, $M(u)$ and $M(v)$ are descendants of different children of $M(w)$. Furthermore, since every back-edge $(x,y)\in B(w)$ has $x$ in $T(u)$ or $T(v)$, there are no other children of $M(w)$ from whose subtrees begin back-edges that end in a proper ancestor of $w$. Thus, one of $M(u)$ and $M(v)$ is a descendant of the $\mathit{low1}$ child of $M(w)$, and the other is a descendant of the $\mathit{low2}$ child of $M(w)$. We may assume, without loss of generality, that $M(u)$ is a descendant of the $\mathit{low1}$ child of $M(w)$, and $M(v)$ is a descendant of the $\mathit{low2}$ child of $M(w)$. Since $B(u)\subset B(w)$, we have that $M(u)$ is a descendant of $M_{low1}(w)$. Furthermore, since $B(w)=B(u)\sqcup B(v)$ and $M(v)$ is not a descendant of the $\mathit{low1}$ child of $M(w)$, there are no back-edges $(x,y)$ with $x$ a descendant of the $\mathit{low1}$ child of $M(w)$ and $y$ a proper ancestor of $w$ apart from those contained in $B(u)$. Thus, $M(u)$ is an ancestor of $M_{low1}(w)$, and $M_{low1}(w)=M(u)$ is established. With the same reasoning, we also get $M_{low2}(w)=M(v)$. \\
($\Leftarrow$) Let $(x,y)\in B(u)$. Then $x$ is a descendant of $u$, and therefore a descendant of $w$. Furthermore, since $\mathit{high}(u)<w$, we have $y<w$, and therefore $y$ is a proper ancestor of $w$. This shows that $(x,y)\in B(w)$, and thus $B(u)\subseteq B(w)$. With the same reasoning, we also get $B(v)\subseteq B(w)$. Thus we have $B(u)\cup B(v)\subseteq B(w)$. Since $u$ and $v$ are not related as ancestor and descendant, we have $B(u)\cap B(v)=\emptyset$. From $B(u)\cup B(v)\subseteq B(w)$, $B(u)\cap B(v)=\emptyset$, and $\mathit{b\_count}(w)=\mathit{b\_count}(u)+\mathit{b\_count}(v)$, we conclude that $B(w)=B(u)\sqcup B(v)$.
\end{proof}

The following lemma shows, that, for every vertex $w$, there is at most one pair $u,v$ of descendants of $w$ which are not related as ancestor and descendant and are such that $\{(u,p(u)),(v,p(v)),(w,p(w))\}$ is a $3$-cut. Thus, the number of $3$-cuts of this type is $O(n)$. Furthermore, it allows us to compute $u$ and $v$ (if such a pair of $u$ and $v$ exists).

\begin{lemma}
\label{lemma:wuv_3}
Let $\{(u,p(u)),(v,p(v)),(w,p(w))\}$ be a $3$-cut such that $u$ and $v$ are not related as ancestor and descendant and let $w$ is an ancestor of both $u$ and $v$. Assume w.l.o.g. that $M_{low1}(w)=M(u)$ and $M_{low2}(w)=M(v)$, and let $m_1=M_{low1}(w)$ and $m_2=M_{low2}(w)$. Then $u$ is the lowest vertex in $M^{-1}(m_1)$ which is greater than $w$, and $v$ is the lowest vertex in $M^{-1}(m_2)$ which is greater that $w$.
\end{lemma}
\begin{proof}
By Proposition \ref{lemma:wuv_1}, we have that $B(w)=B(u)\sqcup B(v)$. Now, suppose that there exists a $u'\in M^{-1}(m_1)$ which is lower than $u$ and greater than $w$. Then, $M(u')=M(u)$ implies that $B(u')\subset B(u)$, and so there is a back-edge $(x,y)\in B(u)\setminus B(u')$. This means that $y$ is not a proper ancestor of $u'$, and therefore not a proper ancestor of $w$, either. But this implies that $(x,y)\notin B(w)$, contradicting the fact that $B(u)\subset B(w)$. A similar argument shows that there does not exist a $v'\in M^{-1}(m_2)$ which is lower than $v$ and greater than $w$.
\end{proof}

Thus we only have to find, for every vertex $w$, the lowest element $u$ of $M^{-1}(M_{low1}(w))$ which is greater than $w$, and the lowest element $v$ of $M^{-1}(M_{low2}(w))$ which is greater than $w$, and check the condition in Lemma \ref{lemma:wuv} - i.e., whether $\mathit{high}(u)<w$, $\mathit{high}(v)<w$, and $\mathit{b\_count}(w)=\mathit{b\_count}(u)+\mathit{b\_count}(v)$. To do this efficiently, we process the vertices in a bottom-up fashion, and we keep in a variable
$\mathit{currentVertex}[x]$ the lowest element of $M^{-1}(x)$ currently under consideration. Thus, we do not need to traverse the list $M^{-1}(x)$ from the beginning each time we process a vertex. Algorithm \ref{algorithm:wuv} is an implementation of this procedure.

\begin{algorithm}
\caption{\textsf{Find all $3$-cuts $\{(u,p(u)),(v,p(v)),(w,p(w))\}$, where $w$ is an ancestor of $u$ and $v$, and $u,v$ are not related as ancestor and descendant}}
\label{algorithm:wuv}
\LinesNumbered
\DontPrintSemicolon
initialize an array $\mathit{currentVertex}$ with $n$ entries\;
\lForEach{vertex $x$}{$\mathit{currentVertex}[x] \leftarrow x$}
\For{$w\leftarrow n$ to $w=1$}{
  $m_1 \leftarrow M_{low1}(w)$, $m_2 \leftarrow M_{low2}(w)$\;

  \lIf{$m_1=\emptyset$ \textbf{or} $m_2=\emptyset$}{\textbf{continue}}

  \tcp{\textit{find the lowest $u$ in $M^{-1}(m_1)$ which is greater than $w$}}
  $u \leftarrow \mathit{currentVertex}[m_1]$\;
  \lWhile{$\mathit{nextM}(u)\neq\emptyset$ \textbf{and} $\mathit{nextM}(u)> w$}{$u \leftarrow \mathit{nextM}(u)$}
  $\mathit{currentVertex}[m_1] \leftarrow u$\;

  \tcp{\textit{find the lowest $v$ in $M^{-1}(m_2)$ which is greater than $w$}}
  $v \leftarrow \mathit{currentVertex}[m_2]$\;
  \lWhile{$\mathit{nextM}(v)\neq\emptyset$ \textbf{and} $\mathit{nextM}(v)> w$}{$v \leftarrow \mathit{nextM}(v)$}
  $\mathit{currentVertex}[m_2] \leftarrow v$\;

  \tcp{\textit{check the condition in Lemma \ref{lemma:wuv}}}
  \If{$\mathit{b\_count}(w)=\mathit{b\_count}(u)+\mathit{b\_count}(v)$ \textbf{and} $\mathit{high}(u)<w$ \textbf{and} $\mathit{high}(v)<w$}{
     mark the triplet $\{(u,p(u)),(v,p(v)),(w,p(w))\}$\;
  }
}
\end{algorithm}

\subsubsection{$u$ and $v$ are related as ancestor and descendant}
\label{section:uvw}

Throughout this section it will be useful to keep in mind the situation depicted in Figure \ref{figure:uvw}.

\begin{proposition}
\label{proposition:uvw}
Let $u,v,w$ be three vertices such that $u$ is a descendant of $v$ and $v$ is a descendant of $w$. Then $\{(u,p(u)),(v,p(v)),(w,p(w))\}$ is a $3$-cut if and only if $B(v)=B(u)\sqcup B(w)$.
\end{proposition}
\begin{proof}
($\Rightarrow$) Let $(x,y)\in B(v)$, and assume that $(x,y)\notin B(u)$. $(x,y)\in B(v)$ implies that $y$ is a proper ancestor of $v$, and therefore a proper ancestor of $u$. Thus, $(x,y)\notin B(u)$ implies that $x$ is not a descendant of $u$. Furthermore, $(x,y)\in B(v)$ implies that $x$ is a descendant of $v$, and therefore a descendant of $w$. Now suppose, for the sake of contradiction, that $y$ is not a proper ancestor of $w$. Then, $w\notin T[p(v),y)$. Now we see that, by removing the edges $(u,p(u))$, $(v,p(v))$ and $(w,p(w))$ from the graph, $v$ remains connected with $p(v)$ through the path $T[v,x],(x,y),T[y,p(v)]$ (since $u,w \notin T[v,x]\cup T[p(v),y]$). This contradicts the fact that $\{(u,p(u)),(v,p(v)),(w,p(w))\}$ is a $3$-cut. Therefore, $y$ is a proper ancestor of $w$, and thus $(x,y)\in B(w)$. Thus far we have established that $B(v)\subseteq B(u)\cup B(w)$. Now let $(x,y)\in B(u)$. Then $x$ is a descendant of $u$, and therefore a descendant of $v$. Suppose, for the sake of contradiction, that $y$ is not a proper ancestor of $v$. Then, $v\notin T[p(u),y)$. Now we see that, by removing the edges $(u,p(u))$, $(v,p(v))$ and $(w,p(w))$ from the graph, $u$ remains connected with $p(u)$ through the path $T[u,x],(x,y),T[y,p(u)]$. This contradicts the fact that $\{(u,p(u)),(v,p(v)),(w,p(w))\}$ is a $3$-cut. Therefore, $y$ is a proper ancestor of $v$, and thus $(x,y)\in B(v)$. This shows that $B(u)\subseteq B(v)$. Now let $(x,y)\in B(w)$. Then $y$ is a proper ancestor of $w$, and therefore a proper ancestor of $v$. Suppose, for the sake of contradiction, that $x$ is not a descendant of $v$. Then $x$ is not a descendant of $u$, either, and so $u,v\notin T[x,w]$. Thus we see that, by removing the edges $(u,p(u))$, $(v,p(v))$ and $(w,p(w))$ from the graph, $w$ remains connected with $p(w)$ through the path $T[w,x],(x,y),T[y,p(w)]$. This contradicts the fact that $\{(u,p(u)),(v,p(v)),(w,p(w))\}$ is a $3$-cut. Therefore, $x$ is a descendant of $v$, and thus $(x,y)\in B(v)$. This shows that $B(w)\subseteq B(v)$. Thus we have established that $B(u)\cup B(w)\subseteq B(v)$, and so we have $B(v)=B(u)\cup B(w)$.

Now suppose, for the sake of contradiction, that there is a back-edge $(x,y)\in B(u)\cap B(w)$. Since $B(u)\neq B(w)$ (for otherwise $u=w$), there must exist a back-edge $(x',y')$ in $B(u)\setminus B(w)$ or in $B(w)\setminus B(u)$. Take the first case, first. Then, since $B(u)\subseteq B(v)$, $y'$ is a proper ancestor of $v$. But since $(x',y')\notin B(w)$, $y'$ cannot be a proper ancestor of $w$. Let $P$ be a path connecting $x'$ with $x$ in $T(u)$. Then, by removing the tree-edges $(u,p(u))$, $(v,p(v))$ and $(w,p(w))$, $w$ remains connected with $p(w)$ through the path $T[w,y'],(x',y'),P,(x,y),T[y,p(w)]$, which contradicts the assumption that $\{(u,p(u)),(v,p(v)),(w,p(w))\}$ is a $3$-cut. Now take the case $\exists (x',y')\in B(w)\setminus B(u)$. Then, since $B(w)\subseteq B(v)$, $x'$ is a descendant of $v$. But since $(x',y')\notin B(u)$, $x'$ cannot be a descendant of $u$. Let $P$ be a path connecting $y$ with $y'$ in $T(r)\setminus T(w)$, and $Q$ be a path connecting $x'$ with $p(u)$ in $T(v)\setminus T(u)$. Then, by removing the tree-edges $(u,p(u))$, $(v,p(v))$ and $(w,p(w))$, $u$ remains connected with $p(u)$ through the path $T[u,x],(x,y),P,(y',x'),Q$, which contradicts the assumption that $\{(u,p(u)),(v,p(v)),(w,p(w))\}$ is a $3$-cut. This shows that $B(u)\cap B(w)=\emptyset$. We conclude that $B(v)=B(u)\sqcup B(w)$.\\
($\Leftarrow$) Consider the sets of vertices $A=T(u)$, $B=T(v)\setminus T(u)$, $C=T(w)\setminus T(v)$ and $D=T(r)\setminus T(w)$. Since $u$ is a descendant of $v$ and $v$ is a descendant of $w$, these sets are mutually disjoint. Now, since $B(u)\subset B(v)$ and $B(u)\cap B(w)=\emptyset$, every back-edge that starts from $A$ ends either in $A$ or in $T(v,w]$, and thus in $C$. Furthermore, every back-edge that starts from $B$ and does not end in $B$, is a back-edge that starts from $T(v)$, but not from $T(u)$, and ends in a proper ancestor of $v$; thus, since $B(v)=B(u)\sqcup B(w)$, it ends in $T(w,r]$, and thus in $D$. Finally, every back-edge that starts from $C$ must end in $C$, since $B(w)\subset B(v)$. Thus we see, that, by removing from the graph the tree-edges $(u,p(u))$, $(v,p(v))$ and $(w,p(w))$, the graph becomes separated into two parts: $A\cup C$ and $B\cup D$.
\end{proof}

\begin{figure}
\begin{center}
\centerline{\includegraphics[trim={0cm 23cm 0cm 0cm}, scale=1,clip, width=\textwidth]{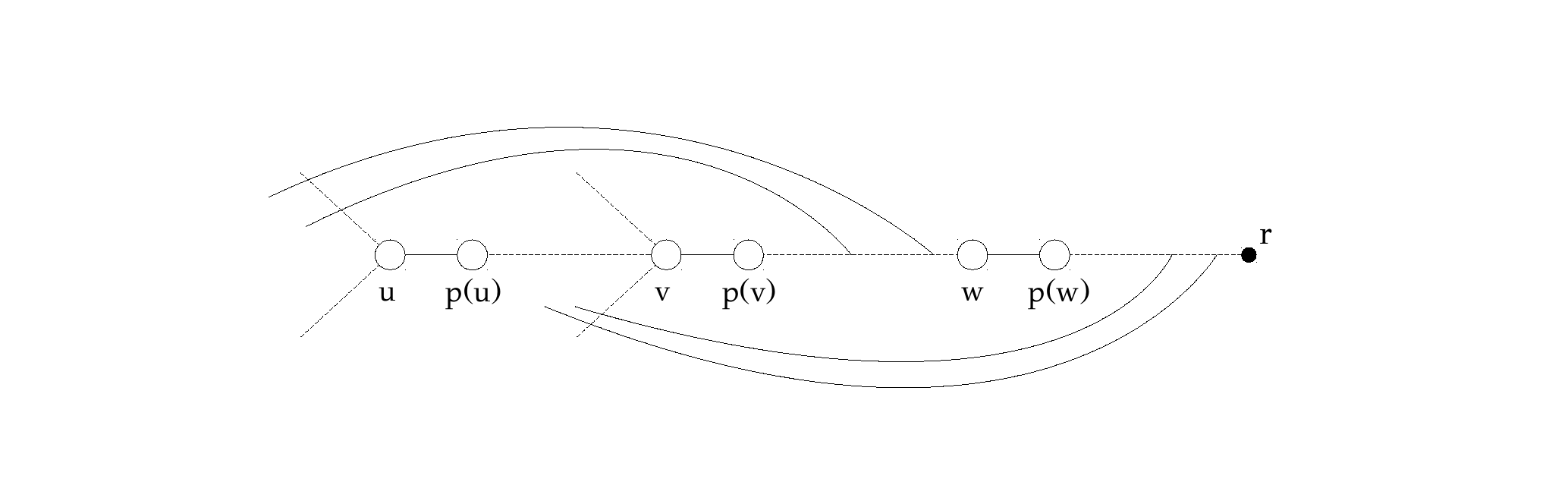}}
\caption{In this example we have $B(v)=B(u)\sqcup B(w)$. By removing the tree-edges $(u,p(u))$, $(v,p(v))$ and $(w,p(w))$, the graph becomes disconnected into two components: $T(u)\cup(T(w)\setminus T(v))$ and $(T(v)\setminus T(u))\cup(T(r)\setminus T(w))$. (See also the ``$\Leftarrow$" part of the proof of proposition \ref{proposition:uvw}.)}
\label{figure:uvw}
\end{center}
\end{figure}

\begin{corollary}
\label{unique_edge}
If $(u,p(u))$, $(v,p(v))$ are two tree-edges, there is at most one $w$ such that $\{(u,p(u)),(v,p(v)),(w,p(w))\}$ is a $3$-cut.
\end{corollary}
\begin{proof}
This is a consequence of propositions \ref{lemma:wuv_1} and \ref{proposition:uvw}.
\end{proof}

Here we distinguish two cases, depending on whether $M(v)=M(w)$ or $M(v)\neq M(w)$.

\paragraph{$M(v)\neq M(w)$}

\begin{lemma}
\label{lemma:M(v)neqM(w)}

Let $u$ be a descendant of $v$ and $v$ a descendant of $w$, and $M(v)\neq M(w)$. Then, $\{(u,p(u)),(v,p(v)),(w,p(w))\}$ is a $3$-cut if and only if: $M(w)=M_{low1}(v)$ and $w$ is the greatest vertex with $M(w)=M_{low1}(v)$ which is lower than $v$, $M(u)=M_{low2}(v)$ and $u$ is the lowest vertex with $M(u)=M_{low2}(v)$, $\mathit{high}(u)<v$ and $\mathit{b\_count}(v)=\mathit{b\_count}(u)+\mathit{b\_count}(w)$. (See Figure \ref{figure:M(v)M(w)}.)
\end{lemma}
\begin{proof}
($\Rightarrow$) By proposition \ref{proposition:uvw}, we have $B(v)=B(u)\sqcup B(w)$. This immediately establishes both $\mathit{high}(u)<v$ and $\mathit{b\_count}(v)=\mathit{b\_count}(u)+\mathit{b\_count}(w)$. Now, since $B(v)=B(u)\sqcup B(w)$, both $M(u)$ and $M(w)$ are descendants of $M(v)$. We will show that $M(u)$ and $M(w)$ are not related as ancestor and descendant. First, suppose that $M(u)$ is an ancestor of $M(w)$. Now let $(x,y)\in B(w)$. Then $x$ is a descendant of $M(w)$, and therefore a descendant of $M(u)$. Furthermore, $y$ is a proper ancestor of $w$, and therefore a proper ancestor of $u$. This shows that $(x,y)\in B(u)$, contradicting the fact that $B(u)\cap B(w)=\emptyset$. Now suppose that $M(w)$ is an ancestor of $M(u)$. Let $(x,y)\in B(v)$. Since $B(v)=B(u)\sqcup B(w)$, $x$ is a descendant of either $M(u)$ or $M(w)$. In either case, $x$ is a descendant of $w$. Due to the generality of $(x,y)$, this shows that $M(v)$ is a descendant of $M(w)$. Since $M(w)$ is also a descendant of $M(v)$, we get $M(w)=M(v)$, contradicting $M(w)\neq M(v)$. Thus we have established that $M(u)$ and $M(w)$ are not related as ancestor and descendant. Since $M(u)$ and $M(v)$ are descendants of $M(v)$, they must be proper descendants of $M(v)$. Now we will show that $M(u)$ and $M(w)$ are descendants of different children of $M(v)$. Suppose, for the sake of contradiction, that $M(u)$ and $M(w)$ are descendants of the same child $c$ of $M(v)$. Then, there must exist a back-edge $(x,y)\in B(v)$ such that $x=M(v)$ or $x$ is a descendant of a child of $M(v)$ different from $c$. (Otherwise, we would have that $M(v)$ is a descendant of $c$, which is absurd.) But this means that $(x,y)$ is neither in $B(u)$ nor in $B(w)$, contradicting the fact that $B(v)=B(u)\sqcup B(w)$. Thus, one of $M(u)$ and $M(w)$ is a descendant of the $\mathit{low1}$ child of $M(v)$, and the other is a descendant of the $\mathit{low2}$ child of $M(v)$. Observe that there does not exist a back-edge $(x,y)\in B(u)$ such that $y=\mathit{low}(v)$, for this would imply that $(x,y)\in B(w)$ (since $u$ is a descendant of $w$), and $B(u)$ does not meet $B(w)$. Thus, since $B(v)=B(u)\sqcup B(w)$, $v$ gets its $\mathit{low}$ point from $B(w)$. This shows that $M(w)$ is a descendant of the $\mathit{low1}$ child of $M(v)$ and $M(u)$ is a descendant of the $\mathit{low2}$ child of $M(v)$. Since $B(w)\subset B(v)$, we have that $M(w)$ is a descendant of $M_{low1}(v)$. Furthermore, since $B(v)=B(u)\sqcup B(w)$ and $M(u)$ is not a descendant of the $\mathit{low1}$ child of $M(v)$, there are no back-edges $(x,y)$ with $x$ a descendant of the $\mathit{low1}$ child of $M(v)$ and $y$ a proper ancestor of $v$ apart from those contained in $B(w)$. Thus, $M(w)$ is an ancestor of $M_{low1}(v)$, and $M_{low1}(v)=M(w)$ is established. With the same reasoning, we also get $M_{low2}(v)=M(u)$.

Now suppose, for the sake of contradiction, that there exists a vertex $w'$ with $M(w')=M(w)$ and $v>w'>w$. This implies that $B(w)\subset B(w')$, and thus there is a back-edge $(x,y)\in B(w')\setminus B(w)$. Then $x$ is a descendant of $M(w')$, and therefore a descendant of $M_{low1}(v)$. Furthermore, $y$ is a proper ancestor of $w'$, and therefore a proper ancestor of $v$. This shows that $(x,y)\in B(v)$, and therefore, since $B(v)=B(u)\sqcup B(w)$ and $(x,y)\notin B(w)$, we have $(x,y)\in B(u)$. But $x$ is not a descendant of $M(u)$, since it is a descendant of $M(w)$ which is not related as ancestor or descendant with $M(u)$. That's a contradiction. Thus we have established that $w$ is the greatest vertex with $M(w)=M_{low1}(v)$ which is lower than $v$. Finally, suppose for the sake of contradiction that there exists a vertex $u'$ with $M(u')=M(u)$ and $u'<u$. This implies that $B(u')\subset B(u)$, and therefore there exists a back-edge $(x,y)\in B(u)\setminus B(u')$. Then, $y$ is a proper ancestor of $u$ and a descendant of $u'$. Since $\mathit{high}(u)<v$, we have $y<v$, and therefore $u'$ is an ancestor of $v$. Now suppose that $u'$ is an ancestor of $w$. Let $(x',y')\in B(u')$. Then $x'$ is a descendant of $M(u')$, and therefore a descendant of $M(u)$, and therefore a descendant of $u$, and therefore a descendant of $w$. Furthermore, $y'$ is a proper ancestor of $u'$, and therefore a proper ancestor of $w$. This shows that $(x',y')\in B(w)$. But this cannot be the case, since $(x',y')\in B(u') \subset B(u)$ and $B(u)\cap B(w)=\emptyset$. Thus, $u'$ is a descendant of $w$. Since $u'$ is an ancestor of $v$, it is also an ancestor of $M_{low1}(v)=M(w)$. Thus, Lemma \ref{lemma:Muv} implies that $M(u')$ is an ancestor of $M(w)$. But, since $M(u')=M(u)$, this contradicts the fact that $M(u)$ and $M(w)$ are not related as ancestor and descendant. Thus we have established that $u$ is the lowest vertex with $M(u)=M_{low2}(v)$.\\
($\Leftarrow$) By proposition \ref{proposition:uvw}, it is sufficient to prove that $B(v)=B(u)\sqcup B(w)$. First, let $(x,y)\in B(u)$. Then $x$ is a descendant of $u$, and therefore a descendant of $v$. Furthermore, $y\leq\mathit{high}(u)<v$ implies that $y$ is a proper ancestor of $v$. This shows that $B(u)\subseteq B(v)$. Now let $(x,y)\in B(w)$. Then $y$ is a proper ancestor of $w$, and therefore a proper ancestor of $v$. Since $M(w)=M_{low1}(v)$, we have that $x$ is a descendant of $v$. This shows that $B(w)\subseteq B(v)$. Thus we have $B(u)\cup B(w)\subseteq B(v)$. Since $M(u)$ and $M(w)$ are not related as ancestor and descendant (for they are descendants of different children of $M(v)$), we have that $B(u)\cap B(w)=\emptyset$. In conjunction with $\mathit{b\_count}(v)=\mathit{b\_count}(u)+\mathit{b\_count}(w)$, from $B(u)\cup B(w)\subseteq B(v)$ and $B(u)\cap B(w)=\emptyset$ we conclude that $B(u)\sqcup B(w)=B(v)$.
\end{proof}

\begin{figure}
\begin{center}
\centerline{\includegraphics[trim={0cm 23cm 0cm 0cm}, scale=1,clip, width=\textwidth]{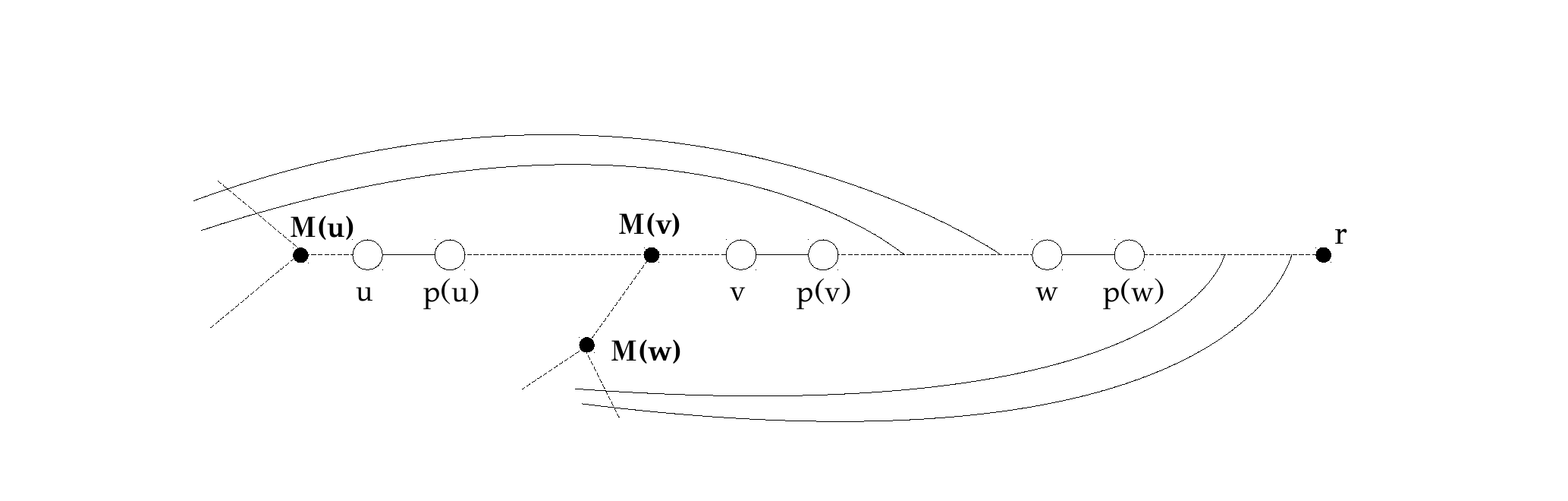}}
\caption{In this example we have $B(v)=B(u)\sqcup B(w)$. Observe that $M_{low1}(v)=M(w)$ and $M_{low2}(v)=M(u)$. $u$ is the last vertex in $M^{-1}(M(u))$, and $w$ is the greatest vertex in $M^{-1}(M(w))$ which is lower than $v$.}
\label{figure:M(v)M(w)}
\end{center}
\end{figure}

This lemma shows that, for every vertex $v$, there is at most one pair of vertices $u,w$, where $u$ is a descendant of $v$, $w$ is an ancestor of $v$, $M(v)\neq M(w)$, and $\{(u,p(u)),(v,p(v)),(w,p(w))\}$ is a $3$-cut. In particular, we have that $w$ is the greatest vertex with $M(w)=M_{low1}(v)$ which is lower than $v$, $u$ is the last vertex in $M^{-1}(M_{low2}(v))$, $\mathit{high}(u)<v$ and $\mathit{b\_count}(v)=\mathit{b\_count}(u)+\mathit{b\_count}(w)$. Thus, Algorithm \ref{algorithm:M(v)neqM(w)} shows how we can compute all $3$-cuts of this type. We only have to make sure that we can compute $w$ without having to traverse the list $M^{-1}(M_{low1}(v))$ from the beginning, each time we process a vertex $v$. To achieve this, we process the vertices in a bottom-up fashion, and we keep in an array $\mathit{currentM}[x]$ the current element of $M^{-1}(x)$ under consideration, so that we do not need to traverse the list $M^{-1}(x)$ from the beginning each time we process a vertex.

\begin{algorithm}[!h]
\caption{\textsf{Find all $3$-cuts $\{(u,p(u)),(v,p(v)),(w,p(w))\}$, where $u$ is a descendant of $v$, $v$ is a descendant of $w$, and $M(v)\neq M(w)$.}}
\label{algorithm:M(v)neqM(w)}
\LinesNumbered
\DontPrintSemicolon
\lForEach{vertex $v$}{$\mathit{currentVertex}[v] \leftarrow v$}
\For{$v\leftarrow n$ to $v=1$}{
  $m_1 \leftarrow M_{low1}(v)$, $m_2 \leftarrow M_{low2}(v)$\;

  \lIf{$m_1=\emptyset$ \textbf{or} $m_2=\emptyset$}{\textbf{continue}}

  \tcp{\textit{find the greatest $w$ in $M^{-1}(m_1)$ which is lower than $v$}}
  $w \leftarrow \mathit{currentVertex}(m_1)$\;
  \lWhile{$w\neq\emptyset$ \textbf{and} $w\geq v$}{$w \leftarrow \mathit{nextM}(w)$}
  $\mathit{currentVertex}[m_1] \leftarrow w$\;

  \tcp{\textit{$u$ is the last element of $M^{-1}(m_2)$}}
  $u \leftarrow \mathit{lastM}(m_2)$

  \tcp{\textit{check the condition in Lemma \ref{lemma:M(v)neqM(w)}}}
  \If{$w\neq\emptyset$ and $\mathit{high}(u)<v$ and $\mathit{b\_count}(v)=\mathit{b\_count}(u)+\mathit{b\_count}(w)$}{
     mark the triplet $\{(u,p(u)),(v,p(v)),(w,p(w))\}$\;
  }
}
\end{algorithm}


\paragraph{$M(v)=M(w)$}

Let $w$ be a proper ancestor of $v$ such that $M(v)=M(w)$. By corollary \ref{unique_edge}, there is at most one descendant $u$ of $v$ such that $\{(u,p(u)),(v,p(v)),(w,p(w))\}$ is a $3$-cut. In order to find this $u$ (if it exists), we distinguish two cases, depending on whether $w=\mathit{nextM}(v)$ or $w\neq\mathit{nextM}(v)$. In any case, we will need the following lemma, which gives a necessary condition for the existence of $u$.

\begin{lemma}
\label{lemma:u_same_high_with_v}
Let $u,v,w$ be three vertices such that $u$ is a descendant of $v$, $v$ is a descendant of $w$, and $M(v)=M(w)$. Then, $B(v)=B(u)\sqcup B(w)$ only if $\mathit{high}(u)=\mathit{high}(v)$ and $\mathit{nextM}(u)=\emptyset$.
\end{lemma}
\begin{proof}
Let $(x,y)\in B(u)$ be such that $y=\mathit{high}(u)$. Then, since $B(v)=B(u)\sqcup B(w)$, we have $(x,y)\in B(v)$, and so $y\leq\mathit{high}(v)$. Suppose for the sake of contradiction that $y\neq\mathit{high}(v)$. Then, since $B(v)=B(u)\sqcup B(w)$,  there exists a $(x',y')\in B(w)$ such that $y'=\mathit{high}(v)$. Furthermore, since $y\neq\mathit{high}(v)$ and $(x,y)\in B(v)$, we have $y'>y$, which means that $y$ is a proper ancestor of $w$. But then, since $x$ is a descendant of $u$, it is also a descendant of $w$, and thus $(x,y)\in B(w)$, contradicting the fact that $B(u)\cap B(w)=\emptyset$. Thus we have shown that $\mathit{high}(u)=\mathit{high}(v)$.

Now suppose, for the sake of contradiction, that there exists a $u'$ which is a proper ancestor of $u$ with $M(u')=M(u)$. Then we have $B(u')\subset B(u)$. Now suppose, for the sake of contradiction, that $u'$ is an ancestor of $v$. Suppose that $u'$ is an ancestor of $w$. Let $(x,y)\in B(u')$. Then $x$ is a descendant of $M(u')$, and therefore a descendant of $M(u)$, and therefore a descendant of $u$, and therefore a descendant of $w$. Furthermore, $y$ is a proper ancestor of $u'$, and therefore a proper ancestor of $w$. This means that $(x,y)\in B(w)$, and thus we have $B(u')\subseteq B(w)$. But this contradicts $B(u)\cap B(w)=\emptyset$, since $B(u')\subset B(u)$. Thus, we have that $u'$ is a descendant of $w$. Then, since $u'$ is an ancestor of $v$, it is also an ancestor of $M(v)=M(w)$, and thus, by Lemma \ref{lemma:Muv}, $M(u')=M(u)$ is an ancestor of $M(v)$. Since $B(v)=B(u)\sqcup B(w)$, we have that $M(v)$ is an ancestor of $M(u)$, and thus $M(u)=M(v)$. In conjunction with $\mathit{high}(u)=\mathit{high}(v)$, this implies that $B(v)=B(u)$, contradicting the fact that the graph is $3$-edge-connected. Thus, we have that $u'$ is not an ancestor of $v$. Since $v$ and $u'$ have $u$ as a common descendant, we infer that $u'$ is a descendant of $v$. Now, since $B(u')\subset B(u)$, we have that there exists a back-edge $(x,y)\in B(u)\setminus B(u')$. Then, $y$ is descendant of $u'$, and therefore a descendant of $v$. But this means that $(x,y)\notin B(v)$, contradicting the fact that $B(u)\subset B(v)$. We conclude that there is not $u'\in M^{-1}(M(u))$ which is a proper ancestor of $u$.
\end{proof}



\subparagraph*{Case $w=\mathit{nextM}(v)$.}

Now we will show how to find, for every vertex $v$, the unique $u$ (if it exists) which is a descendant of $v$ and such that $\{(u,p(u)),(v,p(v)),(w,p(w))\}$ is a $3$-cut, where $w=\mathit{nextM}(v)$. Obviously, the number of $3$-cuts of this type is $O(n)$. According to Lemma \ref{lemma:u_same_high_with_v}, $\mathit{high}(u)=\mathit{high}(v)$, and therefore it is sufficient to seek this $u$ in $\mathit{high}^{-1}(\mathit{high}(v))$.

\begin{proposition}
\label{prop_w=nextM(v)}
Let $h=\mathit{high}(v)$ and $w=\mathit{nextM}(v)$, and suppose that the list $\mathit{high}^{-1}(h)$ is sorted in decreasing order. Then, $u$ is a descendant of $v$ such that $\{(u,p(u)),(v,p(v)),(w,p(w))\}$ is a $3$-cut if and only if $u$ is a predecessor of $v$ in $\mathit{high}^{-1}(h)$, $\mathit{nextM}(u)=\emptyset$, $\mathit{low}(u)\geq w$, $\mathit{b\_count}(u)=\mathit{b\_count}(v)-\mathit{b\_count}(w)$, and all elements of $\mathit{high}^{-1}(h)$ between $u$ and $v$ are ancestors of $u$.
\end{proposition}
\begin{proof}
($\Rightarrow$) By proposition \ref{proposition:uvw}, we have $B(v)=B(u)\sqcup B(w)$. This shows that $\mathit{b\_count}(u)=\mathit{b\_count}(v)-\mathit{b\_count}(w)$ and $\mathit{low}(u)\geq w$ (for if we had $\mathit{low}(u)<w$, then $B(u)$ would intersect $B(w)$). Lemma \ref{lemma:u_same_high_with_v} shows that $\mathit{high}(u)=\mathit{high}(v)$ and $\mathit{nextM}(u)=\emptyset$. Since $u$ is a descendant of $v$, it is greater than $v$, and thus it is a predecessor of $v$ in $\mathit{high}^{-1}(x)$. Now suppose that there exists a $u'\in\mathit{high}^{-1}(x)$ which is lower than $u$ and greater than $v$, but it is not an ancestor of $u$. Since $u$ is a descendant of $v$, $v<u'<u$ implies that $u'$ is also a descendant of $v$. Let $(x,h)$ be a back-edge with $x$ a descendant of $u'$. Then $x$ is a also a descendant of $v$, and thus $(x,h)\in B(v)$. But since $u'$ is not a descendant of $u$, $x$ cannot be a descendant of $u$ either, and so $(x,h)\in B(v)$ and $B(v)=B(u)\sqcup B(w)$ both imply that $(x,h)\in B(w)$. However, $h=\mathit{high}(u)\geq\mathit{low}(u)\geq w$. A contradiction.\\
($\Leftarrow$) By proposition \ref{proposition:uvw}, it is sufficient to show that $B(v)=B(u)\sqcup B(w)$. Let $(x,y)\in B(u)$. Then $x$ is a descendant of $u$, and therefore a descendant of $v$. Furthermore, since $\mathit{high}(u)=\mathit{high}(v)$, we have that $y$ is a proper ancestor of $v$. This shows that $(x,y)\in B(v)$, and thus we have $B(u)\subseteq B(v)$. Now, since $M(v)=M(w)$ and $w=\mathit{nextM}(v)<v$, we have that $B(w)\subset B(v)$. Thus we have established that $B(u)\cup B(w)\subseteq B(v)$. Now observe that $B(u)\cap B(w)=\emptyset$: for if $(x,y)\in B(u)$, then $y\geq\mathit{low}(u)$, and we have assumed that $\mathit{low}(u)\geq w$; thus, $(x,y)\notin B(w)$. Now, since $\mathit{b\_count}(u)=\mathit{b\_count}(v)-\mathit{b\_count}(w)$ and $B(u)\cup B(w)\subseteq B(v)$ and $B(u)\cap B(w)=\emptyset$, we conclude that $B(v)=B(u)\sqcup B(w)$.
\end{proof}

Now let $h$ be a vertex. Based on proposition \ref{prop_w=nextM(v)}, we will show how to find, for every $v$ in the decreasingly sorted list $\mathit{high}^{-1}(h)$, the unique vertex $u\in\mathit{high}^{-1}(h)$ (if it exists) such that $\{(u,p(u)),(v,p(v)),(w,p(w))\}$ is a $3$-cut, where $w=\mathit{nextM}(v)$. To do this, we need an array $A$ of size $m$ (the number of edges of the graph), and a stack $S$. We begin by traversing the list $\mathit{high}^{-1}(h)$ from its first element, and every $u$ we meet that satisfies $\mathit{nextM}(u)=\emptyset$ and is an ancestor of its predecessor (or the first element of the list) we push it in $S$ and also store it in $A[\mathit{b\_count}(u)]$. If $u$ is not an ancestor of its predecessor, we set $A[z]=\emptyset$, for every $z\in S$, while we pop out all elements from $S$; then we push $u$ in $S$ and also store it in $A[\mathit{b\_count}(u)]$. Now, if we meet a vertex $v$ that satisfies $\mathit{nextM}(v)\neq\emptyset$ and is ancestor of its predecessor, we check whether the entry $u=A[\mathit{b\_count}(v)-\mathit{b\_count}(\mathit{nextM}(v))]$ is not $\emptyset$, and if $\mathit{low}(u)\geq\mathit{nextM}(v)$ we mark the triplet $\{(u,p(u)),(v,p(v)),(\mathit{nextM}(v),p(\mathit{nextM}(v)))\}$ (observe that $u$ satisfies all conditions of proposition \ref{prop_w=nextM(v)}). If $v$ is not an ancestor of the top element of $S$, we set $A[u]=\emptyset$, for every $u\in S$, while we pop out all elements from $S$. In any case, we keep traversing the list, following the same procedure, until we reach its end. This process is implemented in Algorithm \ref{algorithm:w=nextM(v)}.

\begin{algorithm}[!h]
\caption{\textsf{Find all $3$-cuts $\{(u,p(u)),(v,p(v)),(w,p(w))\}$, where $u$ is a descendant of $v$ and $w=\mathit{nextM}(v)$.}}
\label{algorithm:w=nextM(v)}
\LinesNumbered
\DontPrintSemicolon
initialize an array $A$ with $m$ entries (where $m$ is the number of edges of the graph)\;
initialize a stack $S$\;
sort the elements of every list $\mathit{high}^{-1}(h)$, for every vertex $h$, in decreasing order\;
\ForEach{vertex $h$}{
  $u \leftarrow $ first element of $\mathit{high}^{-1}(h)$\;
  \While{$u\neq\emptyset$}{
    $z \leftarrow $ next element of $\mathit{high}^{-1}(h)$\;
    \lIf{$z=\emptyset$}{\textbf{break}}
    \If{$z$ is not an ancestor of $u$}{
      \While{$S$ is not empty}{
        $u' \leftarrow S$.pop()\;
        $A[\mathit{b\_count}(u')] \leftarrow \emptyset$\;
      }
    }
    \If{$\mathit{nextM}(z)=\emptyset$}{
      $S$.push($z$)\;
      $A[\mathit{b\_count}(z)] \leftarrow z$\;
    }
    \ElseIf{$\mathit{nextM}(z)\neq\emptyset$}{
      $v \leftarrow z$, $w \leftarrow \mathit{nextM}(v)$\;
      \If{$A[\mathit{b\_count}(v)-\mathit{b\_count}(w)] \neq \emptyset$}{
        $u \leftarrow A[\mathit{b\_count}(v)-\mathit{b\_count}(w)]$\;
        \If{$\mathit{low}(u)\geq w$}{
          mark the triplet $\{(u,p(u)),(v,p(v)),(w,p(w))\}$\;
        }
      }
    }
    $u \leftarrow z$\;
  }
}
\end{algorithm}



\subparagraph*{Case $w\neq\mathit{nextM}(v)$.}

Now we will show how to find, for every vertex $v$, the set of all $u$ which are descendants of $v$ with the property that there exists a $w$ with $M(w)=M(v)$ and $w<\mathit{nextM}(v)$, such that $\{(u,p(u)),(v,p(v)),(w,p(w))\}$ is a $3$-cut. Let $U(v)$ denote this set. (An illustration is given in Figure \ref{figure:M(v)}.) According to Lemma \ref{lemma:u_same_high_with_v}, for every $u\in U(v)$ we have $\mathit{high}(u)=\mathit{high}(v)$, and therefore it is sufficient to seek those $u$ in $\mathit{high}^{-1}(\mathit{high}(v))$.

To do this, we use a stack $\mathit{stackU}[v]$, for every vertex $v$, in which we store vertices $u$ from $\mathit{high}^{-1}(\mathit{high}(v))$. By the time we have filled all stacks $\mathit{stackU}[v]$, the following three properties will be satisfied: $(1)$ for every vertex $v$, $U(v)\subseteq\mathit{stackU}[v]$, $(2)$ if $v\neq v'$, then $\mathit{stackU}[v]\cap\mathit{stackU}[v']=\emptyset$, and $(3)$ every $u$ in $\mathit{stackU}[v]$ is a descendant of its successors in $\mathit{stackU}[v]$. 
The contents of $\mathit{stackU}[v]$ will be all those $u$ satisfying the necessary condition described in the following lemma.

\begin{lemma}
\label{necessity_for_U(v)}
Let $h=\mathit{high}(v)$, and assume that the list $\mathit{high}^{-1}(h)$ is sorted in decreasing order. Then, $u\in U(v)$ only if $u$ is a predecessor of $v$ in $\mathit{high}^{-1}(h)$ such that $\mathit{nextM}(u)=\emptyset$, $\mathit{low}(u)<\mathit{nextM}(v)$, $\mathit{low}(u)\geq\mathit{lastM}(v)$, and all elements of $\mathit{high}^{-1}(h)$ between $u$ and $v$ are ancestors of $u$.
\end{lemma}
\begin{proof}
$u\in U(v)$ means that $u$ is a descendant of $v$ and there is an ancestor $w$ of $v$ such that $M(v)=M(w)$, $w\neq\mathit{nextM}(v)$, and $\{(u,p(u)),(v,p(v)),(w,p(w))\}$ is a $3$-cut. By proposition \ref{proposition:uvw}, we have $B(v)=B(u)\sqcup B(w)$. From this we infer that $\mathit{low}(u)\geq w$ (for otherwise, since $u$ is a descendant of $w$, we would have that $B(u)$ meets $B(w)$). This shows that $\mathit{low}(u)\geq\mathit{lastM}(v)$. Lemma \ref{lemma:u_same_high_with_v} implies that $\mathit{high}(u)=\mathit{high}(v)$ and $\mathit{nextM}(u)=\emptyset$. Furthermore, since $u$ is a descendant of $v$, it is greater than $v$, and thus it is a predecessor of $v$ in $\mathit{high}^{-1}(h)$. Now suppose, for the sake of contradiction, that $\mathit{low}(u)\geq\mathit{nextM}(v)$. Since there is a $w<\mathit{nextM}(v)$ such that $M(w)=M(v)$, there must exist a back-edge $(x,y)\in B(v)$ with $y\in T(\mathit{nextM}(v),w]$. Since $\mathit{low}(u)\geq\mathit{nextM}(v)$, it cannot be the case that $(x,y)\in B(u)$, and therefore $B(v)=B(u)\sqcup B(w)$ implies that $(x,y)\in B(w)$, which is absurd, since $y\geq w$. Thus, $\mathit{low}(u)<\mathit{nextM}(v)$. 
Finally, suppose, for the sake of contradiction, that there exists a $u'\in\mathit{high}^{-1}(h)$ which is lower than $u$ and greater than $v$, but it is not an ancestor of $u$. Since $u$ is a descendant of $v$, $v<u'<u$ implies that $u'$ is also a descendant of $v$. Let $(x,h)$ be a back-edge with $x$ a descendant of $u'$. Then $x$ is a also a descendant of $v$, and thus $(x,h)\in B(v)$. But since $u'$ and $u$ are not related as ancestor or descendant, $x$ cannot be a descendant of $u$. Thus, $(x,h)\notin B(u)$. Since $(x,h)\in B(v)$ and $B(v)=B(u)\sqcup B(w)$, this implies that $(x,h)\in B(w)$. However, $h=\mathit{high}(u)\geq\mathit{low}(u)\geq w$. A contradiction.
\end{proof}

\begin{figure}
\begin{center}
\centerline{\includegraphics[trim={0cm 18cm 0cm 0cm}, scale=1,clip, width=\textwidth]{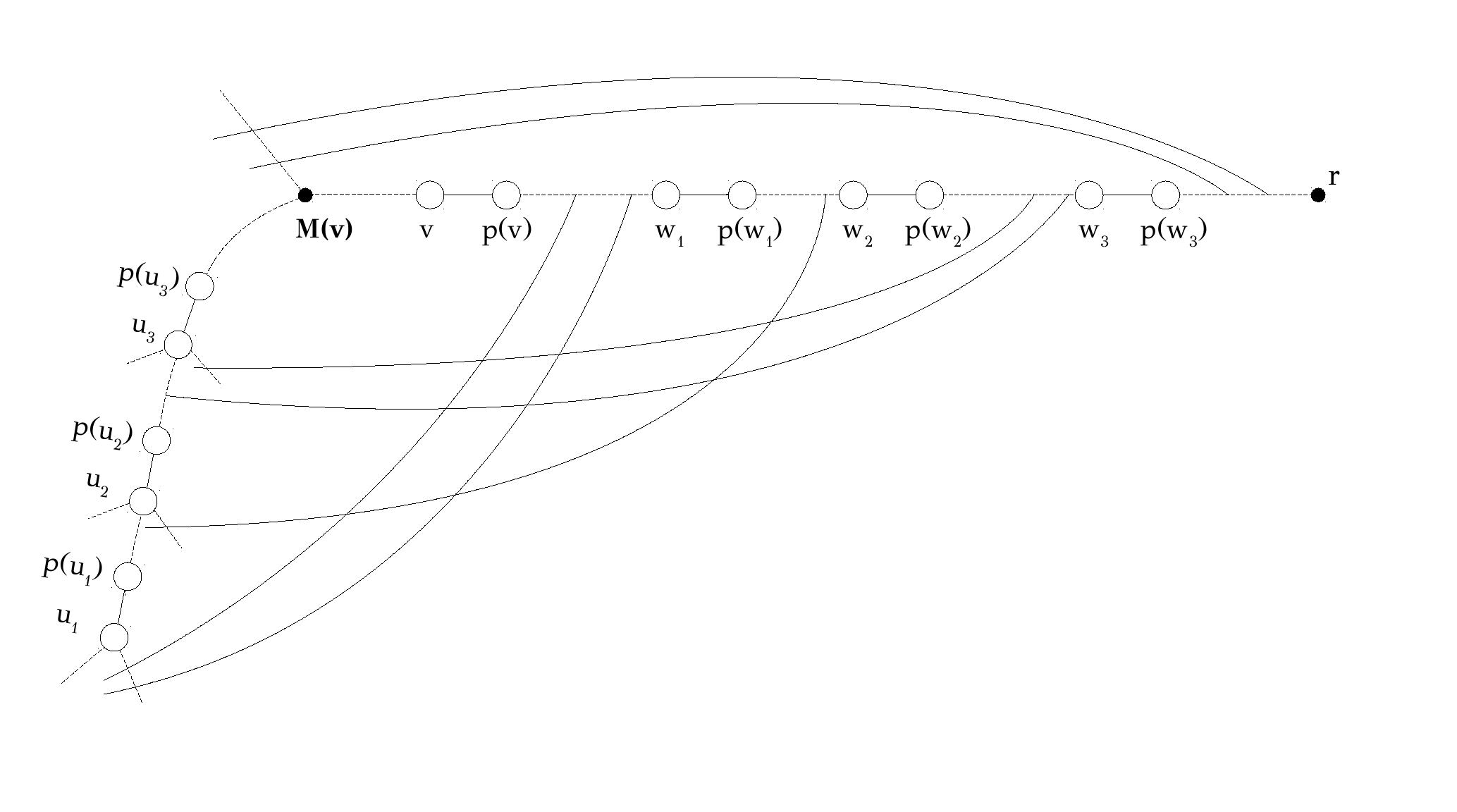}}
\caption{In this example we have $M(v)=M(w_1)=M(w_2)=M(w_3)$, $U(v)=\{u_1,u_2,u_3\}$, and the triplets $\{(u_i,p(u_i)),(v,p(v)),(w_i,p(w_i))\}$, for $i\in\{1,2,3\}$, are $3$-cuts. Observe that all $\{u_1,u_2,u_3\}$ are related as ancestor and descendant. This property is proved in Lemma \ref{necessity_for_U(v)}. Furthermore, all $u\in U(v)$ have $\mathit{high}(u)=\mathit{high}(v)$.}
\label{figure:M(v)}
\end{center}
\end{figure}

Thus, $\mathit{stackU}[v]$ contains all $u$ that are predecessors of $v$ in $\mathit{high}^{-1}(\mathit{high}(v))$ and satisfy $\mathit{nextM}(u)=\emptyset$, $\mathit{low}(u)<\mathit{nextM}(v)$, $\mathit{low}(u)\geq\mathit{lastM}(v)$ and all elements of $\mathit{high}^{-1}(\mathit{high}(v))$ between $u$ and $v$ are ancestors of $u$. By Lemma \ref{necessity_for_U(v)}, property $(1)$ of the stacks $\mathit{stackU}[v]$ is satisfied. The following lemma shows that property $(2)$ is also satisfied.

\begin{lemma}
\label{lemma:property2}
Let $v,v'$ be two vertices such that $v'$ is a proper ancestor of $v$ with $\mathit{high}(v')=\mathit{high}(v)$, and let $u\in\mathit{stackU}[v]$. Then $u\notin\mathit{stackU}[v']$.
\end{lemma}
\begin{proof}
First observe that the stacks $\mathit{stackU}[v]$ and $\mathit{stackU}[v']$ are non-empty only if $\mathit{nextM}(v)\neq\emptyset$ and $\mathit{nextM}(v')\neq\emptyset$.
Now, since $\mathit{high}(v')=\mathit{high}(v)$, by Lemma \ref{lemma:next_and_last}, we have that $\mathit{nextM}(v)<\mathit{lastM}(v')$. Since $u\in\mathit{stackU}[v]$, it has $\mathit{low}(u)<\mathit{nextM}(v)$. But then $\mathit{low}(u)<\mathit{lastM}(v')$, and so $u\notin\mathit{stackU}[v']$.
\end{proof}

This implies that the total number of elements in all stacks $\mathit{stackU}[v]$ (by the time we have filled them) is $O(n)$. Now let $h$ be a vertex, and let us show how to fill the stacks $\mathit{stackU}[v]$, for all $v$ in the decreasingly sorted list $\mathit{high}^{-1}(h)$. To do this, we will need a stack $S$. We begin traversing the list $\mathit{high}^{-1}(h)$ from its first element, and when we process a vertex $u$ such that $\mathit{nextM}(u)=\emptyset$ we push it in $S$ if it is an ancestor of its predecessor (or the first elements of the list). Otherwise, we drop all elements from $S$, push $u$ in $S$, and keep traversing the list. When we meet a vertex $v$ that satisfies $\mathit{nextM}(v)\neq\emptyset$ and is also an ancestor of its predecessor, we check whether the top element $u$ of $S$ satisfies $\mathit{low}(u)<\mathit{lastM}(v)$, in which case we start popping elements out of $S$, until the top element $u$ of $S$ (if $S$ is not left empty) satisfies $\mathit{low}(u)\geq\mathit{lastM}(v)$. Then, as long as the top element $u$ of $S$ satisfies $\mathit{low}(u)<\mathit{nextM}(v)$, we repeatedly pop out the top element from $S$ and push it in $\mathit{stackU}[v]$. If $v$ is not an ancestor of its predecessor, we drop all elements from $S$. In any case, we keep traversing the list, following the same procedure, until we reach its end. This process is implemented in Algorithm \ref{algorithm:fill_stacks}. Property $(3)$ of the stacks $\mathit{stackU}$ is satisfied due to the way we fill them with this algorithm. To prove the correctness of Algorithm \ref{algorithm:fill_stacks} - i.e., that by the time we reach the end of $\mathit{high}^{-1}(h)$, every stack $\mathit{stackU}[v]$, for every $v\in\mathit{high}^{-1}(h)$, contains all elements $u$ satisfying the necessary condition in Lemma \ref{necessity_for_U(v)} -, we need the following two lemmata. 

\begin{lemma}
\label{lemma:high_low}
If $u'$ is an ancestor of $u$ with $\mathit{high}(u)=\mathit{high}(u')$, then $\mathit{low}(u')\leq\mathit{low}(u)$.
\end{lemma}
\begin{proof}
Let $(x,y)\in B(u)$. Then $x$ is a descendant of $u$, and therefore a descendant of $u'$. Furthermore, $y\leq\mathit{high}(u)=\mathit{high}(u')$, and therefore $y$ is a proper ancestor of $u'$. This shows that $(x,y)\in B(u')$, and thus we have $B(u)\subseteq B(u')$. $\mathit{low}(u')\leq\mathit{low}(u)$ is an immediate consequence of this fact.
\end{proof}

\begin{lemma}
\label{lemma:next_and_last}
Let $v,v'$ be two vertices such that $v'$ is a proper ancestor of $v$, $\mathit{nextM}(v)\neq\emptyset$, $\mathit{nextM}(v')\neq\emptyset$, and $\mathit{high}(v')=\mathit{high}(v)$. Then, $\mathit{nextM}(v)<\mathit{lastM}(v')$.
\end{lemma}
\begin{proof}
Let $(x,y)\in B(v)$. Then $x$ is a descendant of $v$, and therefore a descendant of $v'$. Furthermore, since $y\leq\mathit{high}(v)$ and $\mathit{high}(v)=\mathit{high}(v')$ and $\mathit{high}(v')<v'$, we have that $y$ is a proper ancestor of $v'$. This shows that $(x,y)\in B(v')$, and thus $B(v)\subseteq B(v')$. From this we infer that $M(v)$ is a descendant of $M(v')$. Now, since $M(\mathit{nextM}(v))=M(v)$ and $\mathit{nextM}(v)<v$, we have that $B(\mathit{nextM}(v))\subset B(v)$. This means that there exists a back-edge $(x,y)$ such that $x$ is a descendant of $M(v)$ and $y$ is a proper ancestor of $v$ but not a proper ancestor of $\mathit{nextM}(v)$. Then, since $(x,y)\in B(v)$, we have $y\leq\mathit{high}(v)$, and so $\mathit{high}(v)$ is not a proper ancestor of $\mathit{nextM}(v)$, and thus $\mathit{nextM}(v)$ is an ancestor of $\mathit{high}(v)$. Since $\mathit{high}(v)=\mathit{high}(v')$ and $\mathit{high}(v')$ is a proper ancestor of $v'$, we infer that $\mathit{nextM}(v)$ is a proper ancestor of $v'$. Now suppose, for the sake of contradiction, that $\mathit{lastM}(v')$ is an ancestor of $\mathit{nextM}(v)$. Let $(x,y)\in B(\mathit{lastM}(v'))$. Then, $x$ is a descendant of $M(\mathit{lastM}(v'))$, and thus a descendant of $M(v')$, and thus a descendant of $v'$, and thus a descendant of $\mathit{nextM}(v)$. Furthermore, $y$ is a proper ancestor of $\mathit{lastM}(v')$, and therefore a proper ancestor of $\mathit{nextM}(v)$. This shows that $(x,y)\in B(\mathit{nextM}(v))$, and thus we have $B(\mathit{lastM}(v'))\subseteq B(\mathit{nextM}(v))$. From this we infer that $M(\mathit{lastM}(v'))$ is a descendant of $M(\mathit{nextM}(v))$. But $M(\mathit{lastM}(v'))=M(v')$ and $M(\mathit{nextM}(v))=M(v)$. Thus, $M(v')$ is a descendant of $M(v)$. Since $M(v)$ is a descendant of $M(v')$, we conclude that $M(v')=M(v)$. But this implies, in conjunction with $\mathit{high}(v')=\mathit{high}(v)$, that $B(v)=B(v')$, contradicting the fact that the graph is $3$-edge-connected. This shows that $\mathit{nextM}(v)$ is a proper ancestor of $\mathit{lastM}(v')$.
\end{proof}

Now, to prove the correctness of Algorithm \ref{algorithm:fill_stacks}, we have to show that the elements we push into $\mathit{stackU}[v]$ satisfy the necessary condition in Lemma \ref{necessity_for_U(v)}, and the elements we pop out from $S$ do not satisfy this condition either for $v$ or for any successor of $v$ in the list $\mathit{high}^{-1}(h)$. So, let $v$ be a vertex in $\mathit{high}^{-1}(h)$ such that $\mathit{nextM}(v)\neq\emptyset$, and let $v'$ be a successor of $v$ in $\mathit{high}^{-1}(h)$ such that $\mathit{nextM}(v')\neq\emptyset$. Now, when we meet $v$ as we traverse $\mathit{high}^{-1}(x)$, we pop out the top elements $u$ from $S$ that have $\mathit{low}(u)<\mathit{lastM}(v)$. By the definition of $\mathit{stackU}[v]$, these are not included in $\mathit{stackU}[v]$. Now, by Lemma \ref{lemma:next_and_last}, we have $\mathit{nextM}(v)<\mathit{lastM}(v')$. Since $\mathit{low}(u)<\mathit{lastM}(v)\leq\mathit{nextM}(v)$, we have $\mathit{low}(u)<\mathit{lastM}(v')$, and thus $u$ is not in $\mathit{stackU}[v']$ either, so it does not matter that we pop those $u$ out of $S$. Then, once we reach a $\tilde{u}$ in $S$ that satisfies $\mathit{low}(\tilde{u})\geq\mathit{lastM}(v)$, we pop out the top elements $u$ of $S$ that have $\mathit{low}(u)<\mathit{nextM}(v)$, and push them into $\mathit{stackU}[v]$. This is according to the definition of $\mathit{stackU}[v]$. Since $\mathit{nextM}(v)<\mathit{lastM}(v')$ and $\mathit{low}(u)<\mathit{nextM}(v)$, we have $\mathit{low}(u)<\mathit{lastM}(v')$, and so, again, these $u$ are not included in $\mathit{stackU}[v']$, and thus it does not matter that we pop them out of $S$. Now, when we reach a $u$ in $S$ that has $\mathit{low}(u)\geq\mathit{nextM}(v)$, we can be certain, by Lemma \ref{lemma:high_low}, that no $u'$ in $S$ has $\mathit{low}(u')<\mathit{nextM}(v)$, since all elements of $S$ are descendants of $u$ (by the way we fill the stack $S$), and thus they have $\mathit{low}(u')\geq\mathit{low}(u)\geq\mathit{nextM}(v)$. Then it is proper to move on to the next element of $\mathit{high}^{-1}(h)$.

\begin{algorithm}[!h]
\caption{\textsf{Fill all stacks $\mathit{stackU}[v]$, for all vertices $v$}}
\label{algorithm:fill_stacks}
\LinesNumbered
\DontPrintSemicolon
initialize a stack $S$\;
\lForEach{vertex $v$}{initialize a stack $\mathit{stackU}[v]$}
\ForEach{vertex $h$}{
  $u \leftarrow $ first element of $\mathit{high}^{-1}(h)$\;
  \While{$u\neq\emptyset$}{
    $z \leftarrow $ next element of $\mathit{high}^{-1}(h)$\;
    \lIf{$z=\emptyset$}{\textbf{break}}
    \If{$z$ is not an ancestor of $u$}{
      pop out all elements from $S$\;
    }
    \If{$\mathit{nextM}(z)=\emptyset$}{
      $S$.push($z$)\;
    }
    \ElseIf{$\mathit{nextM}(z)\neq\emptyset$}{
      \lWhile{$\mathit{low}(S\mathit{.top()})<\mathit{lastM}(v)$}{$S$.pop()}
      \While{$\mathit{low}(S\mathit{.top()})<\mathit{nextM}(v)$}{
        $u \leftarrow S$.pop()\;
        $\mathit{stackU}[v]$.push($u$)\;
      }
    }
    $u \leftarrow z$\;
  }
}
\end{algorithm}

\begin{lemma}
\label{lemma:lowUinStack}
Let $v$ be a vertex and $u,u'$ two elements in $\mathit{stackU}[v]$, where $u$ is a predecessor of $u'$ in $\mathit{stackU}[v]$. Then, $\mathit{low}(u')\leq\mathit{low}(u)$.
\end{lemma}
\begin{proof}
Since $u,u'\in\mathit{stackU}[v]$, we have $\mathit{high}(u)=\mathit{high}(v)=\mathit{high}(u')$. Since $u$ is a predecessor of $u'$ in $\mathit{stackU}[v]$, by property $(3)$ of $\mathit{stackU}[v]$ we have that $u$ is a descendant of $u'$. Thus, by Lemma \ref{lemma:high_low}, we get $\mathit{low}(u')\leq\mathit{low}(u)$.
\end{proof}

The next lemma is the basis to find all $3$-cuts of the form $\{(u,p(u)),(v,p(v)),(w,p(w))\}$, where $u$ is a descendant of $v$, $M(v)=M(w)$, and $w\neq\mathit{nextM}(v)$.

\begin{lemma}
\label{lemma:greatestW}
Let $u$ be a vertex in $\mathit{stackU}[v]$ and $w$ a proper ancestor of $v$ such that $M(w)=M(v)$. Then, if $\{(u,p(u)),(v,p(v)),(w,p(w))\}$ is a $3$-cut, we have that $\mathit{b\_count}(v)=\mathit{b\_count}(u)+\mathit{b\_count}(w)$ and $w$ is the greatest element of $M^{-1}(M(v))$ such that $w\leq\mathit{low}(u)$. Conversely, if $\mathit{b\_count}(v)=\mathit{b\_count}(u)+\mathit{b\_count}(w)$ and $w\leq\mathit{low}(u)$, then $\{(u,p(u)),(v,p(v)),(w,p(w))\}$ is a $3$-cut.
\end{lemma}
\begin{proof}
($\Rightarrow$) By proposition \ref{proposition:uvw}, we have $B(v)=B(u)\sqcup B(w)$. This explains both $\mathit{b\_count}(v)=\mathit{b\_count}(u)+\mathit{b\_count}(w)$ and $w\leq\mathit{low}(u)$. (For if we had $\mathit{low}(u)<w$, then, since $u$ is a descendant of $w$, $B(u)$ would meet $B(w)$.) Now suppose, for the sake of contradiction, that there is a vertex $w'$ such that $M(w')=M(v)$ and $w<w'\leq\mathit{low}(u)$. Since $B(v)=B(u)\sqcup B(w)$, we have that $\mathit{low}(u)<v$, and therefore $w'<v$. Since $M(w')=M(v)$, this means that $B(w')\subset B(v)$. Furthermore, since $M(w)=M(w')$ and $w<w'$, we infer that $B(w)\subset B(w')$, and therefore there exists a back-edge $(x,y)\in B(w')\setminus B(w)$. Then, by $B(w')\subset B(v)$, we have that $(x,y)\in B(v)$, and $B(v)=B(u)\sqcup B(w)$ implies that $(x,y)\in B(u)$ or $(x,y)\in B(w)$. Since $(x,y)\notin B(w)$, $(x,y)\in B(u)$ is the only option left. But $y$ is a proper ancestor of $w'$, and therefore a proper ancestor of $\mathit{low}(u)$ (since $w'\leq\mathit{low}(u)$). This implies that $(x,y)\notin B(u)$, which is absurd. We conclude that $w$ is the greatest element of $M^{-1}(M(v))$ such that $w\leq\mathit{low}(u)$.\\
($\Leftarrow$) By proposition \ref{proposition:uvw}, is is sufficient to show that $B(v)=B(u)\sqcup B(w)$. $u\in\mathit{stackU}[v]$ implies that $u$ is a descendant of $v$ such that $\mathit{high}(u)=\mathit{high}(v)$. Now let $(x,y)\in B(u)$. Then $x$ is a descendant of $u$, and therefore a descendant of $v$. Furthermore, $y\leq\mathit{high}(u)=\mathit{high}(v)$, and therefore $y$ is a proper ancestor of $v$. This shows that $(x,y)\in B(v)$, and thus we have $B(u)\subseteq B(v)$. Since $M(w)=M(v)$ and $w<v$, we have $B(w)\subset B(v)$. Thus we have established that $B(u)\cup B(w)\subseteq B(v)$. Notice that no $(x,y)\in B(u)$ is contained in $B(w)$, since $y\geq\mathit{low}(u)\geq w$, and thus $y$ is not a proper ancestor of $w$. Thus we have $B(u)\cap B(w)=\emptyset$. Now $B(v)=B(u)\sqcup B(w)$ follows from $B(u)\cup B(w)\subseteq B(v)$, $B(u)\cap B(w)=\emptyset$ and $\mathit{b\_count}(v)=\mathit{b\_count}(u)+\mathit{b\_count}(w)$.
\end{proof}

Now our goal is to find, for every $u\in\mathit{stackU}[v]$, for every vertex $v$, the vertex $w$ (if it exists) which has $M(w)=M(v)$ and $w<\mathit{nextM}(v)$, and is such that $\{(u,p(u)),(v,p(v)),(w,p(w))\}$ is a $3$-cut. By Lemma \ref{lemma:greatestW}, $w$ has the property that it is the greatest vertex in $M^{-1}(M(v))$ which has $w\leq\mathit{low}(u)$. Let us describe a simple method to find the $w$ with this property, which will give us the intuition to provide a linear-time algorithm for our problem. So let $v$ be a vertex, $m=M(v)$, and $u$ be a vertex in $\mathit{stackU}[v]$. A simple idea is to start from $v$ and keep traversing the list $M^{-1}(m)$, through the pointers $\mathit{nextM}$, until we reach a $w\in M^{-1}(m)$ such that $w\leq\mathit{low}(u)$. The problem here is that we may have to pass from the same elements of $M^{-1}(m)$ an excessive amount of times (depending on the number of elements in $\mathit{stackU}[v]$). We can remedy this by keeping in a variable $\mathit{lowestW}$ the $w$ that we reached the last time we processed a $u\in\mathit{stackU}[v]$. Then, when we process the successor of $u$ in $\mathit{stackU}[v]$, we begin the search in $M^{-1}(m)$ from $\mathit{lowestW}$. This will work, since the every $u\in\mathit{stackU}[v]$ is a descendant of its successor $u'$ in $\mathit{stackU}[v]$ (due to the way we have filled the stacks $\mathit{stackU}$ with Algorithm \ref{algorithm:fill_stacks}), and we have $\mathit{high}(u)=\mathit{high}(u')$, and therefore, by Lemma \ref{lemma:high_low}, $\mathit{low}(u')\leq\mathit{low}(u)$. However, this is, again, not a linear-time procedure, since, for every vertex $v$, when we start processing the first vertex in $\mathit{stackU}[v]$, we begin traversing the list $M^{-1}(M(v))$ from $v$, and therefore, every time we process a vertex $v'$ with $M(v')=M(v)$, we may have to pass again from the same vertices that we passed from during the processing of $v$, exceeding the time bound in total. Now, to achieve linear time, we process the vertices from the lowest to the highest, and, for every $v$ that we process, we keep in a variable $\mathit{lowestW}[v]$ the $w$ that we reached the last time we processed a $u\in\mathit{stackU}[v]$. Then, when we have to process a $u\in\mathit{stackU}[v]$, we traverse the list $M^{-1}(M(v))$ through the pointers $\mathit{lowestW}$, starting from $\mathit{lowestW}[v]$. (Initially, we set every $\mathit{lowestW}[v]$ to $\mathit{nextM}(v)$.) Thus we perform a kind of path-compression method, which is shown Algorithm \ref{algorithm:w_neq_nextM(v)}. The next three lemmata will be used in proving the correctness and linear complexity of Algorithm \ref{algorithm:w_neq_nextM(v)}.

\begin{lemma}
\label{lemma:lowestW}
Let $v$ be a vertex and $u\in\mathit{stackU}[v]$. When we reach line \ref{alg_line_assign} during the processing of $u$, we have that $w$ is a vertex in $M^{-1}(M(v))$ such that $w\leq\mathit{low}(u)$ and $w\leq\mathit{min}\{\mathit{low}(u')\mid\exists v' \mbox{ with } M(v')=M(v) \mbox{, } w<v'<v \mbox{ and } u'\in\mathit{stackU}[v']\}$.
\end{lemma}
\begin{proof}
First observe that, during the processing of a vertex $v$, the variables $w$ and $\mathit{lowestW}[v]$ are members of $M^{-1}(M(v))$, and $w$ is an ancestor of $v$ while $\mathit{lowestW}[v]$ is a proper ancestor of $v$. (It is easy to see this inductively. For if this holds for all vertices $v'<v$, then it is also true for $v$, since the \textbf{while} loop in line \ref{alg_line_while} assigns $w$ to $\mathit{lowestW}[w]$, and $w$ is assumed to be an ancestor of $v$ with $M(w)=M(v)$, and thus $\mathit{lowestW}[w]$ is also an ancestor of $v$ with $M(\mathit{lowestW}[w])=M(v)$, due to the inductive hypothesis.) Then it is obvious that, when we reach line \ref{alg_line_assign} during the processing of $u\in\mathit{stackU}[v]$, we have that $M(w)=M(v)$ and $w\leq\mathit{low}(u)$, since the \textbf{while} loop in line \ref{alg_line_while} terminates precisely when such a $w$ is found. Now we will show that, when we process a $u\in\mathit{stackU}[v]$, every time $w$ is assigned $\mathit{lowestW}[w]$ during the execution of the \textbf{while} loop in line \ref{alg_line_while}, we have $w\leq\mathit{low}(u')$, for every $u'\in\mathit{stackU}[v']$, for every $v'$ with $M(v')=M(v)$ and $w<v'<v$. It is easy to see this inductively. Suppose, then, that this was the case for every vertex that we processed before $v$, for every predecessor of $u$ in $\mathit{stackU}[v]$ that we already processed, and for every step of the \textbf{while} loop in line \ref{alg_line_while} in the processing of $u$ so far. Thus, now $w$ has the property that $w\leq\mathit{low}(u')$, for every $u'\in\mathit{stackU}[v']$, for every $v'$ with $M(v')=M(v)$ and $w<v'<v$. So let us perform $w\leftarrow \mathit{lowestW}[w]$ once more (which means that we still have $w>\mathit{low}(u)$), and let $\tilde{w}$ be the current value of $w$, to distinguish it from the previous one which we will denote simply as $w$. Now, due to the inductive hypothesis, we have that $\tilde{w}\leq\mathit{low}(u')$ for every $u'\in\mathit{stackU}[v']$, for every $v'$ with $M(v')=M(v)$ and $\tilde{w}<v'<w$. We also have (again, due to the inductive hypothesis) that $w\leq\mathit{low}(u')$ for every $u'\in\mathit{stackU}[v']$, for every $v'$ with $M(v')=M(v)$ and $w<v'<v$. Since $\tilde{w}<w$, we thus have $\tilde{w}\leq\mathit{low}(u')$, for every $u'\in\mathit{stackU}[v']$, for every $v'$ with $M(v')=M(v)$ and $\tilde{w}<v'<w$ or $w<v'<v$. Thus we only have to consider the case $v'=w$, and prove that every $u'\in\mathit{stackU}[w]$ satisfies $\tilde{w}\leq\mathit{low}(u')$. Observe that $\mathit{lowestW}[w]$ was updated for the last time in line \ref{alg_line_assign} when we were processing the last element $\tilde{u}$ of $\mathit{stackU}[w]$. Then, since $\tilde{w}=\mathit{lowestW}[w]$, due to the inductive hypothesis we have that $\tilde{w}\leq\mathit{low}(\tilde{u})$. Since every $u'\in\mathit{stackU}[w]$ has $\mathit{high}(u')=\mathit{high}(\tilde{u})$ and $\tilde{u}$ is an ancestor of its predecessors in $\mathit{stackU}[w]$ (due to the way we have filled the stacks $\mathit{stackU}$ with Algorithm \ref{algorithm:fill_stacks}), by Lemma \ref{lemma:high_low} we have that $\mathit{low}(\tilde{u})\leq\mathit{low}(u')$, and therefore $\tilde{w}\leq\mathit{low}(u')$. Thus we have shown that $\tilde{w}\leq\mathit{low}(u')$, for every $u'\in\mathit{stackU}[v']$, for every $v'$ with $M(v')=M(v)$ and $\tilde{w}<v'<v$.
\end{proof}

\begin{lemma}
\label{lemma:lowestW2}
Let $v$ be a vertex and $u\in\mathit{stackU}[v]$. When we reach line \ref{alg_line_assign} during the processing of $u$, we have that $w$ is the greatest vertex in $M^{-1}(M(v))$ such that $w\leq\mathit{low}(u)$ and $w\leq\mathit{min}\{\mathit{low}(u')\mid\exists v' \mbox{ with } M(v')=M(v) \mbox{, } w<v'<v \mbox{ and } u'\in\mathit{stackU}[v']\}$.
\end{lemma}
\begin{proof}
We will prove this lemma by induction. Let's assume, then, that, for every vertex $v'\leq v$, and every $u'\in\mathit{stackU}[v']$ that we processed so far, whenever we reached line \ref{alg_line_assign} $w$ was the greatest vertex with $M(w)=M(v')$ such that $w\leq\mathit{low}(u)$ and $w\leq\mathit{min}\{\mathit{low}(u'')\mid\exists v'' \mbox{ with } M(v'')=M(v') \mbox{, } w<v''<v' \mbox{ and } u''\in\mathit{stackU}[v'']\}$. Now let $u$ be the next element of $\mathit{stackU}[v]$ that we process. 
Let $\tilde{w}$ be the greatest vertex with $M(\tilde{w})=M(v)$ such that $\tilde{w}\leq\mathit{low}(u)$ and $\tilde{w}\leq\mathit{min}\{\mathit{low}(u')\mid\exists v' \mbox{ with } M(v')=M(v) \mbox{, } \tilde{w}<v'<v \mbox{ and } u'\in\mathit{stackU}[v']\}$. (The existence of such a $\tilde{w}$ is guaranteed by Lemma \ref{lemma:lowestW}.) 
Let $w$ be the last vertex during the execution of the \textbf{while} loop in line \ref{alg_line_while} that had $w>\mathit{low}(u)$, and let $w'=\mathit{lowestW}[w]$. Then we have that $w'=\mathit{lowestW}[w]\leq\mathit{low}(u)$, and the \textbf{while} loop terminates here. We will show that $w'=\tilde{w}$.
We distinguish two cases, depending on whether $w'=\mathit{nextM}(w)$ or $w'\neq\mathit{nextM}(w)$. In the first case, we have that $w>\mathit{low}(u)$, but $\mathit{nextM}(w)\leq\mathit{low}(u)$. Thus, $w'=\mathit{nextM}(w)$ is the greatest vertex with $M(w')=M(v)$ such that $w'\leq\mathit{low}(u)$, and so we have $w'=\tilde{w}$ (since $w'$ satisfies also $w'\leq\mathit{min}\{\mathit{low}(u')\mid\exists v' \mbox{ with } M(v')=M(v) \mbox{, } w<v'<v \mbox{ and } u'\in\mathit{stackU}[v']\}$, by Lemma \ref{lemma:lowestW}). Now, if $w'\neq\mathit{nextM}(w)$, this means, due to the inductive hypothesis (and since $w'=\mathit{lowestW}[w]$), that $w'$ is the greatest vertex with $M(w')=M(w)$ such that $w'\leq\mathit{low}(\tilde{u})$ and $w'\leq\mathit{min}\{\mathit{low}(u')\mid\exists v' \mbox{ with } M(v')=M(w) \mbox{, } w'<v'<w \mbox{ and } u'\in\mathit{stackU}[v']\}$, where $\tilde{u}$ is the last element in $\mathit{stackU}[w]$. Now, since $\tilde{w}$ satisfies $\tilde{w}\leq\mathit{min}\{\mathit{low}(u')\mid\exists v' \mbox{ with } M(v')=M(v) \mbox{, } \tilde{w}<v'<v \mbox{ and } u'\in\mathit{stackU}[v']\}$ and $\tilde{w}<w<v$, we have $\tilde{w}\leq\mathit{low}(\tilde{u})$ and $\tilde{w}\leq\mathit{min}\{\mathit{low}(u')\mid\exists v' \mbox{ with } M(v')=M(w) \mbox{, } \tilde{w}<v'<w \mbox{ and } u'\in\mathit{stackU}[v']\}$. Thus, $\tilde{w}$ cannot be greater than $w'$, and so we have $w'\geq\tilde{w}$. Since $w'\leq\mathit{low}(u)$, and, as a consequence of Lemma \ref{lemma:lowestW}, $w'\leq\mathit{min}\{\mathit{low}(u')\mid\exists v' \mbox{ with } M(v')=M(v) \mbox{, } w'<v'<v \mbox{ and } u'\in\mathit{stackU}[v']\}$, it must be the case that $w'=\tilde{w}$.
\end{proof}

\begin{lemma}
\label{lemma:lowestW3}
Let $\{(u,p(u)),(v,p(v)),(w,p(w))\}$ be a $3$-cut where $u$ is a descendant of $v$, $v$ is a descendant of $w$ with $M(v)=M(w)$, and $w\neq\mathit{nextM}(v)$. Then, $w$ is the greatest vertex in $M^{-1}(M(v))$ such that $w\leq\mathit{low}(u)$ and $w\leq\mathit{min}\{\mathit{low}(u')\mid\exists v' \mbox{ with } M(v')=M(v) \mbox{, } w<v'<v \mbox{ and } u'\in\mathit{stackU}[v']\}$.
\end{lemma}
\begin{proof}
Suppose, for the sake of contradiction, that there exists a vertex $v'$ with $M(v')=M(v)$ and $w<v'<v$, such that there exists a $u'\in\mathit{stackU}[v']$ with $\mathit{low}(u')<w$. Since $u'\in\mathit{stackU}[v']$, we have that $u'$ is a proper descendant of $v'$ with $\mathit{high}(u')=\mathit{high}(v')$.
Let $(x,y)\in B(u')$ (of course, $B(u')$ is not empty, since the graph is $3$-edge-connected). Then $x$ is a descendant of $u'$, and therefore a descendant of $v'$. Furthermore, $y\leq\mathit{high}(u')=\mathit{high}(v')$, and therefore $y$ is a proper ancestor of $v'$. This shows that $(x,y)\in B(v')$. Thus we have $B(u')\subset B(v')$. Since $M(v')=M(v)$ and $v'<v$, we have $B(v')\subset B(v)$. Thus, $B(u')\subset B(v)$.
Now we will prove that $u'$ is not related as ancestor or descendant with $u$. First, since $\mathit{low}(u')<w\leq\mathit{low}(u)$, it cannot be the case that $u'$ is a descendant of $u$ (for a back-edge $(x,\mathit{low}(u'))\in B(u')$ would also be a back-edge in $B(u)$, and thus we would have $\mathit{low}(u)\leq\mathit{low}(u')$, which is a absurd). Suppose, then, that $u'$ is an ancestor of $u$. Since $v'$ is a proper ancestor of $v$ with $M(v')=M(v)$, we must have $\mathit{high}(v')<\mathit{high}(v)$; and since $\mathit{high}(u')=\mathit{high}(v')$, we therefore have $\mathit{high}(u')<\mathit{high}(v)$. This means that $u'$ (which is related as ancestor or descendant with $v$, since we supposed it is an ancestor of $u$) is a proper ancestor of $v$, and therefore a proper ancestor of $M(v)$. Since, then, $u'$ is a descendant $v'$ and $M(v')=M(v)$, by Lemma \ref{lemma:Muv} we have that $M(u')$ is an ancestor of $M(v)$. But $B(u')\subset B(v)$ implies that $M(u')$ is a descendant of $M(v)$, and therefore $M(u')=M(v)$. Since $M(v)=M(v')$ and $\mathit{high}(v')=\mathit{high}(u')$, we get that $B(u')=B(v')$, which implies that $v'=u'$ - a contradiction. Thus we have shown that $u'$ is not related as ancestor or descendant with $u$.

Now let $(x,y)$, with $y=\mathit{high}(u')$, be a back-edge in $B(u')$. Then we have $(x,y)\in B(v)$. By proposition \ref{proposition:uvw}, we have $B(v)=B(u)\sqcup B(w)$, and therefore $(x,y)\in B(u)$ or $(x,y)\in B(w)$. Since $u'$ is not related as ancestor of descendant with $u$, it cannot be the case that $x$ (which is a descendant of $u'$) is a descendant of $u$, and therefore $(x,y)\in B(u)$ is rejected.
Now, since $B(u')\subset B(v')$, we have $(x,y)\in B(v')$. Since $M(v')=M(w)$ and $w<v'$, we have that $B(w)\subset B(v')$, and thus there exists a back-edge $(x',y')\in B(v')$ such that $y'\in T(v',w]$. But since $y=\mathit{high}(u')=\mathit{high}(v')$, we must have $y'\leq y$. Thus, $y$ is not a proper ancestor of $w$, and so $(x,y)\notin B(w)$, either. We have arrived at a contradiction, as a consequence of our initial supposition. This shows that there is no vertex $v'$ with $M(v')=M(v)$ and $w<v'<v$, such that there exists a $u'\in\mathit{stackU}[v']$ with $\mathit{low}(u')<w$. Thus, $w\leq\mathit{min}\{\mathit{low}(u')\mid\exists v' \mbox{ with } M(v')=M(v) \mbox{, } w<v'<v \mbox{ and } u'\in\mathit{stackU}[v']\}$.
Now, by Lemma \ref{lemma:greatestW}, $w$ is the greatest vertex in $M^{-1}(M(v))$ with $w\leq\mathit{low}(u)$. Thus, $w$ must be the greatest vertex in $M^{-1}(M(v))$ that satisfies both $w\leq\mathit{low}(u)$ and $w\leq\mathit{min}\{\mathit{low}(u')\mid\exists v' \mbox{ with } M(v')=M(v) \mbox{, } w<v'<v \mbox{ and } u'\in\mathit{stackU}[v']\}$.
\end{proof}

\begin{algorithm}[!h]
\caption{\textsf{Find all $3$-cuts $\{(u,p(u)),(v,p(v)),(w,p(w))\}$, where $u$ is a descendant of $v$, $v$ is a descendant of $w$ with $M(v)=M(w)$, and $w\neq\mathit{nextM}(v)$.}}
\label{algorithm:w_neq_nextM(v)}
\LinesNumbered
\DontPrintSemicolon
initialize an array $\mathit{lowestW}$ with $n$ entries\;
\lForEach{vertex $v$}{$\mathit{lowestW}[v] \leftarrow \mathit{nextM}(v)$}
\For{$v\leftarrow 1$ to $v\leftarrow n$}{
\label{alg_line_for}
  \While{$\mathit{stackU}[v]\mathit{.top()}\neq\emptyset$}{
    \label{alg_line_whileU}
    $u \leftarrow \mathit{stackU}[v]$.pop()\;
    $w \leftarrow \mathit{lowestW}[v]$\;
    \label{alg_initialize}
    \lWhile{$w>\mathit{low}(u)$}{$w \leftarrow \mathit{lowestW}[w]$}
    \label{alg_line_while}
    $\mathit{lowestW}[v] \leftarrow w$\;
    \label{alg_line_assign}
    \If{$\mathit{b\_count}(v)=\mathit{b\_count}(u)+\mathit{b\_count}(w)$}{
      mark the triplet $\{(u,p(u)),(v,p(v)),(w,p(w))\}$\;
      \label{alg_mark}
    }
  }
}
\end{algorithm}

\begin{proposition}
Algorithm \ref{algorithm:w_neq_nextM(v)} identifies all $3$-cuts $\{(u,p(u)),(v,p(v)),(w,p(w))\}$, where $u$ is a descendant of $v$, $v$ is a descendant of $w$ with $M(v)=M(w)$, and $w\neq\mathit{nextM}(v)$. Furthermore, it runs in linear time.
\end{proposition}
\begin{proof}
Let $\{(u,p(u)),(v,p(v)),(w,p(w))\}$ be a $3$-cut, where $u$ is a descendant of $v$, $v$ is a descendant of $w$ with $M(v)=M(w)$, and $w\neq\mathit{nextM}(v)$. By Lemma \ref{lemma:lowestW3}, $w$ is the greatest vertex in $M^{-1}(M(v))$ such that $w\leq\mathit{low}(u)$ and $w\leq\mathit{min}\{\mathit{low}(u')\mid\exists v' \mbox{ with } M(v')=M(v) \mbox{, } w<v'<v \mbox{ and } u'\in\mathit{stackU}[v']\}$. By Lemma \ref{lemma:lowestW2}, Algorithm \ref{algorithm:w_neq_nextM(v)} will identify $w$ during the processing of $u$ in line \ref{alg_line_assign}. As a consequence of proposition \ref{proposition:uvw}, we have $\mathit{b\_count}(v)=\mathit{b\_count}(u)+\mathit{b\_count}(w)$, and thus the triplet $\{(u,p(u)),(v,p(v)),(w,p(w))\}$ will be marked in line \ref{alg_mark}. Conversely, let $\{(u,p(u)),(v,p(v)),(w,p(w))\}$ be a triplet that gets marked by Algorithm \ref{algorithm:w_neq_nextM(v)} in line \ref{alg_mark}. Then, we have $u\in\mathit{stackU}[v]$. Furthermore, Lemma \ref{lemma:lowestW2} implies that $w$ has $M(w)=M(v)$ and $w\leq\mathit{low}(u)$. Then, since $u\in\mathit{stackU}[v]$, we have $\mathit{low}(u)<\mathit{nextM}(v)$, and therefore $w$ is a proper ancestor of $v$. Now, since $\mathit{b\_count}(v)=\mathit{b\_count}(u)+\mathit{b\_count}(w)$, Lemma \ref{lemma:greatestW} implies that $\{(u,p(u)),(v,p(v)),(w,p(w))\}$ is a $3$-cut. Thus, the correctness of Algorithm \ref{algorithm:w_neq_nextM(v)} is established.

To prove that Algorithm \ref{algorithm:w_neq_nextM(v)} runs in linear time, we will count the number of times that we access the array $\mathit{lowestW}$ during the \textbf{while} loop in line \ref{alg_line_while}. Specifically, we will show that, by the time the algorithm is terminated, the $v$ entry of $\mathit{lowestW}$, for every vertex $v$, will have been accessed at most once in line \ref{alg_line_while}. We will prove this inductively, using the inductive proposition: $\Pi(v)\equiv$ after processing $v$, we have that $\forall v'<v$ $\mathit{lowestW}[v']$ has been accessed at most once in line \ref{alg_line_while} during the course of the algorithm so far \textbf{and} $\forall v'\leq v$ we have that every $w\in T(v',\mathit{lowestW}[v'])$ has $\mathit{lowestW}[w]\geq\mathit{lowestW}[v']$. Thus, (the first part of) $\Pi(n)$ implies the linearity of Algorithm \ref{algorithm:w_neq_nextM(v)}. Now, suppose that $\Pi(v-1)$ is true for a $v\in\{1,\dotsc,n\}$ (observe that $\Pi(0)$ is trivially true). We will prove that $\Pi(v)$ is also true. Thus we have to show that: after we have processed every $u\in\mathit{stackU}[v]$, we have that $\forall v'<v$ $\mathit{lowestW}[v']$ has been accessed at most once in line \ref{alg_line_while} during the course of the algorithm so far \textbf{and} $\forall v'\leq v$ we have that every $w\in T(v',\mathit{lowestW}[v'])$ has $\mathit{lowestW}[w]\geq\mathit{lowestW}[v']$ $(1)$. Now, suppose that this was a case for a specific $\tilde{u}\in\mathit{stackU}[v]$. We will show that it is still true for the successor $u$ of $\tilde{u}$ in $\mathit{stackU}[v]$. (Of course, due to the inductive hypothesis, $(1)$ is definitely true before we have began processing the elements of $\mathit{stackU}[v]$, and therefore we may also have that $u$ is the first element of $\mathit{stackU}[v]$ in what follows.) Let $\tilde{w}$ be the value of $\mathit{lowestW}[v]$ after the assignment in line \ref{alg_line_assign}, during the processing of $u$. Thus, all vertices that we traversed during the execution of the \textbf{while} loop, during the processing of $u$, are contained in $T[v,\tilde{w}]$. Now let $v'<v$ be a vertex with the property that $\mathit{lowestW}[v']$ has been accessed once in line \ref{alg_line_while} during the course of the algorithm before the processing of $u$, and let $\tilde{v}$ be the vertex during whose processing we had to access $\mathit{lowestW}[v']$ in the \textbf{while} loop. We will show that $\mathit{lowestW}[v']$ will not be accessed in line \ref{alg_line_while} during the processing of $u$. Of course, we may assume that $v'$ is in $T[v,\tilde{w})$, for otherwise it is clear that the $v'$ entry of $\mathit{lowestW}$ will not be accessed during the execution of the \textbf{while} loop (since the traversal in \textbf{while} loop will not reach vertices lower than $\tilde{w}$, and when it reaches $\tilde{w}$ it will terminate). We note that, since the $v'$ entry of $\mathit{lowestW}$ was accessed during the execution of the \textbf{while} loop during the processing of $\tilde{v}$, we have that $\mathit{lowestW}[\tilde{v}]$ is an ancestor of $\mathit{lowestW}[v']$, and therefore a proper ancestor of $v'$. Now, if $\tilde{v}=v$, then $\mathit{lowestW}[v]$ was assigned $\mathit{lowestW}[\tilde{v}]$, in line \ref{alg_line_assign}, during the processing of a predecessor of $u$ in $\mathit{stackU}[v]$. Thus, when we begin processing $u$, $w$ is assigned a proper ancestor of $v'$ in line \ref{alg_initialize}, before entering the \textbf{while} loop, and so the $v'$ entry of $\mathit{lowestW}$ will not be accessed during the execution of the \textbf{while} loop. So let's assume that $\tilde{v}<v$. Initially, the variable $w$ is assigned $\mathit{lowestW}[v]$ in line \ref{alg_initialize}. We claim that $\mathit{lowestW}[v]$ is either a descendant of $\tilde{v}$ or a proper ancestor of $v'$. To see this, suppose, for the sake of contradiction, that $\mathit{lowestW}[v]$ is in $T(\tilde{v},v']$. Then, we have $\tilde{v}\in T(v,\mathit{lowestW}[v])$, and therefore, since $(1)$ is true for $\tilde{u}$ (the predecessor of $u$ in $\mathit{stackU}[v]$), we have that $\mathit{lowestW}[\tilde{v}]\geq\mathit{lowestW}[v]$. Since $\mathit{lowestW}[\tilde{v}]$ is a proper ancestor of $v'$, this implies that $v'>\mathit{lowestW}[v]$, contradicting the supposition $\mathit{lowestW}[v]\leq v'$. Thus, before executing the \textbf{while} loop, we have that $w$ is either a descendant of $\tilde{v}$ or a proper ancestor of $v'$.
Now suppose that the \textbf{while} loop has been executed $0$ or more times, and $w$ is assigned a descendant of $\tilde{v}$ or a proper ancestor of $v'$. We will show that if we execute the \textbf{while} loop once more, $w$ will either be assigned a descendant of $\tilde{v}$ or a proper ancestor of $v'$. Of course, if $w$ is a proper ancestor of $v'$, the same is true for $\mathit{lowestW}[w]$. Moreover, if $w=\tilde{v}$, then, as noted above, we have that $\mathit{lowestW}[w]$ is a proper ancestor of $v'$. So let's assume that $w$ is a proper descendant of $\tilde{v}$, and suppose, for the sake of contradiction, that $\mathit{lowestW}[w]$ is in $T(\tilde{v}, v']$. Then, since $\tilde{v}\in T(w,\mathit{lowest}[w])$, due to the inductive hypothesis
we have that $\mathit{lowestW}[\tilde{v}]\geq\mathit{lowestW}[w]$. Since we also have $v'>\mathit{lowestW}[\tilde{v}]$, this contradicts the supposition $\mathit{lowestW}[w]\geq v'$. Thus, if $w$ is a proper descendant of $\tilde{v}$, $\mathit{lowestW}[w]$ is either a descendant of $\tilde{v}$ or a proper ancestor of $v'$. In any case, then, during the execution of the \textbf{while} loop, $w$ will be assigned either a descendant of $\tilde{v}$ or a proper ancestor of $v'$, and thus the $v'$ entry of $\mathit{lowestW}$ will not be accessed.

It remains to show that, after the processing of $u$, for every $w\in T(v,\tilde{w})$ we have $\mathit{lowestW}[w]\geq\tilde{w}$. Due to the inductive hypothesis, this is definitely true for every $w\in T(v,\mathit{lowestW}[v])$ (where $\mathit{lowestW}[v]$ here has the value after the processing of $\tilde{u}$ and before the processing of $u$), since $\mathit{lowestW}[v]\geq\tilde{w}$, and every such $w$ has $\mathit{lowestW}[w]\geq\mathit{lowestW}[v]$. Now let's assume that $w\in T[\mathit{lowestW}[v],\tilde{w})$, and suppose, for the sake of contradiction, that $\mathit{lowestW}[w]<\tilde{w}$. Then it cannot be that case that $w=\mathit{lowestW}[v]$, since $\tilde{w}\leq\mathit{lowestW}[\mathit{lowestW}[v]]$ (for the existence of a $w\in T[\mathit{lowestW}[v],\tilde{w})$ implies that $\tilde{w}\neq \mathit{lowestW}[v]$). Now, since $\mathit{lowestW}[v]>w>\tilde{w}$, there must exist a $w'$ such that $w'\in T[\mathit{lowestW}[v],w]$, $\mathit{lowestW}[w']<w$ and $\mathit{lowestW}[w']\geq\tilde{w}$. Since $\mathit{lowestW}[w]<\tilde{w}$, we cannot $w'=w$. Then, $w\in T(w',\mathit{lowestW}[w'])$, and thus, due to the inductive hypothesis, we have $\mathit{lowestW}[w]\geq\mathit{lowestW}[w']$. Since $\mathit{lowestW}[w']\geq\tilde{w}$, this implies that $\mathit{lowestW}[w]\geq\tilde{w}$, contradicting the supposition $\mathit{lowestW}[w]<\tilde{w}$. Thus, every $w\in T(v,\tilde{w})$ has $\mathit{lowestW}[w]\geq\tilde{w}$. The proof that $(1)$ is true for $u$ is complete. Due to the generality of $u\in\mathit{stackU}[v]$, this implies that $\Pi(v)$ is true. This shows, by induction, that $\Pi(n)$ is true, and the linearity of Algorithm \ref{algorithm:w_neq_nextM(v)} is thus established.
\end{proof}

\section{Computing the $4$-edge-connected components in linear time}
\label{sec:4ecc}

Now we consider how to compute the $4$-edge-connected components of an undirected graph $G$ in linear time. First, we reduce this problem to the computation of the $4$-edge-connected components of a collection of auxiliary $3$-edge-connected graphs.

\subsection{Reduction to the $3$-edge-connected case}
\label{sec:reduction}

Given a (general) undirected graph $G$, we execute the following steps:
\begin{itemize}
\item Compute the connected components of $G$.
\item For each connected component, we compute the $2$-edge-connected components which are subgraphs of $G$.
\item For each $2$-edge-connected component, we compute its $3$-edge-connected components $C_1,\ldots,C_{\ell}$.
\item For each $3$-edge-connected component $C_i$, we compute a $3$-edge-connected auxiliary graph $H_i$, such that for any two vertices $x$ and $y$, we have $x \stackrel[]{G}{\equiv}_4 y$ if and only if $x$ and $y$ are both in the same auxiliary graph $H_i$ and $x \stackrel[]{H_i}{\equiv}_4 y$.
\item Finally, we compute the $4$-edge-connected components of each $H_i$.
\end{itemize}

Steps 1--3 take overall linear time~\cite{dfs:t,Tsin:3CC}. We describe step 5 in the next section, so it remains to give the details of step 4.
Let $H$ be a $2$-edge-connected component (subgraph) of $G$.
We can construct a compact representation of the $2$-cuts of $H$, which allows us to compute its $3$-edge-connected components $C_1,\dotsc,C_{\ell}$ in linear time \cite{TwinlessSAP,Tsin:3CC}. Now, since the collection $\{C_1,\dotsc,C_{\ell}\}$ constitutes a partition of the vertex set of $H$, we can form the quotient graph $Q$ of $H$ by shrinking each $C_i$ into a single node. Graph $Q$ has the structure of a tree of cycles~\cite{3cuts:Dinitz}; in other words, $Q$ is connected and every edge of $Q$ belongs to a unique cycle. Let $(C_i,C_j)$ and $(C_i,C_k)$ be two edges of $Q$ which belong to the same cycle. Then $(C_i,C_j)$ and $(C_i,C_k)$ correspond to two edges $(x,y)$ and $(x',y')$ of $G$, with $x,x'\in C_i$. If $x\neq x'$, we add a virtual edge $(x,x')$ to $G[C_i]$. (The idea is to attach $(x,x')$ to $G[C_i]$ as a substitute for the cycle of $Q$ which contains $(C_i,C_j)$ and $(C_i,C_k)$.)
Now let $\bar{C_i}$ be the graph $G[C_i]$ plus all those virtual edges. Then $\bar{C_i}$ is $3$-edge-connected and its $4$-edge-connected components are precisely those of $G$ that are contained in $C_i$~\cite{3cuts:Dinitz}. Thus we can compute the $4$-edge-connected components of $G$ by computing the $4$-edge-connected components of the graphs $\bar{C_1},\dotsc,\bar{C_\ell}$ (which can easily be constructed in total linear time). Since every $\bar{C_i}$ is $3$-edge-connected, we can apply Algorithm \ref{algorithm:4-components} of the following section to compute its $4$-edge-connected components in linear time. Finally, we define the multiplicity $m(e)$ of an edge $e\in\bar{C_i}$ as follows: if $e$ is virtual, $m(e)$ is the number of edges of the cycle of $Q$ which corresponds to $e$; otherwise, $m(e)$ is $1$. Then, the number of minimal $3$-cuts of $H$ is given by the sum of all $m(e_1)\cdot m(e_2)\cdot m(e_3)$, for every $3$-cut $\{e_1,e_2,e_3\}$ of $\bar{C_i}$, for every $i\in\{1,\dotsc, l\}$ \cite{3cuts:Dinitz}. Since the $3$-cuts of every $\bar{C_i}$ can be computed in linear time, the minimal $3$-cuts of $H$ can also be computed within the same time bound.


\subsection{Computing the $4$-edge-connected components of a $3$-edge-connected graph}
\label{sec:4ecc-alg}

Now we describe how to compute the $4$-edge-connected components of a $3$-edge-connected graph $G$ in linear time.
Let $r$ be a distinguished vertex of $G$, and let $C$ be a minimum cut of $G$. By removing $C$ from $G$, $G$ becomes disconnected into two connected components. We let $V_C$ denote the connected component of $G\setminus{C}$ that does not contain $r$, and we refer to the number of vertices of $V_C$ as the \emph{$r$-size} of the cut $C$. (Of course, these notions are relative to $r$.)

Let $G=(V,E)$ be a $3$-edge-connected graph, and let $\mathcal{C}$ be the collection of the $3$-cuts of $G$. If the collection $\mathcal{C}$ is empty, then $G$ is $4$-edge-connected, and $V$ is the only $4$-edge-connected component of $G$. Otherwise, let $C\in\mathcal{C}$ be a $3$-cut of $G$. By removing $C$ from $G$, $G$ is separated into two connected components, and every $4$-edge-connected component of $G$ lies entirely within a connected component of $G\setminus{C}$. This observation suggests a recursive algorithm for computing the $4$-edge-connected components of $G$, by successively splitting $G$ into smaller graphs according to its $3$-cuts. Thus, we start with a $3$-cut $C$ of $G$, and we perform the splitting operation shown in Figure \ref{figure:splitting}. Then we take another $3$-cut $C'$ of $G$ and we perform the same splitting operation on the part which contains (the corresponding $3$-cut of) $C'$. We repeat this process until we have considered every $3$-cut of $G$. When no more splits are possible, the connected components of the final split graph correspond (by ignoring the newly introduced vertices) to the $4$-edge-connected components of $G$.

To implement this procedure in linear time, we must take care of two things. First, whenever we consider a $3$-cut $C$ of $G$, we have to be able to know which ends of the edges of $C$ belong to the same connected component of $G\setminus{C}$. And second, since an edge $e$ of a $3$-cut of the original graph may correspond to two virtual edges of the split graph, we have to be able to know which is the virtual edge that corresponds to $e$. We tackle both these problems by locating the $3$-cuts of $G$ on a DFS-tree $T$ of $G$ rooted at $r$, and by processing them in increasing order with respect to their $r$-size. By locating a $3$-cut $C\in\mathcal{C}$ on $T$ we can answer in $O(1)$ time which ends of the edges of $C$ belong to the same connected component of $G\setminus{C}$. And then, by processing the $3$-cuts of $G$ in increasing order with respect to their size, we ensure that (the $3$-cut that corresponds to) a $3$-cut $C\in\mathcal{C}$ that we process lies in the split part of $G$ that contains $r$.

\begin{figure}
\begin{center}
\centerline{\includegraphics[trim={0cm 23cm 0cm 0cm}, scale=1,clip, width=\textwidth]{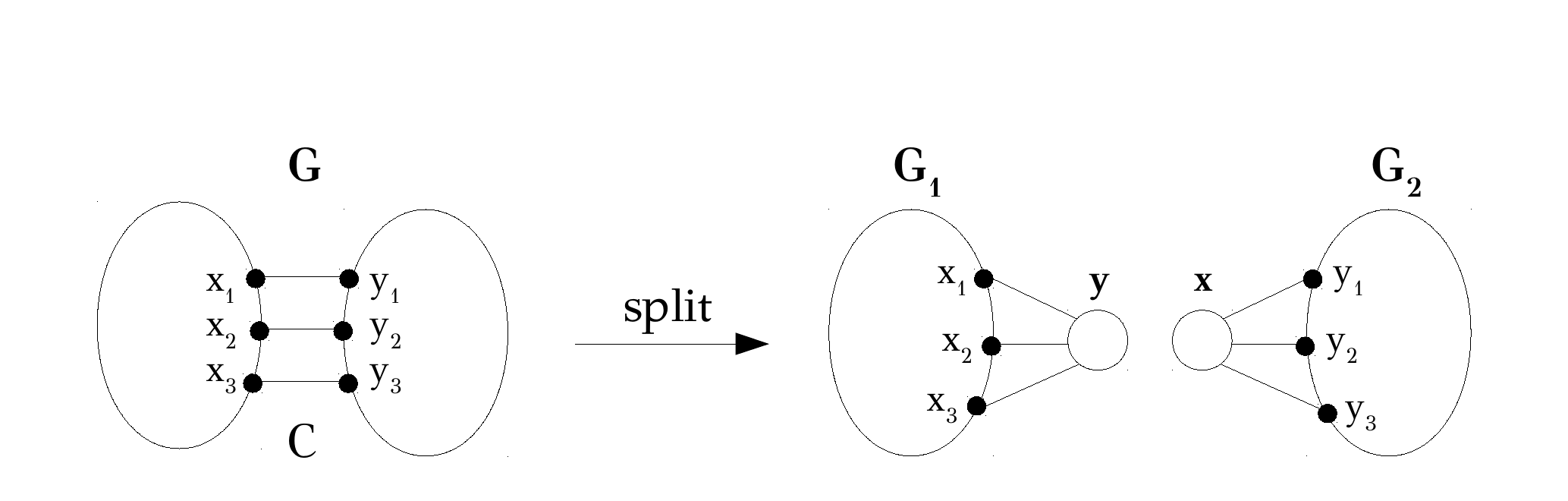}}
\caption{$C=\{(x_1,y_1),(x_2,y_2),(x_3,y_3)\}$ is a $3$-cut of $G$, with $\{x_1,x_2,x_3\}$ and $\{y_1,y_2,y_3\}$ lying in different connected components of $G\setminus{C}$. The split operation of $G$ at $C$ consists of the removal the edges of $C$ from $G$, and the introduction of two new nodes $x,y$, and six virtual edges $(x_1,y),(x_2,y),(x_3,y),(x,y_1),(x,y_2),(x,y_3)$. Now, the split graph is made of two connected components, $G_1$ and $G_2$. Every $3$-cut $C'\neq C$ of $G$ (or more precisely: a $3$-cut that corresponds to $C'$) lies entirely within $G_1$ or $G_2$. Conversely, every $3$-cut of either $G_1$ or $G_2$ corresponds to a $3$-cut of $G$. Thus, every $4$-edge-connected component of $G$ lies entirely within $G_1$ or $G_2$.}
\label{figure:splitting}
\end{center}
\end{figure}

Now, due to the analysis of the preceding sections, we can distinguish the following types of $3$-cuts on a DFS-tree $T$ (see also Figure \ref{figure:all_cases}):
\begin{itemize}
\item (I) $\{(v,p(v)),(x_1,y_1),(x_2,y_2)\}$, where $(x_1,y_1)$ and $(x_2,y_2)$ are back-edges.
\item (IIa) $\{(u,p(u)),(v,p(v)),(x,y)\}$, where $u$ is a descendant of $v$ and $(x,y)\in B(v)$.
\item (IIb) $\{(u,p(u)),(v,p(v)),(x,y)\}$, where $u$ is a descendant of $v$ and $(x,y)\in B(u)$.
\item (III) $\{(u,p(u)),(v,p(v)),(w,p(w))\}$, where $w$ is an ancestor of both $u$ and $v$, but $u,v$ are not related as ancestor and descendant.
\item (IV) $\{(u,p(u)),(v,p(v)),(w,p(w))\}$, where $u$ is a descendant of $v$ and $v$ is a descendant of $w$.
\end{itemize}
Let $r$ be the root of $T$. Then, for every $3$-cut $C\in\mathcal{C}$, $V_C$ is either $T(v)$, or $T(v)\setminus{T(u)}$, or $T(w)\setminus(T(u)\cup T(v))$, or $T(u)\cup(T(w)\setminus T(v))$, depending on whether $C$ is of type (I), (II), (III), or (IV), respectively. Thus we can immediately calculate the size of $C$ and the ends of its edges that lie in $V_C$. In particular, the size of $C$ is either $\mathit{ND}(v)$, or $\mathit{ND}(v)-\mathit{ND}(u)$, or $\mathit{ND}(w)-\mathit{ND}(u)-\mathit{ND}(v)$, or $\mathit{ND}(u)+\mathit{ND}(w)-\mathit{ND}(v)$, depending on whether it is of type (I), (II), (III), or (IV), respectively; $V_C$ contains either $\{v,x_1,x_2\}$, or $\{p(u),v,x\}$, or $\{p(u),v,y\}$, or $\{p(u),p(v),w\}$, or $\{u,p(v),w\}$, depending on whether $C$ is of type (I), (IIa), (IIb), (III), or (IV), respectively.

Algorithm \ref{algorithm:4-components} shows how we can compute the $4$-edge-connected components of $G$ in linear time, by repeatedly splitting $G$ into smaller graphs according to its $3$-cuts. When we process a $3$-cut $C$ of $G$, we have to find the edges of the split graph that correspond to those of $C$, in order to delete them and replace them with (new) virtual edges. That is why we use the symbol $v'$, for a vertex $v\in V$, to denote a vertex that corresponds to $v$ in the split graph. (Initially, we set $v' \leftarrow v$.) Now, if $(x,y)$ is an edge of $C$ with $x\in V_C$, the edge of the split graph corresponding to $(x,y)$ is $(x',y')$. Then we add two new vertices $v_C$ and $\tilde{v_C}$ to $G$, and the virtual edges $(x',\tilde{v_C})$ and $(v_C,y')$. Finally, we let $x$ correspond to $v_C$, and so we set $x' \leftarrow v_C$. This is sufficient, since we process the $3$-cuts of $G$ in increasing order with respect to their size, and so the next time we meet the edge $(x,y)$ in a $3$-cut, we can be certain that it corresponds to $(v_C,y')$.

\begin{algorithm}[!h]
\caption{\textsf{Compute the $4$-edge-connected components of a $3$-edge-connected graph $G=(V,E)$}}
\label{algorithm:4-components}
\LinesNumbered
\DontPrintSemicolon
Find the collection $\mathcal{C}$ of the $3$-cuts of $G$\;
Locate and classify the $3$-cuts of $G$ on a DFS-tree of $G$ rooted at $r$\;
For every $C\in\mathcal{C}$, calculate $\mathit{size}(C)$ (relative to $r$)\;
Sort $\mathcal{C}$ in increasing order w.r.t. the $\mathit{size}$ of its elements\;
\lForEach{$v\in V$}{Set $v' \leftarrow v$}
\ForEach{$C=\{(x_1,y_1),(x_2,y_2),(x_3,y_3)\}\in\mathcal{C}$}{
  \label{line:alg4-comp}
  Find the ends of the edges of $C$ that lie in $V_C$ \tcp{\textit{Let those ends be $x_1$,$x_2$ and $x_3$}}
  \label{line:alg4-comp-1}
  Remove the edges $(x_1',y_1')$,$(x_2',y_2')$,$(x_3',y_3')$ from $G$\;
  Introduce two new vertices $v_C$ and $\tilde{v_C}$ to $G$\;
  Add the edges $(x_1',\tilde{v_C})$,$(x_2',\tilde{v_C})$,$(x_3',\tilde{v_C})$,$(v_C,y_1')$,$(v_C,y_2')$,$(v_C,y_3')$ to $G$\;
  \label{line:alg4-comp-last}
  Set $x_1' \leftarrow v_C$, $x_2' \leftarrow v_C$, $x_3' \leftarrow v_C$\;
}
Output the connected components of $G$, ignoring the newly introduced vertices\;
\end{algorithm}

\clearpage

\bibliography{ltg}

\end{document}